\documentclass[acmsmall, screen, nonacm]{acmart}

\usepackage{stmaryrd}
\usepackage{mathtools}
\usepackage{thmtools}
\usepackage{proof-dashed}
\usepackage{dashrule}
\usepackage[inline]{enumitem}
\usepackage[capitalize]{cleveref}

\newlength{\figinferskip}
\newcommand{\jrule}[1]{\text{\scshape \MakeLowercase{#1}}}
\newcommand{\srule}[1]{\jrule{$#1$S}}
\newcommand{\arule}[1]{\jrule{$#1$F}}
\newcommand{\brule}[1]{\jrule{$#1$B}}
\newcommand{\prule}[1]{\jrule{$#1$P}}

\newcommand{\sig}{\Sigma}
\NewDocumentCommand{\sige}{s}{\IfBooleanT{#1}{(} \cdot \IfBooleanT{#1}{)}}

\newcommand{\alphas}{\vec{\alpha}}
\NewDocumentCommand{\betas}{s}{\IfBooleanTF{#1}{\smash{\betas}\vphantom{beta}}{\vec{\beta}}}

\newcommand{\defd}{\triangleq}

\newcommand{\tensor}{\mathbin{\times}}
\newcommand{\one}{\mathord{\mathbf{1}}}

\newcommand{\imp}{\mathbin{\rightarrow}}

\NewDocumentCommand{\ctxe}{s}{\IfBooleanTF{#1}{(\ctxe)}{\cdot}}

\ExplSyntaxOn

\NewDocumentCommand{\plus}{s}{\IfBooleanTF{#1}{\multiplus}{\mathbin{+}}}
\NewDocumentCommand{\multiplus}{o m}{\mathopen{\plus}\{\tl_if_empty:nTF{#2}{\,}{#2}\}\IfValueT{#1}{\sb{#1}}}
\NewDocumentCommand{\with}{s}{\IfBooleanTF{#1}{\multiwith}{\mathbin{\binampersand}}}
\NewDocumentCommand{\multiwith}{o m}{\mathopen{\with}\{\tl_if_empty:nTF{#2}{\,}{#2}\}\IfValueT{#1}{\sb{#1}}}

\ExplSyntaxOff

\NewDocumentCommand{\dsub}{o}{\leqslant} 
\newcommand{\ndsub}{\nleqslant}
\newcommand{\asub}{\preccurlyeq}

\ExplSyntaxOn

\cs_new:Npn \suspendsubs:nn #1 #2 { #1 \tl_if_empty:nF { #2 } { \langle#2\rangle } }

\NewDocumentCommand{\dsubtype}{o > { \SplitArgument { 1 } { < } } m }{
  \exp_last_unbraced:No \suspendsubs:nn { \use_i:nn #2 }
  \IfValueTF{#1}{\dsub[#1]}{\dsub}
  \exp_last_unbraced:No \suspendsubs:nn { \use_ii:nn #2 }
}

\NewDocumentCommand{\asubtype}{o > { \SplitArgument { 1 } { < } } m }{
  \use_i:nn #2 \asub \use_ii:nn #2 \IfValueT{#1}{\vof #1}
}

\ExplSyntaxOff

\newcommand{\entails}{\Rightarrow}
\newcommand{\nentails}{\Arrownot\entails}

\newcommand{\cov}{\mathord{+}}
\newcommand{\ctv}{\mathord{-}}

\newcommand{\vof}{\mathrel{\text{\ttfamily\#}}}

\NewDocumentCommand{\sube}{s}{\IfBooleanT{#1}{(} \cdot \IfBooleanT{#1}{)}}
\NewDocumentCommand{\subs}{s}{\IfBooleanTF{#1}{\Phi}{\Theta}}
\NewDocumentCommand{\subse}{s}{\IfBooleanT{#1}{(} \cdot \IfBooleanT{#1}{)}}

\newcommand{\set}[1]{\lbrace#1\rbrace}

\newcommand{\applysubs}[2]{\tosub{#1}(#2)}
\newcommand{\tosub}[1]{#1}

\NewDocumentCommand{\bpa}{ s m s m }{ \IfBooleanTF{#1}{ (#2) }{#2} \fatsemi \IfBooleanTF{#3}{ (#4) }{#4} }
\NewDocumentCommand{\bpasum}{ s }{ \IfBooleanTF{#1}{ \bpasumn }{ + } }
\NewDocumentCommand{\bpasumn}{ o m }{ \sum \IfValueT{#1}{ \sb{#1} } #2 }
\NewDocumentCommand{\bpaseq}{ m m }{ #1 \cdot #2 }
\NewDocumentCommand{\bpaemp}{ }{ \epsilon }

\NewDocumentCommand{\simu}{ }{ \lesssim }
\NewDocumentCommand{\umis}{ }{ \gtrsim }
\NewDocumentCommand{\trans}{ s o }{ \IfValueTF{#2}{ \mathrel{\IfBooleanT{#1}{\smash} {\overset{\raisebox{-0.5ex}{$\scriptstyle #2$}}{\trans}}} }{ \longrightarrow } }

\newcommand{\name}[1]{\mathsf{#1}}

\newcommand{\onen}{\name{one}}
\newcommand{\nat}{\name{nat}}
\newcommand{\zero}{\name{z}}
\newcommand{\suc}{\name{s}}
\newcommand{\even}{\name{even}}
\newcommand{\odd}{\name{odd}}

\newcommand{\listn}{\name{list}}

\newcommand{\nelistn}{\name{nelist}}
\newcommand{\elistn}{\name{elist}}
\newcommand{\nil}{\name{nil}}
\newcommand{\cons}{\name{cons}}
\newcommand{\olistn}{\name{olist}}

\newcommand{\tree}{\name{tree}}
\newcommand{\leaf}{\name{leaf}}
\newcommand{\node}{\name{node}}
\newcommand{\otree}{\name{otree}}

\newcommand{\lspine}{\name{spine}}

\newcommand{\fst}{\name{fst}}
\newcommand{\snd}{\name{snd}}

\newcommand{\stree}{\name{stree}}

\newcommand{\trie}{\name{treefn}}

\newcommand{\ptrie}{\name{trie}}

\newcommand{\option}{\name{option}}

\newcommand{\dyck}{\name{d}}
\newcommand{\edyck}{\name{e}}
\newcommand{\rightn}{\name{R}}
\newcommand{\leftn}{\name{L}}

\newcommand{\stack}{\name{stack}}
\newcommand{\pushn}{\name{push}}
\newcommand{\popn}{\name{pop}}
\newcommand{\some}{\name{some}}
\newcommand{\none}{\name{none}}

\newcommand{\rstack}{\name{rstack}}

\newcommand{\fold}{\name{fold}}
\newcommand{\size}{\name{size}}
\usepackage{breakurl}

\AtBeginDocument{%
  }

\setcopyright{acmcopyright}
\copyrightyear{2024}
\acmYear{2024}
\acmDOI{XXXXXXX.XXXXXXX}

\acmConference[POPL '24]{Make sure to enter the correct
  conference title from your rights confirmation emai}{January 17--19,
  2024}{London, UK}
\acmPrice{15.00}
\acmISBN{978-X-XXXX-XXXX-X/XX/XX}

\acmJournal{PACMPL}
\acmVolume{00}
\acmNumber{0}
\acmArticle{0}
\acmMonth{1}

\begin{document}

\title{Parametric Subtyping for Structural Parametric Polymorphism}

\author{Henry DeYoung}
\orcid{0000-0003-1649-9953}
\affiliation{%
  \institution{Carnegie Mellon University}
  \department{Computer Science Department}
  \city{Pittsburgh}
  \country{USA}
}
\email{hdeyoung@cs.cmu.edu}

\author{Andreia Mordido}
\orcid{0000-0002-1547-0692}
\affiliation{%
  \institution{Universidade de Lisboa}
  \department{Faculdade de Ci\^{e}ncias, LASIGE}
  \city{Lisbon}
  \country{Portugal}
}
\email{afmordido@ciencias.ulisboa.pt}

\author{Frank Pfenning}
\orcid{0000-0002-8279-5817}
\affiliation{%
  \institution{Carnegie Mellon University}
  \department{Computer Science Department}
  \city{Pittsburgh}
  \country{USA}
}
\email{fp@cs.cmu.edu}

\author{Ankush Das}
\orcid{0000-0003-2459-1258}
\affiliation{%
  \institution{Amazon}
  \city{Santa Clara}
  \country{USA}
}
\email{daankus@amazon.com}

\renewcommand{\shortauthors}{DeYoung et al.}

\begin{abstract}
  We study the interaction of structural subtyping with parametric polymorphism and recursively defined type constructors.
  Although structural subtyping is undecidable in this setting, we describe a notion of parametricity for type constructors and then exploit it to define \emph{parametric subtyping}, a conceptually simple, decidable, and expressive fragment of structural subtyping that strictly generalizes \emph{rigid subtyping}.
  We present and prove correct an effective saturation-based decision procedure for parametric subtyping, demonstrating its applicability using a variety of examples.
  We also provide an implementation of this decision procedure online.
\end{abstract}

\ccsdesc[500]{Theory of computation~Type theory}
\ccsdesc[500]{Theory of computation~Type structures}
\ccsdesc[500]{Theory of computation~Functional constructs}
\ccsdesc[500]{Software and its engineering~Polymorphism}
\ccsdesc[500]{Software and its engineering~Recursion}

\keywords{structural subtyping, parametric polymorphism, type constructors, saturation-based algorithms}


\maketitle

\section{Introduction}\label{sec:introduction}

Recursive types, parametric polymorphism (also called generics), and subtyping are all essential features for modern programming languages across numerous paradigms.
Recursive types describe unbounded data structures;
parametric polymorphism provides type-level modularity by allowing programmers to use instantiations of $\listn[\alpha]$ rather than separate monomorphic types for integer and boolean lists, for example;
and subtyping provides flexibility in the ways that objects and terms can be used, enabling code reuse.
Structural subtyping, in particular, is especially flexible and expressive and, in principle, relatively lightweight for programmers to incorporate.

This combination of features is present to varying degrees in many of today's widely used languages, such as Go, Rust, TypeScript and Java, but the combination is difficult to manage.
For example, subtyping for generics in Java is known to be undecidable~\cite{Grigore:ACM17},
so various restrictions on types' structure have been proposed, such as material-shape separation~\cite{Greenman+:ACM14,Mackay+:APLAS20}, and the prohibition of contravariance, unbounded expansion of types, or multiple instantiation inheritance~\cite{Kennedy+Pierce:06}, to name a few.

We contend that these restrictions are often either too limiting or too unintuitive for programmers to readily reason about.
A reconstruction of the interaction between recursive types, parametric polymorphism, and structural subtyping from first principles is needed, accompanied by a clear, relatively simple declarative characterization of subtyping.
However, to the best of our knowledge, no such work has been undertaken thus far.
This paper fills that gap.

As a first step, we prove that structural subtyping is undecidable in the presence of recursive types and parametric polymorphism.%
\footnote{This proof revises a prior, unpublished proof by the present authors~\cite{Das+:arXiv21}.
Independently, a related result was proven by \citet{Padovani:TOPLAS19} for subtyping of context-free session types~\cite{Thiemann+Vasconcelos:ICFP16}.
Our proof occurs in the setting of recursively defined type constructors (which has been shown to be more general than context-free session types~\cite{Gay+:FoSSaCS22}) and identifies two other minimal undecidable fragments.}
Given this undecidability, our goal is to identify an expressive, practical fragment of structural subtyping that has three properties:
\begin{enumerate}
\item The fragment should have a relatively simple \emph{declarative characterization}, so that the valid subtypings are readily predictable by the programmer.
\item The fragment should be \emph{decidable}, with an effective algorithm that performs well on the kinds of subtyping problems that arise in practice.
\item
  The fragment should \emph{strictly generalize ``rigid subtyping''}, a form of subtyping in which subtypings exist only between types with the same outermost type constructor, such as $\dsubtype{{\listn[\name{int}]}{} < {\listn[\name{real}]}{}}$ but never $\dsubtype{{\listn[\name{int}]}{} < {\listn'[\name{real}]}{}}$.
\end{enumerate}

It is not immediately clear that such a fragment of structural subtyping should even exist, as seemingly innocent variations of the problem are already undecidable or impractical.
\Citet{Solomon:POPL78} showed that structural \emph{equality} for parametric polymorphism can be reduced to equivalence of deterministic pushdown automata, but it took more than 20 additional years to establish decidability~\cite{Senizergues:TCS01,Stirling:TCS01,Stirling:FSTTCS01}, albeit by an intractable algorithm.
As another example, even \emph{without} recursive types, subtyping for implicit, Curry-style polymorphism~\cite{Tiuryn+Urzyczyn:IC02,Wells:BU95} and bounded quantification~\cite{Pierce:IC94} are both undecidable.

Nevertheless, in this paper, we are able to achieve our goal:
we propose a notion of parametricity for type constructors that forms the basis of a suitable fragment of structural subtyping, a fragment that we call \emph{parametric subtyping}.
Parametric constructors will map subtyping-related arguments to subtyping-related results, echoing \citeauthor{Reynolds:IFIP83}'s~\shortcite{Reynolds:IFIP83} characterization of parametric functions as those that map related arguments to related results.\footnote{This analogy will be discussed further in \cref{sec:parametric}.}
Moreover, by exploiting parametricity, we avoid unintuitive restrictions on types' structure
and can support even non-regular types~\cite{Bird+Meertens:MPC98,Mycroft:ISP84}.

Because of its fundamental nature, our notion of parametric subtyping and associated decision procedure could be applied to a wide variety of languages:
object-oriented languages; lazy and eager functional languages; imperative languages; mixed inductive/coinductive languages, such as call-by-push-value~\cite{Levy:PhD01}; session-typed languages~\cite{Honda+:ESOP98,Gay+Hole:Acta05,Caires+Pfenning:CONCUR10,Silva+:CONCUR23}; and so on.\footnote{Of course, when applied to a given language, there will be additional language-specific considerations, \textit{e.g.}, interaction with intersection types in TypeScript or type classes in Haskell.
We do not claim to address such considerations here.}

We want to emphasize this broad applicability by keeping this paper's technical framework as general as possible.
This leads us to make several concrete choices in this paper's presentation.
\begin{itemize}
\item
  We do not consider subtypings such as $\dsubtype{{\mathbf{0}}{} < \one{}}$ and $\dsubtype{\one{} < {\mathbf{0} \imp \one}{}}$ that arise when some types are uninhabited~\cite{Ligatti+:TOPLAS17}.
  This is because we choose to interpret all types coinductively, making them all inhabited, even $\mathbf{0}$.
  Nevertheless, the parametric subtyping rules in this paper are sound in languages where some types are interpreted inductively.
  Had we instead insisted on an inductive treatment of some types, parametric subtypings such as the above would be unsound for lazy functional and session-typed languages, undercutting broader applicability.

\item
  Neither do we consider subtypings that rely on implicit, Curry-style polymorphism, such as $\dsubtype{{(\forall x.\, \listn[x])}{} < {\listn[\name{int}]}{}}$, or bounded quantification, such as $\forall (\dsubtype{x{} < {\with*{\name{x}\colon \name{real}}}{}}).\, x \tensor \name{real} \imp x$, because subtyping is already undecidable in those settings, even \emph{without} recursive types~\cite{Tiuryn+Urzyczyn:IC02,Wells:BU95,Pierce:IC94}.
\end{itemize}
In summary, our primary aim is to examine the interaction of explicit, Church-style polymorphism, recursive type constructors, and the fundamental core of structural subtyping.

\subsection{Overview of parametric subtyping}

To provide the reader with some intuition for our notion of parametric subtyping, we will now sketch, at a high level, how parametric subtyping satisfies the three desired properties.
\begin{enumerate}
\item
  \emph{Relatively simple declarative characterization.}
  A pair of type constructors, $t[\alphas]$ and $u[\betas]$, will be considered parametric if
  the subtyping problem $\dsubtype{{t[\alphas]}{} < {u[\betas*]}{}}$ can be reduced to (finitely many) subtyping problems among the arguments $\alphas$ and $\betas*$ alone.

  \quad As an example, consider an interface for stack objects (or, from a functional perspective, a record type for stacks), parameterized by a type~$\alpha$ of stack elements:
  \begin{alignat*}{2}
    \stack[\alpha] &\defd \mathrlap{\with*{\pushn\colon \alpha \imp \stack[\alpha] ,\, \popn\colon \option[\alpha \tensor \stack[\alpha]]}}
\intertext{where $\option[\beta] \defd \plus*{\none\colon \one ,\, \some\colon \beta}$.
  Programmers sometimes want to ensure that a stack be used according to a particular protocol.
  For example, when implementing a queue using a pair of stacks (as sometimes done in functional languages), the protocol in which all pushes must occur before any pops can be expressed by the types}
  \name{qstack_1}[\beta] &\defd \with*{\pushn\colon \beta \imp \name{qstack_1}[\beta] ,\, &&\popn\colon \option[\beta \tensor \name{qstack_2}[\beta]]}
  \\
  \name{qstack_2}[\beta] &\defd \with*{&&\popn\colon \option[\beta \tensor \name{qstack_2}[\beta]]}
  \mathrlap{\,.}
\end{alignat*}
  By virtue of having definitions with compatible structures, $\stack[-]$ is a subtype of itself, $\name{qstack_1}[-]$, and $\name{qstack_2}[-]$ according to the admissible subtyping rules
  \begin{equation*}
    \qquad\quad
    \infer{\dsubtype{{\stack[\alpha]}{} < {\stack[\beta]}{}}}{
      \dsubtype{\alpha{} < {\beta}{}} & \dsubtype{{\beta}{} < \alpha{}}}
    \:,\quad
    \infer{\dsubtype{{\stack[\alpha]}{} < {\name{qstack_1}[\beta]}{}}}{
      \dsubtype{\alpha{} < \beta{}} & \dsubtype{\beta{} < \alpha{}}}
    \:,\enspace\,\text{\emph{and}}\quad
    \infer{\dsubtype{{\stack[\alpha]}{} < {\name{qstack_2}[\beta]}{}}}{
      \dsubtype{\alpha{} < {\beta}{}}}
    \mathrlap{\,.}
  \end{equation*}
  These rules are admissible in the sense that there exist corresponding infinite derivations that use only the standard structural subtyping rules.
  More importantly, these are valid \emph{parametric} subtyping rules because the premises involve only arguments, $\alpha$ and $\beta$.
  A rule such as ``$\dsubtype{{\name{intlist}}{} < {\listn[\beta]}{}}$ if $\dsubtype{{\name{int}}{} < \beta{}}$'', where $\name{intlist} \defd \plus*{\nil\colon \one ,\, \cons\colon \name{int} \tensor \name{intlist}}$, would \emph{not} be parametric because its premise involves a type, $\name{int}$, that is not an argument.

  \quad Given such rules, parametric subtyping is then conceptually rather straightforward:
  A subtyping between types holds because it has a finite derivation from the admissible \emph{parametric} subtyping rules.
  For example, $\dsubtype{{\stack[\stack[\name{int}]]}{} < {\name{qstack_2}[\name{qstack_1}[\name{int}]]}{}}$ is a valid parametric subtyping because we can derive
  \begin{equation*}
    \infer{\dsubtype{{\stack[\stack[\name{int}]]}{} < {\name{qstack_2}[\name{qstack_1}[\name{int}]]}{}}}{
      \infer{\dsubtype{{\stack[\name{int}]}{} < {\name{qstack_1}[\name{int}]}{}}}{
        \infer{\dsubtype{{\name{int}}{} < {\name{int}}{}}}{} &
        \infer{\dsubtype{{\name{int}}{} < {\name{int}}{}}}{}}}
  \end{equation*}
  from the above admissible parametric subtyping rules.
  On the other hand, $\dsubtype{{\name{intlist}}{} < {\listn[\name{int}]}{}}$ is \emph{not} a valid parametric subtyping because there is no admissible parametric subtyping rule for $\dsubtype{{\name{intlist}}{} < {\listn[\beta]}{}}$.%
  \footnote{Our interest in subtyping is primarily motivated by the desire to express more program properties (\emph{e.g.}, $\dsubtype{{\nelistn[\name{int}]}{} < {\listn[\name{int}]}{}}$ for a function that always returns a nonempty list), rather than a desire to type as many programs as possible.
  From this point of view, we find it acceptable that programmers may sometimes be forced to use $\listn[\name{int}]$ in place of $\name{intlist}$.}

  \quad In \cref{sec:parametric}, we present an equivalent characterization of parametric subtyping that is more amenable to metatheoretic proofs.
  A series of examples in \cref{sec:examples} demonstrates that the valid parametric subtypings are readily predictable by the programmer and expressive enough for many subtypings desired in practice.\\[-0.75\baselineskip]

\item
  \emph{Decidable.}
  In \cref{sec:decide-poly}, we prove that parametric subtyping is decidable by giving a saturation algorithm that is sound and complete with respect to the declarative characterization of parametric subtyping (\cref{thm:sound-poly,thm:complete-poly}).
  The algorithm infers, for each pair of type constructors, the most general parametric subtyping rule, if one exists.
  Moreover, when no such parametric rule exists, the algorithm determines whether the cause is a fundamental violation of structural subtyping or merely a violation of parametricity.
  After inferring all such admissible rules, a given subtyping problem can be decided by backward proof construction of a finite derivation
 using the inferred rules.

  \quad
  We have implemented this decision procedure, and it is available online in a virtual machine image~\cite{DeYoung+:Zenodo23}; the source files are also available in an online repository~\cite{DeYoung+:Bitbucket23}.
\\[-0.75\baselineskip]

\item
  \emph{Generalizes rigid subtyping.}
  Rigid subtyping is characterized by those parametric rules that relate identical type constructors, such as the above rule for $\dsubtype{{\stack[\alpha]}{} < {\stack[\beta]}{}}$.
  Our notion of parametric subtyping is indeed strictly more general than rigid subtyping, in that it also admits those parametric rules that relate distinct type constructors, such as the above rule for $\dsubtype{{\stack[\alpha]}{} < {\name{qstack_1}[\beta]}{}}$.

  \quad
  This is a simple but important property.
  If parametric subtyping somehow did not generalize rigid subtyping, that failure of type constructor ``reflexivity'' would make parametric subtyping very unintuitive and would likely be indicative of other serious problems.
  (Comparison to \emph{nominal subtyping}~\cite[\emph{e.g.},][]{Kennedy+Pierce:06} as common in object-oriented languages is left to future work.)
\end{enumerate}

The most closely related work is that on refinement types, specifically datasort refinements~\cite{Freeman+Pfenning:PLDI91}; there are significant differences, however.
First, whereas the refinement system refines a nominal type into a collection of structural sorts, we use a single-layer, fully structural system.
Second, \citet[Sec.~7.4]{Davies:PhD05} defines an algorithm for subsorting parameterized sort constructors that respects parametricity, but requires explicit declarations for the constructors' variance and handles only very limited cases of nested sorts.
(On the other hand, he deals with general pattern matching, module boundaries, and intersections, which are beyond the scope of the present work.)
Third, \citet{Skalka:MSc97} gives an algorithm to decide the emptiness of refinement types, but does not give a subtyping algorithm and handles only regular type constructors.
Last, whereas the nominal core of refinement types means that a defined type cannot later be widened into a supertype, our fully structural system has the advantage of naturally permitting widening.

In summary, the contributions of this paper are:
to identify several minimal fragments for which structural subtyping is undecidable~(\cref{sec:structural:undecidability});
to give a simple, declarative characterization of parametric subtyping, as a fragment of structural subtyping~(\cref{sec:parametric});
to present a saturation algorithm for deciding parametric subtyping for parametric polymorphism~(\cref{sec:decide-poly}), as well as proofs of its soundness and completeness with respect to the declarative characterization (\cref{thm:sound-poly,thm:complete-poly});
to implement this decision procedure~(\cref{sec:implementation}); and
to give, as a special case of this decision procedure, a saturation-based decision procedure for structural subtyping of monomorphic types that has several advantages over existing algorithms~(\cref{sec:decide-mono}).
Details of the proofs sketched here can be found in Appendix~\ref{sec:appendix:proofs}.

\section{Structural subtyping for parametric polymorphism}\label{sec:structural}

In this \lcnamecref{sec:structural}, we describe the syntax of types, present a declarative characterization of structural subtyping, and show that it is undecidable in the presence of recursively defined type constructors.

\subsection{Syntax of types}

Programmers write types in the form to which they are accustomed, such as in the type definition $\listn[\alpha] \defd \plus*{\nil\colon \one ,\, \cons\colon \alpha \tensor \listn[\alpha]}$.
However, throughout this paper, it will often be convenient to work with types in a normal form that maintains a strict distinction and alternation between \emph{named types} $\tau$ and \emph{structural types} $A$.
For this reason, the programmer-defined types will be normalized in a preliminary elaboration phase that
inserts additional type constructors, in a manner reminiscent of the conversion of context-free grammars to \citeauthor{Greibach:JACM65} normal form~\shortcite{Greibach:JACM65} and the syntax \citet{Huet:MSCS98} used in deciding extensional equality of total B\"{o}hm trees.
Details of this elaboration are postponed to \cref{sec:implementation}.

For types in normal form, the syntax is as follows.
In addition to these syntactic categories, we use $\alpha$ for type constructor parameters and $x$ for explicitly quantified type variables.
\begin{alignat*}{2}
  \text{\emph{Named types}} &\quad&
    \tau , \sigma &\Coloneqq t[\theta] \mid \alpha \mid x
  \\
  \text{\emph{Type substitutions}} &&
    \theta , \phi &\Coloneqq \sube* \mid \theta , \tau/\alpha
  \\
  \text{\emph{Structural types}} &&
    A , B &\Coloneqq \begin{lgathered}[t]
                       \tau_1 \tensor \tau_2 \mid \one \mid \plus*[\ell \in L]{\ell\colon \tau_\ell} \mid \exists x.\tau \\[-\jot]
                       \mathllap{\mid {}} \tau_1 \imp \tau_2 \mid \with*[\ell \in L]{\ell\colon \tau_\ell} \mid \forall x.\tau
                     \end{lgathered}
  \\
  \text{\emph{Definitions}} &&
    \sig &\Coloneqq \sige* \mid \sig , t[\alphas] \defd A
\quad
\text{\emph{(exactly one defn.\ }}\mathrlap{\text{\emph{per $t$)}}}
\end{alignat*}

\subsubsection{Structural types}
The structural types $A$ consist of:
product types $\tau_1 \tensor \tau_2$ and the unit type~$\one$;
variant record types $\plus*[\ell \in L]{\ell\colon \tau_\ell}$, indexed by (possibly empty) finite sets $L$ of alternatives~$\ell$;
existentially quantified types $\exists x.\tau$;
function types $\tau_1 \imp \tau_2$;
record types $\with*[\ell \in L]{\ell\colon \tau_\ell}$, again indexed by (possibly empty) finite sets $L$; and
universally quantified types $\forall x.\tau$.
The somewhat nonstandard feature of this syntax is that structural types $A$ have only named types $\tau$ as immediate subformulas.
This enforces the first part of the strict alternation between structural and named types that is prescribed by the normal form.

\subsubsection{Named types and definitions}\label{sec:structural:named}
Named types $\tau$ are primarily type constructor instantiations of the form $t[\theta]$, where $t$ is a defined type constructor\footnote{%
\emph{Defined} type constructors $t$ are distinct from structural type constructors like $\imp$.
However, in the remainder of this paper, we will frequently drop the `defined' qualifier for conciseness and simply use `type constructor' to refer exclusively to defined type constructors.}
and $\theta$ is a type substitution.
Such type constructors $t$ are recursively defined\footnote{%
We could have chosen to use a recursion operator $\mu$ and explicit $\mathsf{fold}$s and $\mathsf{unfold}$s, but by using definitions, we avoid the complication of comparing $\mu$-types for equality.
We also find definitions easier to read and closer to actual practice.}
in a set $\sig$ of definitions.
There are finitely many definitions of the form $t[\alphas] \defd A$, exactly one for each defined type constructor, where the structural type $A$ may contain free occurrences of the parameters $\alphas$ but must be otherwise closed.
The substitution~$\theta$ in $t[\theta]$ then serves to instantiate the type parameters $\alphas$ used in $t$'s definition.
(In examples, we use an application-like syntax in place of substitutions, such as $t[\tau]$ instead of $t[\tau/\alpha]$.)

Notice that definitions enforce the other part of the strict alternation between structural types and named types that is prescribed by this normal form: type constructors $t[\alphas]$ are defined only in terms of structural types $A$, not named types $\tau$.
Moreover, this ensures that all definitions are contractive~\cite{Gay+Hole:Acta05}.

Given the shallow syntax of structural types, named types $\tau$ must also include type parameters $\alpha$, so that the structural body of a definition $t[\alphas] \defd A$ may indeed contain occurrences of parameters $\alphas$.
Similarly, named types $\tau$ also include type variables $x$ bound by the $\forall$ and $\exists$ quantifiers.

\subsubsection{Type substitutions}\label{sec:structural:subs}
In structural subtyping, definitions $t[\alphas] \defd A$ will be interpreted transparently, with $t[\theta]$ and its unfolding, $\theta(A)$, being treated indistinguishably (aside from belonging to distinct syntactic categories).
Because such type definitions are closed apart from their parameters $\alphas$, the domains of type substitutions $\theta$ consist only of type parameters $\alpha$.
Moreover, substitutions map these parameters to named types $\tau$, not to structural types, so that the instantiation of a structural type, $\theta(A)$, is itself a well-formed structural type.

\subsubsection{Examples}

Here we present two examples to which we will repeatedly return in this paper.

\paragraph{Even and odd natural numbers}

As a simple example of a type, the programmer could write the following type definitions to describe a unary representation of natural numbers, as well as even and odd natural numbers.%
\footnote{Recall that we choose, in this paper, to treat all types coinductively.
  Strictly speaking, the types $\nat$, $\even$, and $\odd$ therefore represent the respective finite natural numbers \emph{together with} the first limit ordinal, $\omega = \suc\mkern2mu\suc \dotsb$.
  This subtlety is familiar from lazy functional languages such as Haskell.}
(We omit $[]$ when a defined type takes no parameters.)
\begin{alignat*}{5}
%
  &&\qquad&&&&&&
  \odd &\defd \plus*{\suc\colon \even}
  \\[-\jot]
  &&&&
  \nat &\defd \plus*{\zero\colon \one ,\, \suc\colon \nat}
  &\qquad&&
  \even &\defd \plus*{\zero\colon \one ,\, \suc\colon \odd}
\intertext{The elaboration phase would normalize these types by introducing an auxiliary type name $\onen$ and revising the definitions of $\nat$ and $\even$ so that structural and named types alternate:}
  &&&&&&&&
  \odd &\defd \plus*{\suc\colon \even}
  \\[-\jot]
  \onen &\defd \one 
  &&&
  \nat &\defd \plus*{\zero\colon \onen ,\, \suc\colon \nat} 
  &&&
  \even &\defd \plus*{\zero\colon \onen ,\, \suc\colon \odd}
  %
  %
\end{alignat*}
To avoid the bureaucracy of having to write types in normal form, future examples given in this paper presume that types will be normalized during elaboration.

The even and odd natural numbers are, of course, subsets of the natural numbers.
So, taking a sets-of-values interpretation of subtyping, we ought to have $\even$ and $\odd$ as subtypes of $\nat$, but we ought \emph{not} to have $\nat$ as a subtype of $\even$ and $\odd$.

\paragraph{Context-free languages}
As a more complex example, we consider the type of words belonging to the context-free language $\set{\leftn^n \rightn^n \$ \mid n \geq 0}$.
(The terminal symbol, $\$$, is necessary to make the language prefix-free~\cite{Korenjak+Hopcroft:SWAT66} and thereby represent the empty word in a typable way.)
To aid intuition, we show both the context-free grammar (in \citeauthor{Greibach:JACM65} normal form~\shortcite{Greibach:JACM65}) for this language on the left and the corresponding, quite parallel, type definitions on the right.%
\footnote{Once again, because all types are interpreted coinductively in this paper, the type $\edyck_0$ would also be inhabited by the infinite word $\leftn\mkern2mu\leftn\dotsb$.}
\begin{equation*}
  \begin{array}[t]{@{}c@{}}
    \text{\emph{CFG in Greibach normal form}}
    \\ \midrule \addlinespace
    \,\begin{aligned}[t]
      e_0 &\rightarrow \leftn \, e \, \mathit{end} \mid \$
      \\
      e &\rightarrow \leftn \, e \, r \mid \rightn
    \end{aligned}
    \quad\enspace
    \begin{aligned}[t]
      \mathit{end} &\rightarrow \$
      \\
      r &\rightarrow \rightn
    \end{aligned}
  \end{array}
  \qquad\qquad
  \begin{array}[t]{@{}c@{}}
    \text{\emph{Type definitions}}
    \\ \midrule \addlinespace
    \,\begin{aligned}[t]
      \edyck_0 &\defd \plus*{\leftn\colon \edyck[\name{end}] ,\, \$\colon \onen}
      \\
      \edyck[\kappa] &\defd \plus*{\leftn\colon \edyck[\name{r}[\kappa]] ,\, \rightn\colon \kappa}
    \end{aligned}
    \enspace\text{\emph{where}}\enspace
    \begin{aligned}[t]
      \name{end} &\defd \plus*{\$\colon \onen}
      \\
      \name{r}[\kappa] &\defd \plus*{\rightn\colon \kappa}
    \end{aligned}
  \end{array}
\end{equation*}
Here, the type $\edyck_0$ relies on the constructor $\edyck[\kappa]$, which describes the language $\set{\leftn^n \rightn^{n+1} \kappa \mid n \geq 0}$; that is, the parameter $\kappa$ maintains a continuation to be used when the unmatched $\rightn$ is produced.
Because the type $\edyck_0$ refers to $\edyck[-]$ only after producing an initial $\leftn$, the words described by $\edyck_0$ are indeed a string of $\leftn$s followed by the same number of $\rightn$s (followed by $\$$).

In a similar way the context-free grammar (again in Greibach normal form) and the type $\dyck_0$ given below describe the ($\$$-terminated) Dyck language of balanced delimiters, here $\leftn$ and $\rightn$.
\begin{equation*}
  \begin{array}[t]{@{}c@{}}
    \text{\emph{CFG in Greibach normal form}}
    \\ \midrule \addlinespace \,
    \begin{aligned}
      d_0 &\rightarrow \leftn \, d \, d_0 \mid \$
      \\
      d &\rightarrow \leftn \, d \, d \mid \rightn
    \end{aligned}
  \end{array}
  \qquad\qquad
  \begin{array}[t]{@{}c@{}}
    \text{\emph{Type definitions}}
    \\ \midrule \addlinespace \,
    \begin{aligned}
      \dyck_0 &\defd \plus*{\leftn\colon \dyck[\dyck_0] ,\, \$\colon \onen}
      \\
      \dyck[\kappa'] &\defd \plus*{\leftn\colon \dyck[\dyck[\kappa']] ,\, \rightn\colon \kappa'}
    \end{aligned}
  \end{array}
\end{equation*}
The type $\dyck_0$ relies on the type constructor $\dyck[\kappa']$, which describes the context-free language of ``nearly balanced'' delimiters, in which words of balanced delimiters are followed by one additional unmatched $\rightn$;
once again, the type parameter $\kappa$ maintains a continuation to be used when that unmatched $\rightn$ is produced.
The type $\dyck_0$ refers to $\dyck[-]$ only after producing an initial $\leftn$, so the words described by $\dyck_0$ are indeed balanced.

Because $\set{\leftn^n \rightn^n \$ \mid n \geq 0}$ is a subset of the $\$$-terminated Dyck language, we ought to have $\edyck_0$ as a subtype of $\dyck_0$, but \emph{not} $\dyck_0$ as a subtype of $\edyck_0$.

\subsection{Structural subtyping}

Because our normalized types are separated into named types and structural types, structural subtyping will be given a declarative characterization in terms of derivations of two judgments:
$\dsubtype{\tau{} < \sigma{}}$ for named type $\tau$ as a subtype of named type $\sigma$, and $\dsubtype{A{} < B{}}$ for structural type $A$ as a subtype of structural type $B$.
Derivations of the $\dsubtype{\tau{} < \sigma{}}$ and $\dsubtype{A{} < B{}}$ judgments will be defined coinductively.
That is, these derivations may be infinitely deep (but will be finitely wide).
Stated differently, subtyping's coinductive nature and underlying greatest fixed point mean that a subtyping relationship holds in the absence of a counterexample, and that absence is witnessed by a potentially infinite derivation.\footnote{%
For monomorphic subtyping, merely \emph{circular} derivations~\cite{Brotherston+Simpson:JLC10}, which are finite representations of \emph{regular} infinite derivations, would suffice~\cite{Lakhani+:ESOP22}.}

Returning to the first of our running examples, for $\even$ to be a subtype of $\nat$, we must be able to construct infinite derivations of $\dsubtype{\even{} < \nat{}}$.
On the other hand, because $\nat$ ought \emph{not} to be a subtype of $\even$, there must \emph{not} exist a derivation of $\dsubtype{\nat{} < \even{}}$.

The entire set of inference rules used to construct (potentially) infinite derivations of subtyping judgments can be found in \cref{fig:structural-subtyping}.
\begin{figure}
  \begin{gather*}
    \infer[\srule{\jrule{UNF-}}]{\dsubtype{{t[\theta]}{} < {u[\phi]}{}}}{
      t[\alphas] \defd A & u[\betas*] \defd B &
      \dsubtype{{\theta(A)}{} < {\phi(B)}{}}}
    \\[\figinferskip]
    \infer[\srule{\jrule{VAR-}}]{\dsubtype{x{} < x{}}}{}
    \qquad
    \raisebox{0.15\baselineskip}{(no rules for $\dsubtype{x{} < \tau{}}$ and $\dsubtype{\tau{} < x{}}$ when $\tau \neq x$)}
    \\[\figinferskip]
    \infer[\srule{\tensor}]{\dsubtype{{\tau_1 \tensor \tau_2}{} < {\sigma_1 \tensor \sigma_2}{}}}{
      \dsubtype{{\tau_1}{} < {\sigma_1}{}} &
      \dsubtype{{\tau_2}{} < {\sigma_2}{}}}
    \qquad
    \infer[\srule{\one}]{\dsubtype{\one{} < \one{}}}{}
    \qquad
    \infer[\srule{\plus}]{\dsubtype{{\plus*[\ell \in L]{\ell\colon \tau_\ell}}{} < {\plus*[k \in K]{k\colon \sigma_k}}{}}}{
      (L \subseteq K) &
      \forall \ell \in L\colon 
        \dsubtype{{\tau_\ell}{} < {\sigma_\ell}{}}}
    \\[\figinferskip]
    \infer[\srule{\imp}]{\dsubtype{{\tau_1 \imp \tau_2}{} < {\sigma_1 \imp \sigma_2}{}}}{
      \dsubtype{{\sigma_1}{} < {\tau_1}{}} &
      \dsubtype{{\tau_2}{} < {\sigma_2}{}}}
    \qquad
    \infer[\srule{\with}]{\dsubtype{{\with*[\ell \in L]{\ell\colon \tau_\ell}}{} < {\with*[k \in K]{k\colon \sigma_k}}{}}}{
      (K \subseteq L) &
      \forall k \in K\colon 
        \dsubtype{{\tau_k}{} < {\sigma_k}{}}}
    \\[\figinferskip]
    \infer[\srule{\forall}]{\dsubtype{{\forall x.\tau}{} < {\forall y.\sigma}{}}}{
      \text{($z$ fresh)} &
      \dsubtype{{[z/x]\tau}{} < {[z/y]\sigma}{}}}
    \qquad
    \infer[\srule{\exists}]{\dsubtype{{\exists x.\tau}{} < {\exists y.\sigma}{}}}{
      \text{($z$ fresh)} &
      \dsubtype{{[z/x]\tau}{} < {[z/y]\sigma}{}}}
  \end{gather*}
  \caption{Structural subtyping rules ({\normalfont\scshape s} for `structural').  These inference rules are interpreted coinductively and are most clearly read bottom-up, from conclusion to premises.}\label{fig:structural-subtyping}
\end{figure}
These rules are interpreted coinductively and are most clearly read bottom-up, from conclusion to premises.
We will now comment on a few of the rules.

\subsubsection{Structural subtyping of named types}

Structural subtyping treats type definitions in an entirely transparent way:
when $t[\alphas] \defd A$ and $u[\betas*] \defd B$, the type $t[\theta]$ is a subtype of $u[\phi]$ exactly when the same subtyping relationship holds for their unfoldings, $\theta(A)$ and $\phi(B)$, respectively.
This is expressed by the $\srule{\jrule{UNF-}}$ rule.
Additionally, a type variable $x$ is considered to be a subtype of only itself, as captured in the $\srule{\jrule{VAR-}}$ rule.

\subsubsection{Structural subtyping of structural types}

Aside from the alternation of structural and named types, the rules for structural subtyping of structural types, $\dsubtype{A{} < B{}}$, are standard~\cite{Pierce:TAPL02}.
The rules decompose the structural types into their immediate subformulas and then require certain subtyping relationships on those subformulas.
For example, the $\srule{\plus}$ rule for variant record types is standard (see e.g.~\cite{Pierce:TAPL02,Gay+Hole:Acta05}).
For $\plus*[\ell \in L]{\ell\colon \tau_\ell}$ to be a subtype of $\plus*[k \in K]{k\colon \sigma_k}$, the condition $L \subseteq K$ demands that the latter type offer at least as many alternatives as, but possibly more than, the former type, thereby accounting for width subtyping of variant record types.
Moreover, by requiring that $\dsubtype[i]{{\tau_\ell}{} < {\sigma_\ell}{}}$ holds for all alternatives $\ell$ shared by the two types, this rule also accounts for covariant depth subtyping of variant record types.

Subtyping for the polymorphic quantifiers $\forall x.\tau$ and $\exists x.\tau$ is also standard for explicit, Church-style polymorphism.
We certainly could have unified the $\srule{\forall}$ and $\srule{\exists}$ rules into a single $\srule{\mu}$ rule, with a side condition that $\mu \in \set{\forall,\exists}$.
However, because $\forall x.\tau$ and $\exists x.\tau$ will type different terms, we prefer to maintain distinct subtyping rules for $\forall$ and $\exists$.
Moreover, as previously mentioned in \cref{sec:introduction}, we do not consider subtyping for implicit, Curry-style polymorphism, nor bounded quantification, in this paper, leaving these as future work.

As previously mentioned, we do \emph{not} consider subtyping for implicit, Curry-style polymorphism~\cite{Odersky+Laeufer:POPL96} or bounded quantification~\cite{Cardelli+Wegner:CSUR85,Cardelli+:IC94} in this paper.
Because our interest is in the interaction of subtyping, Church-style polymorphism, and recursion, these are outside the scope of this paper and left as future work.

\subsubsection{Examples}\label{sec:structural:examples}

We now return to the running examples in the context of structural subtyping.

\paragraph{Even and odd natural numbers}

For $\even$ and $\odd$ to be subtypes of $\nat$, we must be able to construct derivations of $\dsubtype{\even{} < \nat{}}$ and $\dsubtype{\odd{} < \nat{}}$.
Because structural subtyping derivations are (potentially) infinite, they cannot be directly written down in their entirety.
A finite, constructive proof of their existence instead suffices, and a useful proof technique here is coinduction.
For example, for $\dsubtype{\even{} < \nat{}}$ and $\dsubtype{\odd{} < \nat{}}$, mutual coinduction can be used:
\begin{equation*}
  \infer[\srule{\jrule{UNF-}}]{\dsubtype{\even{} < \nat{}}}{
    \infer[\mathrlap{\srule{\plus}}]{\dsubtype{{\plus*{\zero\colon \onen ,\, \suc\colon \odd}}{} < {\plus*{\zero\colon \onen ,\, \suc\colon \nat}}{}}}{
      \infer[\srule{\jrule{UNF-}}]{\dsubtype{\onen{} < \onen{}}}{
        \infer[\srule{\one}]{\dsubtype{\one{} < \one{}}}{}} &
      \infer-{\dsubtype{\odd{} < \nat{}}}{}}}
  \quad\;\text{\emph{and}}\quad\enspace
  \infer[\srule{\jrule{UNF-}}]{\dsubtype{\odd{} < \nat{}}}{
    \infer[\mathrlap{\srule{\plus}}]{\dsubtype{{\plus*{\suc\colon \even}}{} < {\plus*{\zero\colon \onen ,\, \suc\colon \nat}}{}}}{
      \infer-{\dsubtype{\even{} < \nat{}}}{}}}
\end{equation*}
We use a dotted line to indicate the coinductive appeals to $\dsubtype{\odd{} < \nat{}}$ and $\dsubtype{\even{} < \nat{}}$, which can also be thought of as admissible structural subtyping rules -- admissible in the sense that they can always be unfolded to the corresponding infinite derivations, which involve only rules found in \cref{fig:structural-subtyping}.
Each of these appeals is guarded by the $\srule{\jrule{UNF-}}$ and $\srule{\plus}$ rules.

Here, the full expressive power of infinite derivations is not needed.
Because the types are monomorphic, circular derivations~\cite{Brotherston+Simpson:JLC10}, which are finite representations of \emph{regular} infinite derivations, would suffice~\cite{Lakhani+:ESOP22}:
For $\dsubtype{\even{} < \nat{}}$, the right-hand derivation segment could be inlined within the left-hand segment, with the inlined coinductive appeal to $\dsubtype{\even{} < \nat{}}$ then circling back to $\dsubtype{\even{} < \nat{}}$ at the ``root''.
Then $\dsubtype{\odd{} < \nat{}}$ is similar.

As a negative example, we \emph{cannot} derive $\dsubtype[i]{\nat{} < \odd{}}$ because, after unfolding $\nat$ and $\odd$ with the $\srule{\jrule{UNF-}}$ rule, we would need to show that $\set{\zero , \suc} \subseteq \set{\suc}$, which is simply false.
Similarly, we \emph{cannot} derive $\dsubtype[i+1]{\nat{} < \even{}}$ because, after unfolding $\nat$ and $\even$, we would need to derive $\dsubtype[i]{\nat{} < \odd{}}$.

\paragraph{Context-free languages}

Recall that $\set{\leftn^n \rightn^n \$ \mid n \geq 0}$ is a subset of the Dyck language and that the type $\edyck_0$ should accordingly be a subtype of $\dyck_0$; there ought therefore to exist a derivation of $\dsubtype{{\edyck_0}{} < {\dyck_0}{}}$.
However, direct application of coinduction is not enough to establish $\dsubtype{{\edyck[\name{end}]}{} < {\dyck[\dyck_0]}{}}$
because it produces an infinite stream of subgoals, $\dsubtype{{\edyck[\name{end}]}{} < {\dyck[\dyck_0]}{}} , \dsubtype{{\edyck[\name{r}[\name{end}]]}{} < {\dyck[\dyck[\dyck_0]]}{}} , \dotsc$, none of which is a direct instance of any preceding one. 
Instead, we generalize the coinductive hypothesis, proving that $\dsubtype{\kappa{} < {\kappa'}{}}$ implies $\dsubtype{{\edyck[\kappa]}{} < {\dyck[\kappa']}{}}$ for all named types $\kappa$ and $\kappa'$.
Then, because $\dsubtype{{\name{end}}{} < {\dyck_0}{}}$, derivations of $\dsubtype{{\edyck[\name{end}]}{} < {\dyck[\dyck_0]}{}}$ and hence of $\dsubtype{{\edyck_0}{} < {\dyck_0}{}}$ indeed exist.
\begin{equation*}
  \small
  \infer[\srule{\jrule{UNF-}}]{\dsubtype{{\edyck_0}{} < {\dyck_0}{}}}{
    \infer[\srule{\plus}]{\dsubtype{{\plus*{\leftn\colon \edyck[\name{end}] ,\, \$\colon \one}}{} < {\plus*{\leftn\colon \dyck[\dyck_0] ,\, \$\colon \one}}{}}}{
      \infer-{\dsubtype{{\edyck[\name{end}]}{} < {\dyck[\dyck_0]}{}}}{
        \infer[\mathrlap{\srule{\jrule{UNF-}}}]{\dsubtype{{\name{end}}{} < {\dyck_0}{}}}{
          \infer[\srule{\plus}]{\dsubtype{{\plus*{\$\colon \one}}{} < {\plus*{\leftn\colon \dyck[\dyck_0] ,\, \$\colon \one}}{}}}{
            \infer[\srule{\one}]{\dsubtype{\one{} < \one{}}}{}}}} & \hspace{-1em}
      \infer[\srule{\one}]{\dsubtype{\one{} < \one{}}}{}}}
  \enspace\text{\emph{and}}\hspace{-0.5em}
  \infer[\srule{\jrule{UNF-}}]{\dsubtype{{\edyck[\kappa]}{} < {\dyck[\kappa']}{}}}{
    \infer[\srule{\plus}]{\dsubtype{{\plus*{\leftn\colon \edyck[\name{r}[\kappa]] ,\, \rightn\colon \kappa}}{} < {\plus*{\leftn\colon \dyck[\dyck[\kappa']] ,\, \rightn\colon \kappa'}}{}}}{
      \infer-{\dsubtype{{\edyck[\name{r}[\kappa]]}{} < {\dyck[\dyck[\kappa']]}{}}}{
        \infer[\mathrlap{\srule{\jrule{UNF-}}}]{\dsubtype{{\name{r}[\kappa]}{} < {\dyck[\kappa']}{}}}{
          \infer[\srule{\plus}]{\dsubtype{{\plus*{\rightn\colon \kappa}}{} < {\plus*{\leftn\colon \dyck[\dyck[\kappa']] ,\, \rightn\colon \kappa'}}{}}}{
            \dsubtype{\kappa{} < {\kappa'}{}}}}} & \hspace{-1em}
      \dsubtype{\kappa{} < {\kappa'}{}}}{}}
\end{equation*}
In other words, the rules marked with dotted lines are admissible and can always be unfolded to the partial derivation on the right-hand side above.

This example demonstrates why circular derivations do not suffice for subtyping of recursively defined type constructors that employ non-regular recursion:
In the right-hand derivation, we cannot directly close a cycle from $\dsubtype{{\edyck[\name{r}[\kappa]]}{} < {\dyck[\dyck[\kappa']]}{}}$ back to $\dsubtype{{\edyck[\kappa]}{} < {\dyck[\kappa']}{}}$.
The former is indeed an instance of the latter, but the subtyping depends on having $\dsubtype{\kappa{} < {\kappa'}{}}$ and so we are required to show the instance $\dsubtype{{\name{r}[\kappa]}{} < {\dyck[\kappa']}{}}$: we need the expressive power of \emph{non}-regular infinite derivations.
Given that $\set{\leftn^n \rightn^n \$ \mid n \geq 0}$ and the Dyck language are context-free languages, perhaps it is not surprising that the regularity of circular derivations is insufficiently expressive to establish $\dsubtype{{\edyck_0}{} < {\dyck_0}{}}$.

\subsection{Undecidability of structural subtyping}\label{sec:structural:undecidability}

\textit{A priori}, it seems possible that structural subtyping in the presence of recursively defined type constructors might be decidable.
After all, structural \emph{equality} for coinductively interpreted types is decidable~\cite{Das+:TOPLAS22}, although intractable, by reducing from trace equivalence for deterministic first-order grammars~\cite{Jancar:JCSS21}.

But as it turns out, structural subtyping is undecidable in the presence of recursively defined type constructors.
As previously discussed in \cref{sec:introduction,sec:structural}, the types in this paper are interpreted coinductively.
To prove undecidability, we therefore need a reduction from a correspondingly coinductive property.
We choose to reduce from simulation of guarded Basic Process Algebra~(BPA) processes~\cite{Bergstra+Klop:IC84}, a property which is itself undecidable~\cite{Groote+Huettel:IC94}.

\subsubsection{Background on basic process algebra}

Guarded BPA processes are defined by a set of guarded equations.
For our purposes, a general definition of guardedness is unimportant; what is important is that any set of guarded BPA equations can be put into the following Greibach normal form~\cite{Baeten+:JACM93}:
\begin{equation*}
  X \defd \bpasum*[\ell \in L]{(\bpaseq{\ell}{p'_\ell})}
  \,
  \text{, where $L$ is nonempty and $p, q \Coloneqq \bpaemp \mid \bpaseq{Y}{p} \,.$}
\end{equation*}
As usual for process algebras, there is a labeled transition system to describe process behavior.
When restricted to processes in Greibach normal form, the labeled transition system consists of a single rule:
\begin{equation*}
  \infer{\bpaseq{X}{q} \trans[a] p'_a \odot q}{
    X \defd \bpasum*[\ell \in L]{(\bpaseq{\ell}{p'_\ell})} &
    (a \in L)}
  \qquad
  \raisebox{0.85\baselineskip}{\text{(no rule for $\bpaemp$)}}
  \qquad\text{\emph{where}}\qquad
  \begin{aligned}[b]
    \bpaemp \odot q &= q \\
    (\bpaseq{X}{p}) \odot q &= \bpaseq{X}{(p \odot q)}
  \end{aligned}
  \,.
\end{equation*}

The simulation (or ``is-simulated-by'') relation, $\simu$, for BPA processes is the largest relation such that whenever $p \simu q$ and $p \trans*[a] p'$ hold, there exists a process $q'$ for which $p' \simu q'$ and $q \trans*[a] q'$ hold.
In particular, $\bpaemp \simu q$ holds for all processes $q$ because $\bpaemp$ cannot make any transitions.

\subsubsection{Reduction of BPA simulation to structural subtyping}

For each guarded BPA equation in Greibach normal form, $X \defd \bpasum*[\ell \in L]{(\bpaseq{\ell}{p'_\ell})}$ where $L$ is nonempty, we define a corresponding type constructor $t_X[\alpha]$ that encodes the behavior of process variable $X$, parametrically in a type $\alpha$ that describes the behavior to follow that of $X$:
\begin{equation*}
  t_X[\alpha] \defd \with*[\ell \in L]{\ell\colon \bpa{p'_\ell}{\alpha}}
  \quad\text{\emph{where}}\quad\,
    \bpa{\bpaemp}{\tau} = \tau 
    \enspace\text{\emph{and}}\enspace
    \bpa*{\bpaseq{X}{p}}{\tau} = t_X[\bpa{p}{\tau}]
  \,.
\end{equation*}
($\bpa{p}{\tau}$ yields a type in normal form because $p$ is finite and $\tau$ is a named type.)
The ideas behind this encoding are twofold.
First, width subtyping of $\with$ ensures that a process $\bpaseq{Y}{q}$ can match any transition that $\bpaseq{X}{p}$ can take: width subtyping ensures that the type $t_Y[\beta]$ offers at least those alternatives that the type $t_X[\alpha]$ does.
Second, depth subtyping for $\with$ ensures that this simulation holds hereditarily for the processes to which $\bpaseq{X}{p}$ and $\bpaseq{Y}{q}$ transition.
We can prove the following.%
\begin{restatable}{theorem}{undecidablewith}\label{thm:undecidable-with-empty}
  Let $t \defd \with*{}$.
  Then $p \simu q$ if and only if $\dsubtype{{(\bpa{q}{t})}{} < {(\bpa{p}{t})}{}}$, for all processes $p$ and $q$.
\end{restatable}
\begin{proof}[Proof sketch]
  We prove each direction separately, by coinduction on the respective conclusion.
\end{proof}
\noindent
The key properties of $t$ necessary for the proof are: that $\dsubtype{{(\bpa{q}{t})}{} < t{}}$, for all processes $q$; and
that $\dsubtype{t{} < {(\bpa{p}{t})}{}}$ implies $p = \bpaemp$, for all processes $p$.
Because simulation for BPA processes is undecidable~\cite{Groote+Huettel:IC94}, we therefore have the following \lcnamecref{cor:undecidable-with-empty}.
\begin{corollary}\label{cor:undecidable-with-empty}
  In the presence of record types with no alternatives and recursively defined type constructors, structural subtyping is undecidable.
\end{corollary}

\noindent
Although they make for an arguably cleaner proof, record types with no alternatives are not at all essential.
Even if all record types must have at least one alternative, structural subtyping is still undecidable.
The encoding can be revised to include an endmarker, $\$$, as an additional alternative for each $t_X$.
Let $t_0$ be any closed type, such as $t_0 \defd \with*{\$\colon t_0}$ or $t_0 \defd \one$ (among others), and define
\begin{equation*}
  t_X[\alpha] \defd \with*[\ell \in L]{\ell\colon \bpa{p'_\ell}{\alpha}} \with \lbrace\$\colon t_0\rbrace
  \,,
\end{equation*}
where $\bpa{p}{\tau}$ is defined as above.
With the revised encoding, we can prove the following \lcnamecref{thm:undecidable-with-eow}.
\begin{theorem}\label{thm:undecidable-with-eow}
  Let $t \defd \with*{\$\colon t_0}$.
  Then $p \simu q$ if and only if $\dsubtype{{(\bpa{q}{t})}{} < {(\bpa{p}{t})}{}}$, for all processes $p$ and $q$.
\end{theorem}
\begin{corollary}\label{cor:undecidable-with-eow}
  In the presence of record types and recursively defined type constructors, structural subtyping is undecidable.
\end{corollary}
\noindent
Furthermore, because we assign a coinductive interpretation to all types, virtually the same \lcnamecrefs{thm:undecidable-with-empty} hold for variant record types, with $\plus$ substituted for $\with$ in the definitions of $t_X[\alpha]$ -- only the subtyping direction changes to ``$p \simu q$ if and only if $\dsubtype{{(\bpa{p}{t})}{} < {(\bpa{q}{t})}{}}$.''
That is, structural subtyping remains undecidable in the presence of variant record types and recursively defined type constructors, \emph{even if\/} there are no record types at all.%
\footnote{In a mixed inductive/coinductive setting, where variant record types would be interpreted inductively, we conjecture that structural subtyping would remain undecidable even in the purely inductive fragment (\textit{i.e.}, without record types), and that this could be proved by reducing from context-free/BPA language inclusion~\cite{Friedman:TCS76,Groote+Huettel:IC94}, not BPA simulation.
}

\section{Parametric subtyping for parametric polymorphism}\label{sec:parametric}

In this \lcnamecref{sec:parametric}, we identify a fragment of structural subtyping for parametric polymorphism that has a relatively simple declarative characterization.
(\cref{sec:decide-poly} will show that this fragment is also decidable.)
We call this fragment \emph{parametric subtyping}, for its basis is a notion of parametricity.

Recall the context-free languages example from \cref{sec:structural:examples} in which we proved that the type $\edyck_0$, corresponding to $\set{\leftn^n\rightn^n \$ \mid n \geq 0}$, is a subtype of the Dyck language type $\dyck_0$.
We could not use direct coinduction to prove the existence of a derivation of $\dsubtype{{\edyck_0}{} < {\dyck_0}{}}$ because that led to an infinite stream of subgoals, 
$\dsubtype{{\edyck[\name{end}]}{} < {\dyck[\dyck_0]}{}} ,
\dsubtype{{\edyck[\name{r}[\name{end}]]}{} < {\dyck[\dyck[\dyck_0]]}{}} ,
\dsubtype{{\edyck[\name{r}[\name{r}[\name{end}]]]}{} < {\dyck[\dyck[\dyck[\dyck_0]]]}{}} ,
\dotsc$,
none of which is an instance of any preceding one.
We somehow need to quotient this space into finitely many subproblems, each of which is decidable.

The key insight behind this quotienting comes in revisiting the coinductive generalization that we used to prove $\dsubtype{{\edyck_0}{} < {\dyck_0}{}}$: in \cref{sec:structural:examples}, we proved that $\dsubtype{\kappa{} < {\kappa'}{}}$ implies $\dsubtype{{\edyck[\kappa]}{} < {\dyck[\kappa']}{}}$ for all named types $\kappa$ and $\kappa'$.
This could also be viewed as proving that the following inference rule is admissible from the structural subtyping rules of \cref{fig:structural-subtyping}.
\begin{equation*}
  \infer-{\dsubtype{{\edyck[\kappa]}{} < {\dyck[\kappa']}{}}}{
    \dsubtype{\kappa{} < {\kappa'}{}}}
\end{equation*}
This rule looks very much like the kind of rules around which rigid subtyping is based.
However, there is a key difference: rigid subtyping requires such parametric rules to use the same type constructor on both sides of the conclusion.

Most importantly for our purposes, this admissible rule is parametric, in the sense that the type constructors $\edyck[-]$ and $\dyck[-]$ map related arguments, $\dsubtype{\kappa{} < {\kappa'}{}}$, to related results, $\dsubtype{{\edyck[\kappa]}{} < {\dyck[\kappa']}{}}$.
This echoes \citeauthor{Reynolds:IFIP83}'s characterization of parametric functions as those that map related arguments to related results~\shortcite[Sec.~3]{Reynolds:IFIP83}.
Specifically, the right-hand side of his clause,
\begin{equation*}
  (f_1, f_2) \in r \imp r' \text{ iff } (f_1\,x_1, f_2\,x_2) \in r' \text{ for all } (x_1, x_2) \in r
  \,,
\end{equation*}
parallels the parametric subtyping rule for $\edyck[-]$ and $\dyck[-]$ above.%
\footnote{Moreover, the observation that parametric subtyping strictly generalizes rigid subtyping (which relates types only if they have the same outermost constructor) echoes \citeauthor{Reynolds:IFIP83}'s Abstraction Theorem, which states that his relational semantics relates the same expression in different environments.}

The importance of parametricity in the coinductive generalization used in this example suggests
that we ought to consider a notion of subtyping that uses parametric rules, \textit{i.e.}, rules of the form
\begin{equation*}
  \infer{\dsubtype{{t[\alphas]}{} < {u[\betas]}{}}}{
    \dsubtype{{\alpha_{i_1}}{} < {\beta_{j_1}}{}} 
    \enspace\!\dotsb\enspace
    \dsubtype{{\alpha_{i_m}}{} < {\beta_{j_m}}{}} &
    \dsubtype{{\beta_{j_{m+1}}}{} < {\alpha_{i_{m+1}}}{}}
    \enspace\!\dotsb\enspace
    \dsubtype{{\beta_{j_{m+n}}}{} < {\alpha_{i_{m+n}}}{}}}
  \enspace\text{{but not rules like}}\enspace
  \infer{\dsubtype{{t[\alphas]}{} < {u[\betas]}{}}}{
    \dsubtype{\alpha{} < \one{}} &
    \dsubtype{{t'[\alpha]}{} < \beta{}}}
  ,
\end{equation*}
as a candidate for being a decidable fragment of structural subtyping with a relatively simple declarative characterization.

\subsection{Declarative characterization of parametric subtyping}\label{sec:parametric:declarative}

The requirement that the candidate fragment use only parametric rules could already serve as a relatively simple declarative characterization.
However, to develop a decision procedure and prove its correctness, it is very useful to devise an equivalent declarative characterization that is more closely aligned with the presentation of structural subtyping.
Before doing so, it is helpful to see why structural subtyping, as defined in \cref{fig:structural-subtyping}, violates parametricity.

Consider the type definition $\name{snat}[\kappa] \defd \plus*{\zero\colon \kappa ,\, \suc\colon \name{snat}[\kappa]}$, which generalizes the type $\nat$ via the structural subtyping $\dsubtype{\nat{} < {\name{snat}[\one]}{}}$.
(The name $\name{snat}$, for ``serialized $\nat$'', alludes to serialized data structures, as discussed in \cref{sec:examples:serialized}.)
However, the admissible rule for $\nat$ and $\name{snat}[\kappa]$ would be the \emph{non-parametric} rule ``$\dsubtype{\nat{} < {\name{snat}[\kappa]}{}}$ if $\dsubtype{\one{} < \kappa{}}$.''
The $\srule{\jrule{UNF-}}$ rule of structural subtyping cannot detect non-parametric judgments, such as $\dsubtype{\one{} < \kappa{}}$ here, because unfolding eagerly applies substitutions and free type parameters do not appear:
by the time that structural subtyping reaches this non-parametric judgment, it will be $\dsubtype{\one{} < {[\one/\kappa]\kappa = \one}{}}$, with the non-parametricity no longer apparent in the judgment.

Therefore, instead of eagerly applying substitutions when unfolding, we need to postpone the substitutions, applying them only after determining that they do not conceal any non-parametricity.
The idea of postponing substitutions in this way is inspired by the \citeauthor{Girard:PhD72}--\citeauthor{Reynolds:IFIP83} logical relation for parametricity~\cite{Girard:PhD72,Reynolds:IFIP83}.
The judgments $\dsubtype{\tau{} < \sigma{}}$ and $\dsubtype{A{} < B{}}$ are revised to postpone substitutions by pushing them onto stacks when unfolding type constructor instantiations.
Substitution stacks are given by the grammar
\begin{alignat*}{2}
  \text{\emph{Substitution stacks}} &\quad&
    \subs , \subs* &\Coloneqq \sube* \mid \theta ; \subs
\end{alignat*}
and we thus arrive at the judgments $\dsubtype{\tau{\subs} < \sigma{\subs*}}$ and $\dsubtype{A{\subs} < B{\subs*}}$ for parametric subtyping.
As with structural subtyping, a parametric subtyping judgment holds if there exists a (potentially infinite) derivation of that judgment using the rules found in \cref{fig:parametric-subtyping}.
These declarative rules are again interpreted coinductively and are most clearly read bottom-up, from conclusion to premises.
\begin{figure}
  \begin{gather*}
    \infer[\prule{\jrule{INST-}}]{\dsubtype{{t[\theta]}{\subs} < {u[\phi]}{\subs*}}}{
      t[\alphas] \defd A & u[\betas*] \defd B &
      \dsubtype{A{\theta ; \subs} < B{\phi ; \subs*}}}
    \\[\figinferskip]
    \infer[\prule{\jrule{PARAM-}}]{\dsubtype{\alpha{\theta ; \subs} < \beta{\phi ; \subs*}}}{
      \dsubtype{{\theta(\alpha)}{\subs} < {\phi(\beta)}{\subs*}}}
    \qquad
    \infer[\prule{\jrule{VAR-}}]{\dsubtype{x{\subs} < x{\subs*}}}{}
    \\[\figinferskip]
    \text{(no rules for $\dsubtype{\alpha{\subs} < \tau{\subs*}}$ and $\dsubtype{\tau{\subs} < \beta{\subs*}}$ when $\tau$ is not a parameter)}
    \\[\figinferskip]
    \text{(no rules for $\dsubtype{x{\subs} < \tau{\subs*}}$ and $\dsubtype{\tau{\subs} < x{\subs*}}$ when $\tau \neq x$)}
    \\[\figinferskip]
    \infer[\prule{\tensor}]{\dsubtype{{\tau_1 \tensor \tau_2}{\subs} < {\sigma_1 \tensor \sigma_2}{\subs*}}}{
      \dsubtype{{\tau_1}{\subs} < {\sigma_1}{\subs*}} &
      \dsubtype{{\tau_2}{\subs} < {\sigma_2}{\subs*}}}
    \qquad
    \infer[\prule{\one}]{\dsubtype{\one{\subs} < \one{\subs*}}}{}
    \\[\figinferskip]
    \infer[\prule{\plus}]{\dsubtype{{\plus*[\ell \in L]{\ell\colon \tau_\ell}}{\subs} < {\plus*[k \in K]{k\colon \sigma_k}}{\subs*}}}{
      (L \subseteq K) &
      \forall \ell \in L\colon 
        \dsubtype{{\tau_\ell}{\subs} < {\sigma_\ell}{\subs*}}}
    \\[\figinferskip]
    \infer[\prule{\imp}]{\dsubtype{{\tau_1 \imp \tau_2}{\subs} < {\sigma_1 \imp \sigma_2}{\subs*}}}{
      \dsubtype{{\sigma_1}{\subs*} < {\tau_1}{\subs}} &
      \dsubtype{{\tau_2}{\subs} < {\sigma_2}{\subs*}}}
    \qquad
    \infer[\prule{\with}]{\dsubtype{{\with*[\ell \in L]{\ell\colon \tau_\ell}}{\subs} < {\with*[k \in K]{k\colon \sigma_k}}{\subs*}}}{
      (K \subseteq L) &
      \forall k \in K\colon 
        \dsubtype{{\tau_k}{\subs} < {\sigma_k}{\subs*}}}
    \\[\figinferskip]
    \infer[\prule{\forall}]{\dsubtype{{\forall x.\tau}{\subs} < {\forall y.\sigma}{\subs*}}}{
      (\text{$z$ fresh}) &
      \dsubtype{{[z/x]\tau}{\subs} < {[z/y]\sigma}{\subs*}}}
    \qquad
    \infer[\prule{\exists}]{\dsubtype{{\exists x.\tau}{\subs} < {\exists y.\sigma}{\subs*}}}{
      (\text{$z$ fresh}) &
      \dsubtype{{[z/x]\tau}{\subs} < {[z/y]\sigma}{\subs*}}}
  \end{gather*}
  \caption{Parametric subtyping rules ({\normalfont\scshape p} for `parametric').  These rules are interpreted coinductively.}\label{fig:parametric-subtyping}
\end{figure}

To better understand these rules, it can be helpful to imagine constructing a (potentially infinite) derivation of $\dsubtype{{t[\theta]}{\subs} < {u[\phi]}{\subs*}}$, where $t[\alphas] \defd A$ and $u[\betas*] \defd B$.
That would proceed as follows.
\begin{enumerate}
\item\label{item:inst-p}
  This judgment can only be derived by the $\prule{\jrule{INST-}}$ rule, which is parametric subtyping's answer to structural subtyping's $\srule{\jrule{UNF-}}$ rule.
  It unfolds $t[\theta]$ and $u[\phi]$, but it does not apply the substitutions $\theta$ and $\phi$ eagerly, instead postponing them by pushing them onto their respective stacks, $\subs$ and $\subs*$.
  The unfoldings, \textit{i.e.}, structural types $A$ and $B$, are then compared under the extended stacks, $(\theta ; \subs)$ and $(\phi ; \subs*)$, respectively.
  Notice that $A$ and $B$ will, in general, contain free occurrences of the respective parameters $\alphas$ and $\betas*$.

\item
  Next, these structural types $A$ and $B$ are decomposed according to parametric subtyping's rules for structural types, such as $\prule{\tensor}$ and $\prule{\imp}$.
  These rules are virtually the same as structural subtyping's rules for structural types, with the only difference being that substitution stacks are threaded through, unchanged, from conclusion to premises.
  After decomposing $A$ and $B$, there are several parametric subtyping subgoals of the form $\dsubtype{\tau{\theta ; \subs} < \sigma{\phi ; \subs*}}$ (or, in the case of the $\prule{\imp}$ rule's first premise, of the form $\dsubtype{\sigma{\phi ; \subs*} < \tau{\theta ; \subs}}$).

\item
  Depending on the structure of the named types $\tau$ and $\sigma$, there are several possibilities for each such judgment:
\begin{itemize}
\item If both $\tau$ and $\sigma$ are type constructor instantiations, then the judgment can only be derived by the $\prule{\jrule{INST-}}$ rule, returning us to step~(\ref{item:inst-p}).

\item If $\tau$ and $\sigma$ are type parameters $\alpha$ and $\beta$, then there is no violation of parametricity here, and the judgment can be derived by the $\prule{\jrule{PARAM-}}$ rule.
  The substitutions $\theta$ and $\phi$ are popped from their respective stacks and finally applied; the resulting subgoal is of the form $\dsubtype{{\theta(\alpha)}{\subs} < {\phi(\beta)}{\subs*}}$.
In types in normal form, substitutions map parameters to \emph{named} types only (recall \cref{sec:structural:subs}), so one of these three cases will again apply.

\item If either $\tau$ or $\sigma$ is a type parameter and the other one is not, this judgment violates parametricity.
  Accordingly, there is no rule that can derive this judgment, and therefore $\dsubtype{{t[\theta]}{\subs} < {u[\phi]}{\subs*}}$ does \emph{not} hold.
\end{itemize}
\end{enumerate}

The notion of parametric subtyping given in \cref{fig:parametric-subtyping} is sound with respect to structural subtyping as defined in \cref{fig:structural-subtyping}.
The substitutions postponed in stacks $\subs$ and $\subs*$ can instead be composed and applied eagerly, transforming instances of the $\prule{\jrule{INST-}}$ rule into instances of the $\srule{\jrule{UNF-}}$ rule and eliminating occurrences of the $\prule{\jrule{PARAM-}}$ rule.
\begin{restatable}[Soundness of parametric subtyping]{theorem}{parasound}\leavevmode
  If $\dsubtype{\tau{\subs} < \sigma{\subs*}}$, then $\dsubtype{{\applysubs{\subs}{\tau}}{} < {\applysubs{\subs*}{\sigma}}{}}$.
  Likewise, if $\dsubtype{A{\subs} < B{\subs*}}$, then $\dsubtype{{\applysubs{\subs}{A}}{} < {\applysubs{\subs*}{B}}{}}$.
\end{restatable}
\begin{proof}[Proof sketch]
  Using the mixed induction and coinduction proof technique described by \citet{Danielsson+Altenkirch:MPC10}.
  Specifically, the proof is by lexicographic mixed induction and coinduction, first by coinduction on the (potentially infinite) structural subtyping derivation, and then by induction on the finite substitution stack $\subs$.
\end{proof}
\noindent
However, the converse does not hold: parametric subtyping is incomplete with respect to structural subtyping, as the above example involving $\nat$ and $\name{snat}[\one]$ demonstrates.
\begin{theorem}[Incompleteness of parametric subtyping]
  $\tau$ and $\sigma$ exist such that the structural subtyping $\dsubtype{\tau{} < \sigma{}}$ holds but the parametric subtyping $\dsubtype{\tau{\subs} < \sigma{\subs*}}$ does not, for any\/ $\subs$ and $\subs*$.
\end{theorem}

\section{Deciding structural subtyping for monomorphic types}\label{sec:decide-mono}

In the next \lcnamecref{sec:decide-poly}, we will present a saturation-based decision procedure for parametric subtyping for parametric polymorphism.
But parameterized type constructors involve some complications, so in this \lcnamecref{sec:decide-mono}, we will provide a gentle introduction to the algorithm by presenting a saturation-based decision procedure for structural subtyping of monomorphic types -- that is, recursively defined types that do not take parameters nor use the structural types $\forall x.\tau$ and $\exists x.\tau$.

Establishing the decidability of structural subtyping for monomorphic types is not a contribution of this paper.
One existing decision procedure (see, \emph{e.g.},~\cite{Lakhani+:ESOP22}) directly employs backward search for a derivation of the structural subtyping judgment $\dsubtype{t{} < u{}}$, using the subtyping rules themselves~(\cref{fig:structural-subtyping}).
This procedure crucially depends on three properties: for monomorphic types, merely circular derivations suffice to characterize structural subtyping; circular derivations are finite; and there are finitely many subtyping problems involving named monomorphic types.

For polymorphic types, these key properties will no longer hold -- which is why we will introduce a forward-inference, saturation-based procedure here.
But even if one is uninterested in polymorphic types, this forward-inference procedure offers several distinct advantages over the backward-search algorithm, as we will discuss below.

\subsection{A forward-inference decision procedure for monomorphic structural subtyping}

To devise a decision procedure for monomorphic subtyping based on forward inference, we will exploit the fact that subtyping is a safety property and return to the idea that, in keeping with the safety slogan ``nothing bad ever happens,'' a subtyping relationship $\dsubtype{t{} < u{}}$ holds when there is no counterexample.
Very roughly speaking, our algorithm proceeds as a kind of automated refutation by contradiction, assuming that a derivation of $\dsubtype{t{} < u{}}$ exists and repeatedly inverting that assumed derivation to check that no violations of subtyping occur (\textit{i.e.}, that nothing bad happens) before reaching another subtyping problem, $\dsubtype{{t'}{} < {u'}{}}$.
Because the given set $\sig$ of type definitions contains finitely many definitions $t \defd A$ and $u \defd B$ and there are therefore finitely many subtyping problems $\dsubtype{t{} < u{}}$, we can check each problem in this way.

More precisely, the forward-inference procedure uses three judgments: the primary judgment, $\asubtype{t < u} \entails \bot$; and two intermediate judgments, $\asubtype{t < u} \entails \asubtype{A < B}$ and $\asubtype{t < u} \entails \asubtype{\tau < \sigma}$.
(Notice that we use $\preccurlyeq$ to distinguish these from the declarative $\dsub$.)
Ultimately, the judgment $\asubtype{t < u} \entails \bot$ will be inferred if and only if $t \ndsub u$, allowing us to decide the structural subtyping $\dsubtype{t{} < u{}}$ by running forward inference to saturation and checking that $\asubtype{t < u} \entails \bot$ has \emph{not} been inferred.
(Saturation is guaranteed, as we will prove in \cref{thm:mono-saturation} below.)
The judgments $\asubtype{t < u} \entails \asubtype{A < B}$ and $\asubtype{t < u} \entails \asubtype{\tau < \sigma}$ are inferred if and only if $\dsubtype{A{} < B{}}$ and $\dsubtype{\tau{} < \sigma{}}$, respectively, would necessarily occur as subderivations of any derivation of $\dsubtype{t{} < u{}}$ (assuming such a derivation exists).

Alternatively, these judgments can be seen as describing necessary consequences of $\dsubtype{t{} < u{}}$.
From yet another perspective, these judgments can be seen as stating those \emph{constraints} that must hold for the structural subtyping $\dsubtype{t{} < u{}}$ to be derivable, with $\bot$ being the unsatisfiable constraint.
This last perspective will prove particularly useful in \cref{sec:decide-poly} and the decision procedure for parametric subtyping of polymorphic types presented there.

\subsubsection{Forward inference}
Forward inference proceeds according to the rules found in \cref{fig:decide-mono}.%
\begin{figure}
  \begin{gather*}
    \infer[\arule{\jrule{INIT-}}]{\asubtype{t < u} \entails \asubtype{A < B}}{
      t \defd A & u \defd B}
    \qquad
    \infer[\arule{\jrule{COMPOSE-}}_\bot]{\asubtype{t < u} \entails \bot}{
      \asubtype{t < u} \entails \asubtype{t' < u'} &
      \asubtype{t' < u'} \entails \bot} 
    \\[\figinferskip]
    \infer[\arule{\tensor}]{\asubtype{t < u} \entails \asubtype{\tau_i < \sigma_i}}{
      \asubtype{t < u} \entails \asubtype{\tau_1 \tensor \tau_2 < \sigma_1 \tensor \sigma_2} & (i \in \set{1,2})}
    \qquad
    \raisebox{0.5\baselineskip}{(no $\arule{\one}$ rule for $\asubtype{t < u} \entails \asubtype{\one < \one}$)}
    \\[\figinferskip]
    \infer[\arule{\plus}]{\asubtype{t < u} \entails \asubtype{\tau_\ell < \sigma_\ell}}{
      \asubtype{t < u} \entails \asubtype{\plus*[\ell \in L]{\ell\colon \tau_\ell} < \plus*[k \in K]{k\colon \sigma_k}} &
      (\ell \in L \subseteq K)}
    \\[\figinferskip]
    \infer[\arule{\imp}_1]{\asubtype{t < u} \entails \asubtype{\sigma_1 < \tau_1}}{
      \asubtype{t < u} \entails \asubtype{\tau_1 \imp \tau_2 < \sigma_1 \imp \sigma_2}}
    \qquad
    \infer[\arule{\imp}_2]{\asubtype{t < u} \entails \asubtype{\tau_2 < \sigma_2}}{
      \asubtype{t < u} \entails \asubtype{\tau_1 \imp \tau_2 < \sigma_1 \imp \sigma_2}}
    \\[\figinferskip]
    \infer[\arule{\with}]{\asubtype{t < u} \entails \asubtype{\tau_k < \sigma_k}}{
      \asubtype{t < u} \entails \asubtype{\with*[\ell \in L]{\ell\colon \tau_\ell} < \with*[k \in K]{k\colon \sigma_k}} &
      (k \in K \subseteq L)}
    \\[\figinferskip]
    \infer[\arule{\plus\plus}_\bot]{\asubtype{t < u} \entails \bot}{
      \asubtype{t < u} \entails \asubtype{\plus*[\ell \in L]{\ell\colon \tau_\ell} < \plus*[k \in K]{k\colon \sigma_k}} &
      (L \nsubseteq K)}
    \\[\figinferskip]
    \infer[\arule{\with\with}_\bot]{\asubtype{t < u} \entails \bot}{
      \asubtype{t < u} \entails \asubtype{\with*[\ell \in L]{\ell\colon \tau_\ell} < \with*[k \in K]{k\colon \sigma_k}} &
      (K \nsubseteq L)}
    \quad
    \infer[\arule{\jrule{MISMATCH-}}_\bot]{\asubtype{t < u} \entails \bot}{
      \asubtype{t < u} \entails \asubtype{A < B} &
      A \mathrel{\Bot} B}
  \end{gather*}
  \caption{Forward inference rules for deciding structural subtyping of monomorphic types ({\normalfont\scshape f} for `forward').
These rules are interpreted inductively.
The notation $A \mathrel{\Bot} B$ means that $A$ and $B$ use distinct top-level structural type constructors, such as $\plus$ and $\one$.}\label{fig:decide-mono}
\end{figure}
Unlike the structural subtyping rules of \cref{fig:structural-subtyping}, these algorithmic rules are interpreted inductively and most clearly read top-down, from premises to conclusion.
To provide some intuition for this forward-inference decision procedure, we will walk through a few of the rules in detail.

\paragraph{The $\arule{\jrule{INIT-}}$ rule}
Suppose that the premises $t \defd A$ and $u \defd B$ hold.
If $\dsubtype{t{} < u{}}$ is derivable, then, by inversion, it must have been derived by applying the $\srule{\jrule{UNF-}}$ structural subtyping rule to a subderivation of $\dsubtype{A{} < B{}}$.
That is, $\dsubtype{A{} < B{}}$ would necessarily occur as a subderivation of $\dsubtype{t{} < u{}}$ when $t$ and $u$ are defined by $t \defd A$ and $u \defd B$, justifying the inference of $\asubtype{t < u} \entails \asubtype{A < B}$ by the $\arule{\jrule{INIT-}}$ rule.

\paragraph{The $\arule{\plus}$ and $\arule{\plus\plus}_\bot$ rules}
Suppose that the shared premise $\asubtype{t < u} \entails \asubtype{\plus*[\ell \in L]{\ell\colon \tau_\ell} < \plus*[k \in K]{k\colon \sigma_k}}$ has already been inferred -- that is, that any derivation of $\dsubtype{t{} < u{}}$ would necessarily contain a subderivation of $\dsubtype{{\plus*[\ell \in L]{\ell\colon \tau_\ell}}{} < {\plus*[k \in K]{k\colon \sigma_k}}{}}$.
By inversion, this subderivation can only be formed by applying the $\srule{\plus}$ structural subtyping rule with subderivations of $\dsubtype{{\tau_\ell}{} < {\sigma_\ell}{}}$, for all $\ell \in L$, and only when $L \subseteq K$.
Therefore, when $L \subseteq K$, the inference of $\asubtype{t < u} \entails \asubtype{\tau_\ell < \sigma_\ell}$, for all $\ell \in L$, by the $\arule{\plus}$ rule is justified.
On the other hand, when $L \nsubseteq K$, the subtypings $\dsubtype{{\plus*[\ell \in L]{\ell\colon \tau_\ell}}{} < {\plus*[k \in K]{k\colon \sigma_k}}{}}$ and hence $\dsubtype{t{} < u{}}$ are \emph{not} derivable, justifying the inference of $\asubtype{t < u} \entails \bot$ by the $\arule{\plus\plus}_\bot$ rule.

\paragraph{The $\arule{\jrule{COMPOSE-}}_\bot$ rule}
Suppose that the premises $\asubtype{t < u} \entails \asubtype{t' < u'}$ and $\asubtype{t' < u'} \entails \bot$ have already been inferred.
Thus, any derivation of $\dsubtype{t{} < u{}}$ would necessarily contain a subderivation of $\dsubtype{{t'}{} < {u'}{}}$, and moreover $\dsubtype{{t'}{} < {u'}{}}$ is \emph{not} derivable.
Therefore, $\dsubtype{t{} < u{}}$ is also not derivable, justifying the inference of $\asubtype{t < u} \entails \bot$ by the $\arule{\jrule{COMPOSE-}}_\bot$ rule.

\paragraph{The $\arule{\jrule{MISMATCH-}}_\bot$ rule}
Suppose that the premises $\asubtype{t < u} \entails \asubtype{A < B}$ and $A \mathrel{\Bot} B$ have already been inferred, where $A \mathrel{\Bot} B$ indicates that $A$ and $B$ have distinct top-level structural type constructors, such as $\plus*[\ell \in L]{\ell\colon \tau_\ell} \mathrel{\Bot} \one$.
Because the first premise has been inferred, any derivation of $\dsubtype{t{} < u{}}$ would necessarily contain a subderivation of $\dsubtype{A{} < B{}}$.
Because $A \mathrel{\Bot} B$, inversion shows there is no structural subtyping rule that could possibly form this subderivation.
Hence $\dsubtype{t{} < u{}}$ is \emph{not} derivable, justifying the inference of $\asubtype{t < u} \entails \bot$ by the $\arule{\jrule{MISMATCH-}}_\bot$ rule.

\subsubsection{Example}
Returning to our running example of even and odd natural numbers, we can examine the inferences made by our algorithm.
By virtue of the $\arule{\jrule{INIT-}}$ rule, the following judgments, among others, will be inferred:
\begin{align*}
  &\asubtype{\even < \nat} \entails \asubtype{\plus*{\zero\colon \onen ,\, \suc\colon \odd} < \plus*{\zero\colon \onen ,\, \suc\colon \nat}}
   ; \\[-\jot]
  &\asubtype{\odd < \nat} \entails \asubtype{\plus*{\suc\colon \even} < \plus*{\zero\colon \onen ,\, \suc\colon \nat}}
    ;
  \\[-\jot]
  &\asubtype{\onen < \onen} \entails \asubtype{\one < \one}
  \,.
\intertext{Because $\set{\zero,\suc} \subseteq \set{\zero,\suc}$ as well as $\set{\suc} \subseteq \set{\zero,\suc}$, the $\arule{\plus}$ rule then allows us to infer}
  &\asubtype{\even < \nat} \entails \asubtype{\onen < \onen}
  \phantom{;}\enspace\text{\emph{and}}\enspace
  \\[-\jot]
  &\asubtype{\even < \nat} \entails \asubtype{\odd < \nat}
   {;}
   \enspace\text{\emph{as well as}}\enspace\phantom{;}
  \asubtype{\odd < \nat} \entails \asubtype{\even < \nat}
\end{align*}
as necessary consequences of the initial judgments about $\asubtype{\even < \nat}$ and $\asubtype{\odd < \nat}$.
At this point, saturation has been reached: no inference deduces any judgment that has not already been inferred.
Because $\asubtype{\even < \nat} \entails \bot$ has \emph{not} been inferred upon saturation, we may conclude that $\dsubtype{\even{} < \nat{}}$ is derivable -- \textit{i.e.}, that $\even$ is a subtype of $\nat$.
Likewise, we conclude that $\odd$ is a subtype of $\nat$.

\subsection{Correctness of the forward-inference decision procedure}

The forward-inference algorithm is both sound and complete with respect to (the monomorphic fragment of) structural subtyping as defined in \cref{fig:structural-subtyping}. 
The proof of soundness relies on a key \lcnamecref{lem:sound-mono}.%
\begin{lemma}\label{lem:sound-mono}\label{lem:mono-sound}
  Upon saturation:
  \begin{enumerate}
  \item\label{item:lem-sound-mono:named}
    If $\asubtype{t < u} \entails \asubtype{\tau < \sigma}$ and
       $\asubtype{t < u} \nentails \bot$,
    then $\tau \dsub[i] \sigma$.
  \item\label{item:lem-sound-mono:structural}
    If $\asubtype{t < u} \entails \asubtype{A < B}$ and
       $\asubtype{t < u} \nentails \bot$,
    then $A \dsub[i] B$.
  \end{enumerate}
\end{lemma}
\begin{proof}[Proof sketch]
  By mutual coinduction on the derivations of $\dsubtype{\tau{} < \sigma{}}$ and $\dsubtype{A{} < B{}}$.
\end{proof}
\begin{theorem}[Soundness and completeness]\label{thm:mono-sound-complete}
  Upon saturation, $\asubtype{t < u} \nentails \bot$ if and only if $\dsubtype{t{} < u{}}$.
\end{theorem}
\begin{proof}[Proof sketch]
  From left to right, by appealing to \cref{lem:sound-mono};
  from right to left, by induction on the finite derivation of $\asubtype{t < u} \entails \bot$ to establish a (meta-)contradiction.
\end{proof}
\noindent
We do not provide further details of these specific proofs here because this forward-inference procedure for structural subtyping of monomorphic types will be subsumed by the decision procedure for parametric subtyping of polymorphic types that will eventually be presented in \cref{sec:decide-poly}.

The preceding \lcnamecref{thm:mono-sound-complete} establishes that the above forward-inference algorithm is a semi-decision procedure.
However, in this setting, forward inference is, in fact, guaranteed to saturate, making our algorithm a full-fledged decision procedure for structural subtyping of monomorphic types.
\begin{theorem}[Termination]\label{thm:mono-saturation}
  Forward inference according to the rules of \cref{fig:decide-mono} always saturates.
\end{theorem}
\begin{proof}[Proof sketch]
  Finitely many definitions of the form $t \defd A$ and $u \defd B$ can be drawn from a given set $\sig$ of definitions.
  For each such pair of structural types $A$ and $B$, there are finitely many subformulas (without unfolding type definitions).
  Each of the rules found in \cref{fig:decide-mono} infers a judgment $\asubtype{t < u} \entails \asubtype{\tau < \sigma}$ only if either $\tau$ and $\sigma$ are subformulas of $A$ and $B$, respectively, or $\tau$ and $\sigma$ are subformulas of $B$ and $A$, respectively (again, without unfolding definitions).
  Therefore, only finitely many judgments can be inferred, so forward inference must eventually saturate.
\end{proof}

\subsection{Further remarks}\label{subsec:remarks-forward-search}

With respect to a backward-search decision procedure for structural subtyping of monomorphic types (see, \emph{e.g.}, \cite{Lakhani+:ESOP22}), the above forward-inference algorithm has two advantages.
First, it is naturally incremental and compositional:
If additional type definitions are introduced later in the program, inferences involving only prior definitions still hold and need not be performed again;
only inferences involving the newly introduced definitions need to be performed.
Second, the forward-inference algorithm can take advantage of inferences made along one branch when considering another branch.

With respect to the backward-search algorithm, our forward-inference algorithm, as formulated in \cref{fig:decide-mono}, does not account for structural subtypings that arise from uninhabited types that exist when types' interpretation is inductive or mixed inductive/coinductive, such as those in work by \citet{Ligatti+:TOPLAS17} and \citet{Lakhani+:ESOP22}.
In this paper, we choose to work only with types that are interpreted coinductively.
Because all such types are inhabited, the above forward-inference algorithm needs not account for such subtypings.
We conjecture that the algorithm can be extended to inductive and mixed inductive/coinductive settings, but we leave that as future work.

\section{Deciding parametric subtyping for parametric polymorphism}\label{sec:decide-poly}

Leveraging the structure of the forward-inference algorithm for deciding structural subtyping of monomorphic types presented in the preceding \lcnamecref{sec:decide-mono}, we will now present a related algorithm for deciding parametric subtyping of polymorphic types.

At a high level, the algorithm uses saturating forward inference to derive the most general admissible parametric rules for each pair of defined type constructors, such as ``$\dsubtype{{\edyck[\kappa]}{} < {\dyck[\kappa']}{}}$ if $\dsubtype{\kappa{} < {\kappa'}{}}$''.
Then, once these rules have been derived, a parametric subtyping problem can be decided by a second, backward proof construction phase that builds a finite derivation using the rules derived during the first phase.

\subsection{Details of the decision procedure}

As in the special case algorithm for monomorphic types~(\cref{sec:decide-mono}), there are three judgments for necessary consequences of a subtyping relationship between two types involving type constructors.
However, now that type constructors may take parameters, these judgments must account for those parameters.
Also, we choose to explicitly incorporate variances into the judgments for convenience.
Possible variances $\xi$ and $\zeta$ are co- and contravariance, which we write as $\cov$ and $\ctv$, respectively.
(Bivariance is handled as mutual co- and contravariance, and nonvariance is handled implicitly by the algorithm.)
An operation, $\lnot$, on variances, given by $\lnot(\cov) = \ctv$ and $\lnot (\ctv) = \cov$, is also useful.

Because the forward inference judgments will use explicit variances and because we will want to relate them to the declarative characterization of parametric subtyping, it is helpful to define the abbreviation $\dsubtype{\tau{\subs} < {_\xi \sigma}{\subs*}}$ such that:
$\dsubtype{\tau{\subs} < {_{\cov} \sigma}{\subs*}}$ if and only if $\dsubtype{\tau{\subs} < \sigma{\subs*}}$; and $\dsubtype{\tau{\subs} < {_{\ctv} \sigma}{\subs*}}$ if and only if $\dsubtype{\sigma{\subs*} < \tau{\subs}}$.
The abbreviation  $\dsubtype{A{\subs} < {_\xi B}{\subs*}}$ is defined analogously.

We finally arrive at the following three judgments for forward inference.
(We again use $\preccurlyeq$ for distinction.)
The judgment $\asubtype[\xi]{t[\alphas] < u[\betas*]} \entails \bot$ will be inferred if and only if there are no substitutions $\theta$ and $\phi$ and stacks $\subs$ and $\subs*$ for which a derivation of $\dsubtype{{t[\theta]}{\subs} < {_\xi u[\phi]}{\subs*}}$ exists.
The judgments $\asubtype[\xi]{t[\alphas] < u[\betas*]} \entails \asubtype[\xi']{\tau < \sigma}$ and $\asubtype[\xi]{t[\alphas] < u[\betas*]} \entails \asubtype[\xi']{A < B}$ will be inferred if and only if $\dsubtype{{\tau}{\theta ; \subs} < {_{\xi'} \sigma}{\phi ; \subs*}}$ and $\dsubtype{A{\theta ; \subs} < {_{\xi'} B}{\phi ; \subs*}}$, respectively, would necessarily occur as subderivations of any derivation of $\dsubtype{{t[\theta]}{\subs} < {_\xi u[\phi]}{\subs*}}$ (assuming such a derivation exists).

\subsubsection{Phase 1: Forward inference}

Forward inference proceeds according to the rules found in \cref{fig:decide-poly}.
\begin{figure}
  \small
  \begin{gather*}
    \infer[\arule{\jrule{INIT-}}]{\asubtype[\xi]{t[\alphas] < u[\betas]} \entails \asubtype[\xi]{A < B}}{
      t[\alphas] \defd A & u[\betas] \defd B}
    \\[\figinferskip]
    \infer[\arule{\tensor}]{\asubtype[\xi]{t[\alphas] < u[\betas]} \entails \asubtype[\xi']{\tau_i < \sigma_i}}{
      \asubtype[\xi]{t[\alphas] < u[\betas]} \entails \asubtype[\xi']{\tau_1 \tensor \tau_2 < \sigma_1 \tensor \sigma_2} &
      (i \in \set{1,2})}
    \quad
    \raisebox{0.5\baselineskip}{(no $\arule{\one}$ rule for $\asubtype[\xi]{t[\alphas] < u[\betas]} \entails \asubtype[\xi']{\one < \one}$)}
    \\[\figinferskip]
    \infer[\arule{\plus}]{\asubtype[\xi]{t[\alphas] < u[\betas]} \entails \asubtype[\xi']{\tau_\ell < \sigma_\ell}}{
      \asubtype[\xi]{t[\alphas] < u[\betas]} \entails \asubtype[\xi']{\plus*[\ell \in L]{\ell\colon \tau_\ell} < \plus*[k \in K]{k\colon \sigma_k}} &
      (L \subseteq_{\xi'} K) & (\ell \in L \cap K)}
    \\[\figinferskip]
    \infer[\arule{\imp}_1]{\asubtype[\xi]{t[\alphas] < u[\betas]} \entails \asubtype[\lnot \xi']{\tau_1 < \sigma_1}}{
      \asubtype[\xi]{t[\alphas] < u[\betas]} \entails \asubtype[\xi']{\tau_1 \imp \tau_2 < \sigma_1 \imp \sigma_2}}
    \qquad
    \infer[\arule{\imp}_2]{\asubtype[\xi]{t[\alphas] < u[\betas]} \entails \asubtype[\xi']{\tau_2 < \sigma_2}}{
      \asubtype[\xi]{t[\alphas] < u[\betas]} \entails \asubtype[\xi']{\tau_1 \imp \tau_2 < \sigma_1 \imp \sigma_2}}
    \\[\figinferskip]
    \infer[\arule{\with}]{\asubtype[\xi]{t[\alphas] < u[\betas]} \entails \asubtype[\xi']{\tau_k < \sigma_k}}{
      \asubtype[\xi]{t[\alphas] < u[\betas]} \entails \asubtype[\xi']{\with*[\ell \in L]{\ell\colon \tau_\ell} < \with*[k \in K]{k\colon \sigma_k}} &
      (K \subseteq_{\xi'} L) & (k \in L \cap K)}
    \\[\figinferskip]
    \infer[\arule{\forall}]{\asubtype[\xi]{t[\alphas] < u[\betas]} \entails \asubtype[\xi']{[z/x]\tau < [z/y]\sigma}}{
      \asubtype[\xi]{t[\alphas] < u[\betas]} \entails \asubtype[\xi']{\forall x.\tau < \forall y.\sigma} &
      (\text{$z$ fresh})}
    \qquad\!\!
    \infer[\arule{\exists}]{\asubtype[\xi]{t[\alphas] < u[\betas]} \entails \asubtype[\xi']{[z/x]\tau < [z/y]\sigma}}{
      \asubtype[\xi]{t[\alphas] < u[\betas]} \entails \asubtype[\xi']{\exists x.\tau < \exists y.\sigma} &
      (\text{$z$ fresh})}
    \\[\figinferskip]
    \infer[\arule{\jrule{COMPOSE-}}]{\asubtype[\xi]{t[\alphas] < u[\betas]} \entails \asubtype[\zeta]{\theta'(\alpha') < \phi'(\beta')}}{
      \asubtype[\xi]{t[\alphas] < u[\betas]} \entails \asubtype[\xi']{t'[\theta'] < u'[\phi']} &
      \asubtype[\xi']{t'[\alphas'] < u'[\betas']} \entails \asubtype[\zeta]{\alpha' < \beta'}}
    \\[\figinferskip]
    \infer[\arule{\jrule{COMPOSE-}}_\bot]{\asubtype[\xi]{t[\alphas] < u[\betas]} \entails \bot}{
      \asubtype[\xi]{t[\alphas] < u[\betas]} \entails \asubtype[\xi']{t'[\theta'] < u'[\phi']} &
      \asubtype[\xi']{t'[\alphas'] < u'[\betas']} \entails \bot}
    \\[\figinferskip]
    \infer[\!\arule{\jrule{PARAM-L-}}_{\bot}]{\asubtype[\xi]{t[\alphas] < u[\betas]} \entails \bot}{
      \asubtype[\xi]{t[\alphas] < u[\betas*]} \entails \asubtype[\xi']{\alpha < \sigma} & \!\!(\sigma \neq \beta)}
    \qquad\!\!\!\!\!
    \infer[\!\arule{\jrule{PARAM-R-}}_{\bot}]{\asubtype[\xi]{t[\alphas] < u[\betas]} \entails \bot}{
      \asubtype[\xi]{t[\alphas] < u[\betas*]} \entails \asubtype[\xi']{\tau < \beta} & \!\!(\tau \neq \alpha)}
    \\[\figinferskip]
    \text{(no rule for $\asubtype[\xi]{t[\alphas] < u[\betas]} \entails \asubtype[\xi']{x < x}$)}
    \\[\figinferskip]
    \infer[\arule{\jrule{VAR-L-}}_\bot]{\asubtype[\xi]{t[\alphas] < u[\betas]} \entails \bot}{
      \asubtype[\xi]{t[\alphas] < u[\betas]} \entails \asubtype[\xi']{x < \tau} &
      (\tau \neq x)}
    \qquad
    \infer[\arule{\jrule{VAR-R-}}_\bot]{\asubtype[\xi]{t[\alphas] < u[\betas]} \entails \bot}{
      \asubtype[\xi]{t[\alphas] < u[\betas]} \entails \asubtype[\xi']{\tau < x} &
      (\tau \neq x)}
    \\[\figinferskip]
    \infer[\arule{\plus\plus}_\bot]{\asubtype[\xi]{t[\alphas] < u[\betas]} \entails \bot}{
      \asubtype[\xi]{t[\alphas] < u[\betas]} \entails \asubtype[\xi']{\plus*[\ell \in L]{\ell\colon \tau_\ell} < \plus*[k \in K]{k\colon \sigma_k}} &
      (L \nsubseteq_{\xi'} K)}
    \\[\figinferskip]
    \infer[\arule{\with\with}_\bot]{\asubtype[\xi]{t[\alphas] < u[\betas]} \entails \bot}{
      \asubtype[\xi]{t[\alphas] < u[\betas]} \entails \asubtype[\xi']{\with*[\ell \in L]{\ell\colon \tau_\ell} < \with*[k \in K]{k\colon \sigma_k}} &
      (K \nsubseteq_{\xi'} L)}
    \\[\figinferskip]
    \infer[\arule{\jrule{MISMATCH-}}_\bot]{\asubtype[\xi]{t[\alphas] < u[\betas]} \entails \bot}{
      \asubtype[\xi]{t[\alphas] < u[\betas]} \entails \asubtype[\xi']{A < B} &
      A \mathrel{\Bot} B}
    \hspace{-4em}
  \end{gather*}
  \caption{Forward inference rules for phase 1 of deciding parametric subtyping for parametric polymorphism ({\normalfont\scshape f} for `forward').
These rules are interpreted inductively.
The notation $A \mathrel{\Bot} B$ means that $A$ and $B$ use distinct top-level structural type constructors, such as $\plus$ and $\one$.
Also, $L \subseteq_{\cov} K$ iff $L \subseteq K$; and $L \subseteq_{\ctv} K$ iff $L \supseteq K$.}\label{fig:decide-poly}
\end{figure}%
Once again, these rules are interpreted inductively and are more clearly read top-down, from premises to conclusion.
Many of the rules are carried over from the decision procedure for structural subtyping of monomorphic types that was described in \cref{fig:decide-mono} (page~\pageref{fig:decide-mono}, \cref{sec:decide-mono}), with the addition of parameters and variances.
For example, the essential aspects of the $\arule{\tensor}$, $\arule{\plus}$, and $\arule{\jrule{MISMATCH-}}_\bot$ rules are unchanged from \cref{fig:decide-mono}.
We will detail a few of the other rules.

\paragraph{The $\arule{\jrule{COMPOSE-}}$ rule}
The most important difference between the rules of \cref{fig:decide-poly} and those of \cref{fig:decide-mono} is that it is now possible to infer judgments of the form $\asubtype[\xi]{t[\alphas] < u[\betas*]} \entails \asubtype[\xi']{\alpha < \beta}$.
These represent constraints that must hold of any instantiation $\asubtype[\xi]{t[\theta] < u[\phi]}$ of type constructors $t$ and $u$.
This idea is captured in the $\arule{\jrule{COMPOSE-}}$ rule:
Suppose that the first premise, $\asubtype[\xi]{t[\alphas] < u[\betas*]} \entails \asubtype[\xi']{t'[\theta'] < u'[\phi']}$, has already been inferred -- that is, that any derivation of $\dsubtype{{t[\theta]}{\subs} < {_\xi u[\phi]}{\subs*}}$ would necessarily contain a subderivation of $\dsubtype{{t'[\theta']}{\theta ; \subs} < {_{\xi'} u'[\phi']}{\phi ; \subs*}}$.
Furthermore, suppose that the premise $\asubtype[\xi']{t'[\alphas'] < u'[\betas*']} \entails \asubtype[\zeta]{\alpha' < \beta'}$ has already been inferred -- that is, that any derivation of $\dsubtype{{t'[\theta']}{\theta ; \subs} < {_{\xi'} u'[\phi']}{\phi ; \subs*}}$ would necessarily contain a subderivation of $\dsubtype{{\alpha'}{\theta' ; (\theta ; \subs)} < {_\zeta \beta'}{\phi' ; (\phi ; \subs*)}}$.
It then follows from the $\prule{\jrule{PARAM-}}$ rule and transitivity of containment that $\dsubtype{{\theta'(\alpha')}{\theta ; \subs} < {_\zeta \phi'(\beta')}{\phi ; \subs*}}$ would necessarily occur as a subderivation of $\dsubtype{{t[\theta]}{\subs} < {_\xi u[\phi]}{\subs*}}$, thereby justifying the
inference of $\asubtype[\xi]{t[\alphas] < u[\betas*]} \entails \asubtype[\zeta]{\theta'(\alpha') < \phi'(\beta')}$.

\paragraph{The $\arule{\jrule{PARAM-L-}}_\bot$ and $\arule{\jrule{PARAM-R-}}_\bot$ rules}
Suppose that the premise $\asubtype[\xi]{t[\alphas] < u[\betas*]} \entails \asubtype[\xi']{\alpha < \sigma}$ of the $\arule{\jrule{PARAM-L-}}_\bot$ rule has already been inferred, with $\sigma$ not a parameter $\beta$.
That is, any derivation of $\dsubtype{{t[\theta]}{\subs} < {_\xi u[\phi]}{\subs*}}$ would necessarily contain a subderivation of $\dsubtype{{\alpha}{\theta ; \subs} < {_{\xi'} \sigma}{\phi ; \subs*}}$.
However, because $\sigma$ is not a parameter, there is no rule in \cref{fig:parametric-subtyping} that could have derived that subderivation.
Therefore, no derivation of $\dsubtype{{t[\theta]}{\subs} < {_\xi u[\phi]}{\subs*}}$ can exist, thereby justifying the inference of $\asubtype[\xi]{t[\alphas] < u[\betas*]} \entails \bot$ by the $\arule{\jrule{PARAM-L-}}_\bot$ rule.
The $\arule{\jrule{PARAM-R-}}_\bot$ rule is analogous, as are the $\arule{\jrule{VAR-L-}}_\bot$ and $\arule{\jrule{VAR-R-}}_\bot$ rules.

\subsubsection{Example}
Returning to our running example of context-free languages, we can examine the inferences made by our algorithm to infer admissible subtyping rules for pairs of type constructors.  
By virtue of the $\arule{\jrule{INIT-}}$ rule, the following judgments, among others, will be inferred.
\begin{alignat*}{2}
  &(1)\enspace &&\asubtype[\cov]{\edyck[\kappa] < \dyck[\kappa']} \entails \asubtype[\cov]{\plus*{\leftn\colon \edyck[\name{r}[\kappa]] ,\, \rightn\colon \kappa} < \plus*{\leftn\colon \dyck[\dyck[\kappa']] ,\, \rightn\colon \kappa'}}
  \\
  &(2) &&\asubtype[\cov]{\name{r}[\kappa] < \dyck[\kappa']} \entails \asubtype[\cov]{\plus*{\rightn\colon \kappa} < \plus*{\leftn\colon \dyck[\dyck[\kappa']] ,\, \rightn\colon \kappa'}}
  \\
  &(3) &&\asubtype[\cov]{\edyck_0 < \dyck_0} \entails \asubtype[\cov]{\plus*{\leftn\colon \edyck[\name{end}] ,\, \$\colon \onen} < \plus*{\leftn\colon \dyck[\dyck_0] ,\, \$\colon \onen}}
  \\
  &(4) &&\asubtype[\cov]{\name{end} < \dyck_0} \entails \asubtype[\cov]{\plus*{\$\colon \onen} < \plus*{\leftn\colon \dyck[\dyck_0] ,\, \$\colon \onen}}
  \\
  &(5) &&\asubtype[\cov]{\onen < \onen} \entails \asubtype[\cov]{\one < \one}
\end{alignat*}
Because $\set{\rightn} \subseteq \set{\leftn,\rightn} \subseteq \set{\leftn, \rightn}$ and $\set{\$} \subseteq \set{\leftn,\$} \subseteq \set{\leftn,\$}$, the $\arule{\plus}$ rule then allows us to infer
\begin{align*}
  (6) \enspace &\asubtype[\cov]{\edyck[\kappa] < \dyck[\kappa']} \entails \asubtype[\cov]{\edyck[\name{r}[\kappa]] < \dyck[\dyck[\kappa']]}
  &
  (9) \enspace & \asubtype[\cov]{\edyck_0 < \dyck_0} \entails \asubtype[\cov]{\edyck[\name{end}] < \dyck[\dyck_0]}
  \\
  (7) \enspace &\asubtype[\cov]{\edyck[\kappa] < \dyck[\kappa']} \entails \asubtype[\cov]{\kappa < \kappa'}
  & \quad
  (10) \enspace &\asubtype[\cov]{\edyck_0 < \dyck_0} \entails \asubtype[\cov]{\onen < \onen}
  \\
  (8) \enspace &\asubtype[\cov]{\name{r}[\kappa] < \dyck[\kappa']} \entails \asubtype[\cov]{\kappa < \kappa'}
  &
  (11) \enspace &\asubtype[\cov]{\name{end} < \dyck_0} \entails \asubtype[\cov]{\onen < \onen}
\end{align*}
as necessary consequences of the initial judgments.
The $\arule{\jrule{COMPOSE-}}$ rule can be applied to $(6)$ and $(7)$, as well as to $(9)$ and $(7)$, to infer
\begin{align*}
  (12) \enspace &\asubtype[\cov]{\edyck[\kappa] < \dyck[\kappa']} \entails \asubtype[\cov]{\name{r}[\kappa] < \dyck[\kappa']}
  \\
  (13) \enspace &\asubtype[\cov]{\edyck_0 < \dyck_0} \entails \asubtype[\cov]{\name{end} < \dyck_0}
    \,.
\end{align*}
(The $\arule{\jrule{COMPOSE-}}$ rule could also be applied to $(12)$ and $(8)$ to infer $\asubtype[\cov]{\edyck[\kappa] < \dyck[\kappa']} \entails \asubtype[\cov]{\kappa < \kappa'}$, but that has already been inferred as $(7)$.)
At this point, saturation has been reached for all pairs of constructors above.
Because no such pair has had $\bot$ inferred as a consequence by the time saturation occurs, admissible subtyping rules for all such pairs do exist.
Collecting the respective atomic constraints, namely $(7)$ and $(8)$, we see that these admissible rules are
\begin{equation*}
  \infer{\asubtype[\cov]{\edyck[\kappa] < \dyck[\kappa']}}{
    \asubtype[\cov]{\kappa < \kappa'}}
  \:,\quad
  \infer{\asubtype[\cov]{\name{r}[\kappa] < \dyck[\kappa']}}{
    \asubtype[\cov]{\kappa < \kappa'}}
  \:,\quad
  \infer{\asubtype[\cov]{\edyck_0 < \dyck_0}}{}
  \:,\quad
  \infer{\asubtype[\cov]{\name{end} < \dyck_0}}{}
  \:,\enspace\text{\emph{and}}\quad
  \infer{\asubtype[\cov]{\onen < \onen}}{}
  \:.
\end{equation*}

\subsubsection{Phase 2: Backward proof construction}

Having inferred the most general admissible parametric rule for each pair of type constructors, the second, backward proof construction phase begins.
For this phase, we introduce a judgment $\asubtype[\xi]{\tau{} < \sigma{}}$.
Given a parametric subtyping problem $\asubtype[\xi]{\tau{} < \sigma{}}$, the types $\tau$ and $\sigma$ are examined, searching for a derivation according to the rules of \cref{fig:decide-poly:backward}. 
\begin{figure}
  \small
  \begin{gather*}
    \infer[\brule{\jrule{COMPOSE-}}]{\asubtype[\xi]{t[\theta] < u[\phi]}}{
      \asubtype[\xi]{t[\alphas] < u[\betas]} \nentails \bot &
      \forall (\asubtype[\xi]{t[\alphas] < u[\betas]} \entails \asubtype[\zeta]{\alpha < \beta})\colon
        \asubtype[\zeta]{\theta(\alpha) < \phi(\beta)}}
    \qquad\!\!
    \infer[\brule{\jrule{VAR-}}]{\asubtype[\xi]{x < x}}{}
  \end{gather*}
  \caption{Backward proof construction rules for phase 2 of deciding parametric subtyping for parametric polymorphism ({\normalfont\scshape b} for `backward').
These rules are interpreted inductively.}\label{fig:decide-poly:backward}
\end{figure}

Because parameters can only appear free within type definitions, neither $\tau$ nor $\sigma$ can be a parameter.
Also, a type variable is a subtype of only itself, so there are no rules for $\asubtype[\xi]{x < u[\phi]}$ nor $\asubtype[\xi]{t[\theta] < x}$.
If $\tau$ is $t[\theta]$ and $\sigma$ is $u[\phi]$, then the algorithm checks that $\asubtype[\xi]{t[\alphas] < u[\betas*]} \entails \bot$ has \emph{not} been inferred during the forward-inference phase.
Provided that is the case, then all atomic constraints $\asubtype[\xi]{t[\alphas] < u[\betas*]} \entails \asubtype[\zeta]{\alpha < \beta}$ are gathered.
Backward proof construction continues by building derivations of each $\asubtype[\zeta]{\theta(\alpha) < \phi(\beta)}$ that corresponds to an atomic constraint on $t[\alphas]$ and $u[\betas*]$.
Backward proof construction terminates here because the types become smaller.

\subsubsection{Example}
Given the admissible rules for the context-free languages example that were inferred in the first phase of the algorithm, we may conclude that $\dsubtype{{\edyck_0}{} < {\dyck_0}{}}$ holds.
Moreover, backward proof construction over these admissible rules can be used to decide other subtyping problems.
For example, $\asubtype[\cov]{\edyck[\edyck[\edyck_0]] < \dyck[\dyck[\dyck_0]]}$ can be confirmed when backward proof construction builds the following derivation.
\begin{equation*}
  \infer[\brule{\jrule{COMPOSE-}}]{\asubtype[\cov]{\edyck[\edyck[\edyck_0]] < \dyck[\dyck[\dyck_0]]}}{
    \infer[\brule{\jrule{COMPOSE-}}]{\asubtype[\cov]{\edyck[\edyck_0] < \dyck[\dyck_0]}}{
      \infer[\brule{\jrule{COMPOSE-}}]{\asubtype[\cov]{\edyck_0 < \dyck_0}}{}}}
  \hspace{-3em}
\end{equation*}

\subsection{Correctness of the decision procedure for parametric subtyping}\label{sec:decide-poly:correctness}

The algorithm described above is both sound and complete with respect to the declarative characterization of parametric subtyping given in \cref{fig:parametric-subtyping}.
We give only sketches of the proofs here; details can be found in the extended version of this paper~\cite{DeYoung+:arXiv23}.

Following the pattern laid out for monomorphic types, the proof of soundness relies on the following key lemma that generalizes \cref{lem:mono-sound}.
\begin{restatable}{lemma}{polysoundlem}\label{lem:poly-sound}
  Given a saturated database:
  \begin{enumerate}
  \item If $\asubtype[\xi]{t[\alphas] < u[\betas]} \entails \asubtype[\xi']{\tau < \sigma}$; $\asubtype[\xi]{t[\alphas] < u[\betas]} \nentails \bot$; and $\dsubtype[i]{\alpha{\subs} < {_\zeta \beta}{\subs*}}$ for each $\asubtype[\xi]{t[\alphas] < u[\betas]} \entails \asubtype[\zeta]{\alpha < \beta}$; then $\dsubtype[i]{\tau{\subs} < {_{\xi'} \sigma}{\subs*}}$.

  \item If $\asubtype[\xi]{t[\alphas] < u[\betas]} \entails \asubtype[\xi']{A < B}$; $\asubtype{t[\alphas] < u[\betas]}{\xi} \nentails \bot$; and $\dsubtype[i]{\alpha{\subs} < {_\zeta \beta}{\subs*}}$ for each $\asubtype[\xi]{t[\alphas] < u[\betas]} \entails \asubtype[\zeta]{\alpha < \beta}$;  then $\dsubtype[i]{A{\subs} < {_{\xi'} B}{\subs*}}$.
  \end{enumerate}
\end{restatable}
\begin{proof}[Proof sketch]
  By mutual coinduction on the (potentially infinite) derivations of $\dsubtype[i]{\tau{\subs} < {_{\xi'} \sigma}{\subs*}}$ and $\dsubtype[i]{A{\subs} < {_{\xi'} B}{\subs*}}$.
\end{proof}
\noindent
Soundness then follows by structural induction on the derivation using the rules of \cref{fig:decide-poly:backward} that was built by backward proof construction.
\begin{restatable}[Soundness]{theorem}{polysoundthm}\label{thm:sound-poly}
  If\/ $\asubtype[\xi]{\tau < \sigma}$,
  then $\dsubtype{\tau{\subs} < {_\xi \sigma}{\subs*}}$ for all stacks $\subs$ and $\subs*$.
\end{restatable}
\noindent
Completeness also requires a lemma, but then follows by structural induction on the type $\tau$.
\begin{restatable}{lemma}{polycompletelem}\label{lem:poly-complete}\leavevmode
  \begin{enumerate}[noitemsep]
  \item If $\dsubtype[i]{{t[\theta]}{\subs} < {_\xi u[\phi]}{\subs*}}$ and  $\asubtype[\xi]{t[\alphas] < u[\betas]} \entails \asubtype[\xi']{A < B}$, then $\dsubtype{A{\theta ; \subs} < {_{\xi'} B}{\phi ; \subs*}}$.
  \item If $\dsubtype[i]{{t[\theta]}{\subs} < {_\xi u[\phi]}{\subs*}}$ and $\asubtype[\xi]{t[\alphas] < u[\betas]} \entails \asubtype[\xi']{\tau < \sigma}$, then $\dsubtype{\tau{\theta ; \subs} < {_{\xi'} \sigma}{\phi ; \subs*}}$.
  \item If $\dsubtype[i]{{t[\theta]}{\subs} < {_\xi u[\phi]}{\subs*}}$, then $\asubtype[\xi]{t[\alphas] < u[\betas]} \nentails \bot$.
  \end{enumerate}
\end{restatable}
\begin{proof}[Proof sketch]
  Each part is proved as follows.
  \begin{enumerate}
  \item Directly, noting that only the $\arule{\jrule{INIT-}}$ rule can derive $\asubtype[\xi]{t[\alphas] < u[\betas]} \entails \asubtype[\xi']{A < B}$.
  \item By induction on the finite derivation of $\asubtype[\xi]{t[\alphas] < u[\betas]} \entails \asubtype[\xi']{\tau < \sigma}$.
  \item We generalize the lemma to show that $\dsubtype[i]{{t[\theta]}{\subs} < {_\xi u[\phi]}{\subs*}}$ and $\asubtype[\xi]{t[\alphas] < u[\betas]} \entails \bot$ together imply a meta-contradiction.
    This is proved by induction on the finite derivation of $\asubtype[\xi]{t[\alphas] < u[\betas]} \entails \bot$.
  \qedhere
  \end{enumerate}
\end{proof}

\begin{restatable}[Completeness]{theorem}{polycomplete}\label{thm:complete-poly}
  If $\dsubtype{\tau{\subs} < {_\xi \sigma}{\subs*}}$, then $\asubtype[\xi]{\applysubs{\subs}{\tau} < \applysubs{\subs*}{\sigma}}$.
  As a particular case, $\dsubtype{\tau{\subse} < {_\xi \sigma}{\subse}}$ implies $\asubtype[\xi]{\tau < \sigma}$.
\end{restatable}
\noindent
We must also prove that forward inference and backward proof construction terminate.
\begin{theorem}[Termination]
  Forward inference and backward proof construction according to the rules of \cref{fig:decide-poly,fig:decide-poly:backward}, respectively, terminate.
\end{theorem}
\begin{proof}[Proof sketch]
  Finitely many definitions $t[\alphas] \defd A$ and $u[\betas*] \defd B$ can be drawn from a given set $\sig$ of definitions and combined with one of two variances, so the $\arule{\jrule{INIT-}}$ rule infers only finitely many judgments $\asubtype[\xi]{t[\alphas] < u[\betas*]} \entails \asubtype[\xi]{A < B}$.
  By induction, we can show that each of the other rules found in \cref{fig:decide-poly}, including the $\arule{\jrule{COMPOSE-}}$ rule, will infer a judgment $\asubtype[\xi]{t[\alphas] < u[\betas*]} \entails \asubtype[\xi']{\tau < \sigma}$ only if $\tau$ and $\sigma$ are proper subformulas of $A$ and $B$.
  Moreover, $A$ and $B$ have finitely many subformulas (without unfolding definitions).

  But special care needs to be taken with the $\arule{\forall}$ and $\arule{\exists}$ rules.
  Because there is an infinite supply of fresh variables, it might seem like these rules could be applied to a given judgment, such as $\asubtype[\xi]{t[\alphas] < u[\betas*]} \entails \asubtype[\xi']{\forall x.\tau < \forall y.\sigma}$, an infinite number of times, inferring judgments $\asubtype[\xi]{t[\alphas] < u[\betas*]} \entails \asubtype[\xi']{[z/x]\tau < [z/y]\sigma}$ and $\asubtype[\xi]{t[\alphas] < u[\betas*]} \entails \asubtype[\xi']{[z'/x]\tau < [z'/y]\sigma}$ and so on.
  However, we treat inferred judgments as equivalent up to $\alpha$-renaming of their free type variables, which are implicitly universally quantified over the judgment.
  Therefore, up to $\alpha$-renaming, the $\arule{\forall}$ and $\arule{\exists}$ rules only infer a single judgment per distinct premise.
  This kind of subsumption of $\alpha$-equivalent judgments is a standard assumption in resolution-based saturation procedures~\cite{Robinson:JACM65}.

  Because only finitely many judgments can be inferred, forward inference in phase~1 must terminate.
  Backward proof construction in phase~2 terminates because its recursion occurs at successive subformulas.
\end{proof}
\noindent
Saturating forward inference facilitates sharing of inferred subtyping constraints across branches.
Because of issues with non-regular type constructors (see \cref{sec:structural:examples}), this is essential to capture non-regularity.
It is thus unclear if the algorithm could be recast in a recursive functional way.

\subsection{Nonrecursive type abbreviations}

Consider the type constructors $t[\alpha] \defd \alpha \tensor \one$ and $u[\beta] \defd \one \tensor \beta$.
Purely structurally, $t[\one]$ would be a subtype of $u[\one]$, with both $t[\one]$ and $u[\one]$ unfolding to $\one \tensor \one$ because type constructors are treated transparently by the $\srule{\jrule{UNF-}}$ rule.
However, when constructors are treated parametrically, this relationship does \emph{not} hold: $\dsubtype{\alpha{\one/\alpha ; \subs} < \one{\one/\beta ; \subs*}}$ and $\dsubtype{\one{\one/\alpha ; \subs} < \beta{\one/\beta ; \subs*}}$ violate parametricity.

If we would like certain definitions to act as mere abbreviations that are (conceptually) always expanded, then our system can easily accommodate this as long as those definitions are not recursive.
To integrate such nonrecursive type abbreviations into the declarative characterization of parametric subtyping, we could add the following rule and restrict the $\prule{\jrule{INST-}}$ to apply only when neither type constructor is a nonrecursive abbreviation.
\begin{equation*}
  \infer[\prule{\jrule{UNF-}}]{\dsubtype{{t[\theta]}{\subs} < {u[\phi]}{\subs*}}}{
    t[\alphas] \defd A & u[\betas] \defd B &
    (\text{$t$ or $u$ is nonrecursive abbrev.}) &
    \dsubtype{{\theta(A)}{\subs} < {\phi(B)}{\subs*}}}
\end{equation*}
Forward inference in phase~1 and backward proof construction in phase~2 of the decision procedure would similarly expand nonrecursive abbreviations.
It is easy to see that both phases still terminate.

\section{Implementation}\label{sec:implementation}

Our implementation in Standard ML is available in a virtual machine image~\cite{DeYoung+:Zenodo23}, and its source is available in an online repository~\cite{DeYoung+:Bitbucket23}. It provides
a syntax for defining types, deriving parametric inference
rules, and checking subtyping.  All examples from this paper are
available in the file \verb|examples/paper.poly| in the VM image and source repository.  Because of mutual recursion, the
implementation proceeds in two phases: first checking basic
consistency and normalizing types, thereby introducing
additional, internal definitions.  The second
phase executes the saturation algorithm 
(\cref{sec:decide-poly}) and answers queries by consulting the saturated
database.
There are only the following few minor points of departure from this paper.

\subsection{Nonrecursive type abbreviations.}\label{sec:implementation:nonrecursive}
The implementation allows
\emph{type definitions} (which are always treated parametrically) but also
explicit \emph{type abbreviations} which may be parameterized but must be nonrecursive.
These abbreviations are expanded structurally, as described above.
Depending on the larger language context, an alternative would be to treat
every nonrecursive type definition as an abbreviation.

\subsection{Elaboration to normal form.}  Before elaboration, every
definition in $\Sigma$ has the form $t[\alphas] \defd A$
where $A$ is structural (which
guarantees contractiveness) but may not be in normal form.  We map each
such definition to $t[\alphas] \defd A^*$, using the auxiliary
translation $A^\dagger$ in which $A$ need not be structural.  In the
process, we may introduce further, internal definitions.
\begin{equation*}
  \begin{aligned}
    (A_1 \tensor A_2)^* &= A_1^\dagger \tensor A_2^\dagger
    \\[-\jot]
    (\one)^* &= \one
    \\[-\jot]
    \smash{(\plus*[\ell \in L]{\ell\colon A_\ell})^*}\vphantom{A_\ell^\dagger} &= \smash{\plus*[\ell \in L]{\ell\colon A_\ell^\dagger}}\vphantom{A_\ell^\dagger}
    \\[-\jot]
    (A_1 \imp A_2)^* &= A_1^\dagger \imp A_2^\dagger
    \\[-\jot]
    (\with*[\ell \in L]{\ell\colon A_\ell})^* &= \with*[\ell \in L]{\ell\colon A_\ell^\dagger}
    \\[-\jot]
    (\forall x.A)^* &=
      \forall x.\, [x/\alpha]([\alpha/x]A)^\dagger \text{, with $\alpha$ fresh}
    \\[-\jot]
    (\exists x.A)^* &=
      \exists x.\, [x/\alpha]([\alpha/x]A)^\dagger \text{, with $\alpha$ fresh}
  \end{aligned}
  \quad
  \begin{aligned}
    (t[\theta])^\dagger &= t[\theta^\dagger]
    \\[-\jot]
    (\alpha)^\dagger &= \alpha
    \\[-\jot]
    (A)^\dagger &= t[\alphas] \begin{lgathered}[t]
        \text{ for $A$ structural,} \\[-\jot]
        \text{ where $\alphas = \mathsf{free}(A)$} \\[-\jot]
        \text{ and $\Sigma \coloneqq \Sigma, t[\alphas] \defd A^*$}
      \end{lgathered}
  \end{aligned}
\end{equation*}
Here, $\theta^\dagger$ is defined pointwise.  Computing the free type
parameters $\mathsf{free}(A)$
avoids creating internal definitions with unnecessary parameters.
Also, notice that, during elaboration, quantified type variables $x$ become parameters $\alpha$ that are then instantiated with the corresponding variable $x$ after normalization.
Therefore, no case for $(x)^\dagger$ is needed.
We can also obtain additional sharing (and therefore faster convergence of the
saturation algorithm) in the clause for $(A)^\dagger$ by reusing a
definition $u[\alphas]$ if $u[\alphas] \defd A^*$ is already in the
definitions $\Sigma$ (modulo renaming of the parameters).

\subsection{Indexing in the database.}  In the implementation, we
combine all facts
$\asubtype[\xi]{t[\alphas] < u[\betas*]} \entails \dotsb$ for a given
$t$, $u$, and $\xi$ into a single entry
$\asubtype[\xi]{t[\alphas] < u[\betas*]} \entails \mathcal{C}$, where
constraints $\mathcal{C}$ are given by 
\begin{alignat*}{2}
  \text{\emph{Constraints}} & \quad &
   \mathcal{C} &\Coloneqq \bot \mid \asubtype[\zeta]{\alpha_i < \beta_j} \mid \top \mid \mathcal{C}_1 \land \mathcal{C}_2
  \,.
\end{alignat*}
We keep constraints $\mathcal{C}$ in a normal form where no entries are
repeated and any conjunction with $\bot$ is reduced to $\bot$ alone.  This
facilitates efficient lookup and detection of saturation.

\section{Examples}\label{sec:examples}

\subsection{Lists}

The polymorphic type of lists of elements of a type $\alpha$, as well as the types of empty and nonempty lists, can be defined as follows.
\begin{alignat*}{3}
  &&\qquad&&
  \elistn &\defd \plus*{\nil\colon \one}
  \\
  \listn[\alpha] &\defd \plus*{\nil\colon \one ,\, \cons\colon \alpha \tensor \listn[\alpha]}
  &&&
  \nelistn[\alpha] &\defd \plus*{\cons\colon \alpha \tensor \listn[\alpha]}
\end{alignat*}
By running our saturation algorithm and examining the atomic constraints on $\listn[\alpha]$ and itself, we can verify that $\listn[\alpha]$ is covariant in $\alpha$; we can similarly verify that $\nelistn[\alpha]$ is covariant in $\alpha$.
Moreover, our algorithm confirms the parametric rules $\dsubtype{\elistn{} < {\listn[\alpha]}{}}$, as well as $\dsubtype{{\nelistn[\alpha]}{} < {\listn[\beta]}{}}$ if $\dsubtype{\alpha{} < \beta{}}$.
And lists of nonempty lists of even natural numbers are, more generally, lists of lists of natural numbers, and our algorithm confirms that $\dsubtype{{\listn[\nelistn[\even]]}{} < {\listn[\listn[\nat]]}{}}$, for example.

\subsection{Binary trees and spines}

\subsubsection{Trees}
Similarly to lists, the polymorphic type of binary trees of $\beta$s can be defined as follows.
\begin{equation*}
  \tree[\beta] \defd \plus*{\leaf\colon \one ,\, \node\colon \beta \tensor (\tree[\beta] \tensor \tree[\beta])}
\end{equation*}
Thus, a tree of $\beta$s is either $\leaf$ with the unit value, or $\node$ with a tuple of an element of type $\beta$ and the left and right subtrees.
Our saturation algorithm verifies that $\tree[\beta]$ is covariant in $\beta$.

Similar to those for lists, the types of empty and nonempty trees of $\beta$s are subtypes of $\tree[\beta]$:
\begin{align*}
  \name{etree} &\defd \plus*{\leaf\colon \one}
  \\
  \name{netree}[\beta] &\defd \plus*{\node\colon \beta \tensor (\tree[\beta] \tensor \tree[\beta])}
\end{align*}

\subsubsection{Spines}\label{sec:examples:trees:spines}
Because the left and right spines of a tree are essentially lists, we might at first expect to have $\dsubtype{{\listn[\tau]}{} < {\tree[\sigma]}{}}$ whenever $\dsubtype{\tau{} < \sigma{}}$.
However, even if we were to coordinate the label names across the two types, that subtyping relationship would \emph{not} hold because it would require $\dsubtype{{\listn[\alpha]}{\tau/\alpha} < {(\tree[\beta] \tensor \tree[\beta])}{\sigma/\beta}}$ -- that is, it would require a (nonempty) sum type to be a subtype of a product type.
Indeed, saturation 
yields $\asubtype[\cov]{\listn[\alpha] < \tree[\beta]} \entails \bot$ for that reason.

Instead, we could define a type of left spines as follows;
right spines would be symmetric.
\begin{align*}
  \lspine[\alpha] &\defd \plus*{\leaf\colon \one ,\, \node\colon \alpha \tensor (\lspine[\alpha] \tensor \name{etree})}
\end{align*}
With this definition, we \emph{do} have $\dsubtype{{\lspine[\tau]}{} < {\tree[\sigma]}{}}$ when $\dsubtype{\tau{} < \sigma{}}$ because the product type $\lspine[\alpha] \tensor \name{etree}$ is a subtype of $\tree[\beta] \tensor \tree[\beta]$ under $\tau/\alpha$ and $\sigma/\beta$.

\subsubsection{Object-oriented lists and trees}
On a related note, we could take a more object-oriented approach to lists and trees, using record types instead of eager products:
\begin{alignat*}{2}
  \olistn[\alpha] &\defd \with*{
                           \name{out}\colon \plus*{\none\colon \one ,\, \some\colon \alpha \tensor \with*{\fst\colon \olistn[\alpha] &&}}
                           ,\,
                           \name{size}\colon \nat}
  \\
  \otree[\beta] &\defd \with*{
                         \name{out}\colon \plus*{\none\colon \one ,\, \some\colon \beta \tensor \with*{\fst\colon \otree[\beta] ,\, \snd\colon \otree[\beta] &&}}
                           ,\,
                           \name{size}\colon \nat}
\end{alignat*}
Even purely structurally, $\dsubtype{{\olistn[\tau]}{} < {\otree[\sigma]}{}}$ does not hold when $\dsubtype{\tau{} < \sigma{}}$, but $\dsubtype{{\otree[\sigma]}{} < {\olistn[\tau]}{}}$ does when $\dsubtype{\sigma{} < \tau{}}$ and is in the parametric fragment.
This is somewhat counterintuitive, but nevertheless the correct relationship: any context that expects a list can use a tree's spine instead.

\subsubsection{Perfect binary trees}

Taking advantage of the support for nested types, we can adapt \citeauthor{Bird+Meertens:MPC98}'s prototypical example of perfect binary trees~\shortcite{Bird+Meertens:MPC98}.
\begin{equation*}
  \name{perfect}[\alpha] \defd \plus*{\leaf\colon \one ,\, \node\colon \alpha \tensor \name{perfect}[\alpha \tensor \alpha]}
\end{equation*}
Our algorithm confirms that $\name{perfect}[\alpha]$ is covariant in $\alpha$.
However, even purely structurally, $\name{perfect}[\tau]$ is \emph{not} a subtype of $\tree[\sigma]$ for any $\tau$ and $\sigma$.
That would require $\name{perfect}[\tau \tensor \tau]$ to be a subtype of $\tree[\sigma] \tensor \tree[\sigma]$, which cannot be: $\name{perfect}[\tau \tensor \tau]$ is a variant record type, whereas $\tree[\sigma] \tensor \tree[\sigma]$ is a product type.
Essentially, the difference amounts to one between breadth-first and depth-first representations of trees.
However, the lack of a subtyping relationship does not mean that the type $\name{perfect}[\alpha]$ is unusable: given the support for nested types, an explicit coercion from $\name{perfect}[\alpha]$ to $\tree[\alpha]$ could still be written.

\subsection{Serialized binary trees and spines}\label{sec:examples:serialized}

\subsubsection{Serialized binary trees}
Here we adapt an example from \citet{Thiemann+Vasconcelos:ICFP16} and consider it in the context of subtyping:
We may sometimes wish to serialize a binary tree to send it across the network or write it to a file.
A type that describes serialized trees is $\stree[\alpha,\kappa]$, parameterized by both the type of data elements, $\alpha$, and a suffix (or continuation) type, $\kappa$:
\begin{equation*}
  \stree[\alpha,\kappa] \defd \plus*{\leaf\colon \kappa ,\, \node\colon \alpha \tensor \stree[\alpha, \stree[\alpha, \kappa]]}
  \,.
\end{equation*}
According to this type, a serialized tree is a list of $\leaf$ and $\node$ labels.
A $\leaf$ is followed by a suffix of type $\kappa$;
a $\node$ is followed by the pair of the tree's root element of type $\alpha$ and the serialization of the left subtree, which itself is followed by the serialization of the right subtree.\footnote{Strictly speaking, this is not a true serialization due to the inclusion of a product type.
However, as there is no uniform way of serializing polymorphic data, this is as near to a true serialization as is possible.
Moreover, for concrete instances of $\tree$, such as with $\nat$ data, it is possible to give true serializations that inline the serialization of their data elements:
  $\name{snattree}[\kappa] \defd \plus*{\leaf\colon \kappa ,\, \node\colon \name{snat}[\name{snattree}[\name{snattree}[\kappa]]]}$, where $\name{snat}[-]$ is defined as in \cref{sec:parametric:declarative}.
}
This type crucially depends on nested types to express the invariant that $\stree[\alpha,\kappa]$ describes preorder traversals of binary trees.
Our saturation algorithm verifies that $\stree[\alpha,\kappa]$ is covariant in both $\alpha$ and $\kappa$.

Although unrelated to subtyping concerns, it is interesting to observe that the above type definition can, in fact, be \emph{derived} syntactically by repeatedly applying type isomorphisms to the definition $\stree[\alpha,\kappa] \defd \tree[\alpha] \tensor \kappa$:
\begin{align*}
  \stree[\alpha,\kappa]
    &\defd \tree[\alpha] \tensor \kappa
    = \plus*{\leaf\colon \one ,\, \node\colon \alpha \tensor (\tree[\alpha] \tensor \tree[\alpha])} \tensor \kappa \\
    &\simeq \plus*{\leaf\colon \one \tensor \kappa ,\, \node\colon \alpha \tensor (\tree[\alpha] \tensor (\tree[\alpha] \tensor \kappa))} \\
    &\simeq \plus*{\leaf\colon \kappa ,\, \node\colon \alpha \tensor (\tree[\alpha] \tensor \stree[\alpha,\kappa])} \\
    &\simeq \plus*{\leaf\colon \kappa ,\, \node\colon \alpha \tensor \stree[\alpha, \stree[\alpha,\kappa]]}
      \,.
\end{align*}

None of the isomorphic types in this sequence are mutual subtypes, however.
In particular, although the types $\stree[\alpha, \kappa]$ and $\tree[\alpha] \tensor \kappa$ are isomorphic, we do \emph{not} have $\tree[\tau]$ as a subtype of $\stree[\tau, \one]$, nor vice versa, for any $\tau$.
Comparing the $\leaf$ branches of both types, we see that, for these parametric subtyping relationships to hold, $\dsubtype{\one{} < \kappa{}}$ and $\dsubtype{\kappa{} < \one{}}$ must hold for \emph{all} types $\kappa$, regardless of the fact that $\stree[\tau, \one]$ ultimately instantiates $\kappa$ with $\one$.
This is simply not true when $\kappa$ is, for example, either $\plus*{}$ or $\one \tensor \one$.

Once again, the absence of subtyping relationships does not mean that the type $\stree[\alpha,\kappa]$ cannot be related to $\tree[\alpha] \tensor \kappa$.
Given a term language with support for nested types, it would still be possible to write explicit coercions between these types to serialize and deserialize trees.

\subsubsection{Serialized spines}
As an extension of this example, we can also define a type constructor $\name{sspine}[\alpha, \kappa]$ to describe serialized (left) spines:
\begin{alignat*}{2}
  \name{sspine}[\alpha,\kappa]
    &\defd {} & \lspine[\alpha] \tensor \kappa
    &\simeq \plus*{\leaf\colon \kappa ,\, \node\colon \alpha \tensor \name{sspine}[\alpha, \name{setree}[\kappa]]}
  \\
  \name{setree}[\kappa]
    &\defd {} & \name{etree} \tensor \kappa
    &\simeq \plus*{\leaf\colon \kappa}
  \,.
\end{alignat*}
The subtyping relationship between (left) spines and trees is preserved under serialization: we have the parametric rule $\dsubtype{{\name{sspine}[\alpha, \kappa]}{} < {\stree[\beta, \kappa']}{}}$ if both $\dsubtype{\alpha{} < \beta{}}$ and $\dsubtype{\kappa{} < {\kappa'}{}}$, as our algorithm confirms.
Notice that the inclusion of $\name{setree}[\kappa]$ and its $\plus*{\leaf\colon \kappa}$ is essential here.
Had we instead used the definition $\name{sspine}[\alpha,\kappa] \defd \plus*{\leaf\colon \kappa ,\, \node\colon \alpha \tensor \name{sspine}[\alpha, \kappa]}$, there would be no subtyping relationship because $\stree[\beta, \kappa']$ would have one more $\tensor$ than $\name{sspine}[\alpha,\kappa]$.

\subsection{Total functions and generalized tries on binary trees and spines}\label{sec:examples:tries}

\subsubsection{Total functions on trees}
We can use the following type to describe \emph{total} functions from $\alpha$~trees to $\beta$s.
As for serialized trees, this definition is derivable by repeatedly applying type isomorphisms to $\tree[\alpha] \imp \beta$.
(Interestingly, this type is, in a sense, dual to that of serialized trees.)
\begin{align*}
  \trie[\alpha,\beta]
    &\defd \tree[\alpha] \imp \beta
    = \plus*{\leaf\colon \one ,\, \node\colon \alpha \tensor (\tree[\alpha] \tensor \tree[\alpha])} \imp \beta \\
    &\simeq \with*{\leaf\colon \one \imp \beta ,\, \node\colon \alpha \tensor (\tree[\alpha] \tensor \tree[\alpha]) \imp \beta} \\
    &\simeq \with*{\leaf\colon \beta ,\, \node\colon \alpha \imp (\tree[\alpha] \imp (\tree[\alpha] \imp \beta))} \\
    &\simeq \with*{\leaf\colon \beta ,\, \node\colon \alpha \imp (\tree[\alpha] \imp \trie[\alpha,\beta])} \\
    &\simeq \with*{\leaf\colon \beta ,\, \node\colon \alpha \imp \trie[\alpha, \trie[\alpha,\beta]]}
\end{align*}
Thus, an object of type $\trie[\alpha,\beta]$ offers two methods, $\leaf$ and $\node$.
To look up the value of a leaf, the $\leaf$ method is invoked, resulting in the associated value of type $\beta$.
To look up a nonempty tree, the $\node$ method is invoked with the root element of type $\alpha$, resulting in an object of type $\trie[\alpha,\trie[\alpha,\beta]]$.
Recursively, the left subtree is looked up in this object; its associated value is an object of type $\trie[\alpha,\beta]$.
Then, the right subtree is looked up in this object, and the value of type $\beta$ associated with the entire nonempty tree is ultimately returned.

This example makes essential use of a record type to represent the object's methods, but most importantly, nested types are crucial to expressing the higher-order nature of lookups.
In the usual way, our algorithm confirms that $\trie[\alpha,\beta]$ is contravariant in $\alpha$ and covariant in $\beta$.

\subsubsection{Total functions on spines}

In a similar way, we can derive a type definition for total functions on (left) spines, using the types $\lspine[\alpha]$ and $\name{etree}$ defined in \cref{sec:examples:trees:spines}.
\begin{alignat*}{2}
  \name{spinefn}[\alpha, \beta]
    &\defd {} & \lspine[\alpha] \imp \beta
    & \simeq \with*{\leaf\colon \beta ,\,
                   \node\colon \alpha \imp \name{spinefn}[\alpha, \name{etreefn}[\beta]]}
  \\
  \name{etreefn}[\beta] &\defd {} & \name{etree} \imp \beta &\simeq \with*{\leaf\colon \beta}
\end{alignat*}
Once again, the subtyping relationship between (left) spines and trees is respected: we have the rule $\dsubtype{{\trie[\alpha, \beta]}{} < {\name{spinefn}[\alpha',\beta']}{}}$ if both $\dsubtype{{\alpha'}{} < {\alpha}{}}$ and $\dsubtype{\beta{} < {\beta'}{}}$, as our algorithm confirms.
(Notice that the subtyping direction is reversed because of contravariance in the underlying function types.)

\subsubsection{Tries for trees and spines}
In prior work~\cite{Wadsworth:EATCS79,Connelly+Morris:MSCS95,Hinze:JFP00}, the trie data structure for lists and strings was generalized to represent partial (not total) functions on more complex algebraic structures, such as binary trees.
A type definition for tries from keys of type $\tree[\alpha]$ to values of type $\beta$ can be derived from $\ptrie[\alpha,\beta] \defd \tree[\alpha] \imp \option[\beta]$, where $\option[\beta] \defd \plus*{\none\colon \one ,\, \some\colon \beta}$.
\begin{align*}
  \ptrie[\alpha,\beta]
    &\defd \tree[\alpha] \imp \option[\beta]
    = \plus*{\leaf\colon \one ,\, \node\colon \alpha \tensor (\tree[\alpha] \tensor \tree[\alpha])} \imp \option[\beta] \\
    &\simeq \with*{\leaf\colon \one \imp \option[\beta] ,\, \node\colon \alpha \imp (\tree[\alpha] \imp (\tree[\alpha] \imp \option[\beta]))} \\
    &\simeq \with*{\leaf\colon \option[\beta] ,\, \node\colon \alpha \imp (\tree[\alpha] \imp \ptrie[\alpha,\beta])} \\
    &\simeq \with*{\leaf\colon \option[\beta] ,\, \node\colon \alpha \imp \trie[\alpha, \ptrie[\alpha,\beta]]}
\end{align*}
For example, tries representing sets of $\tau$ trees could be typed as $\ptrie[\tau,\one]$.
Our algorithm verifies that $\ptrie[\alpha,\beta]$ is contravariant in $\alpha$ and covariant in $\beta$.
The use of nested types here is consistent with \citeauthor{Wadsworth:EATCS79}'s observation~\shortcite{Wadsworth:EATCS79}.
It would also be possible to similarly define a type $\name{spinetrie}[\alpha,\beta]$ of tries for (left) spines; it would be equivalent to $\name{spinefn}[\alpha,\option[\beta]]$.

\subsection{Refined stacks}\label{sec:examples:stacks:refined}

The additional expressive power of nested types allows us to also define a refined type for stacks that tracks the stack's shape:
\begin{equation*}
  \rstack[\alpha,\kappa] \defd \with*{\pushn\colon \alpha \imp \rstack[\alpha, \,\some[\alpha \tensor \rstack[\alpha, \kappa]]] ,\,
                                      \popn\colon \kappa}
  \,,
\end{equation*}
where $\some[\alpha] \defd \plus*{\some\colon \alpha}$.
Here, the type parameter $\kappa$ serves as a continuation to be used when popping from the stack.
Pushing an element onto the stack extends this continuation to reflect the existence (but not the identity) of the newly pushed element.

But we do not have $\rstack[\tau,\sigma]$ as a \emph{parametric} subtype of $\stack[\tau]$ for any $\tau$, even when we are guaranteed that $\dsubtype{\sigma{} < {\option[\tau \tensor \stack[\tau]]}{}}$.
Nevertheless, it is possible to prove that $\rstack[\tau,\sigma]$ is a \emph{structural} subtype of $\stack[\tau]$ when $\dsubtype{\sigma{} < {\option[\tau \tensor \stack[\tau]]}{}}$, making this one example of how the parametric fragment is a proper fragment of structural subtyping.

\subsection{Abstract types}

We can take advantage of the $\forall$ and $\exists$ quantifiers to express ML-style module signatures as abstract types.
Here we define two types that serve as signatures for abstract lists, with constructors $x$ and $\alpha \tensor x \imp x$ (respectively, $\beta \tensor x \imp x$) for the empty list and list concatenation, parameterized by the type $\alpha$ (respectively, $\beta$) of list elements.
Both signatures include a $\fold$ function, and $\name{alist'}[\beta]$ also includes a $\size$ function.
\begin{alignat*}{2}
  \name{alist}[\alpha]
    &\defd \exists x.\,
             x \tensor (\alpha \tensor x \imp x) \tensor
             \with*{\fold\colon \forall z.\, z \imp (\alpha \tensor z \imp z) \imp x \imp z&&}
  \\
  \name{alist'}[\beta]
    &\defd \exists x.\,
             x \tensor (\beta \tensor x \imp x) \tensor
             \with*{\fold\colon \forall z.\, z \imp (\beta \tensor z \imp z) \imp x \imp z ,\,
                    \size\colon x \imp \nat&&}
\end{alignat*}
Our algorithm confirms the signature subtyping relationship that derives from the additional $\size$ function: $\dsubtype{{\name{alist'}[\beta]}{} < {\name{alist}[\alpha]}{}}$ if both $\dsubtype{\alpha{} < \beta{}}$ and $\dsubtype{\beta{} < \alpha{}}$.

\section{Related work}\label{sec:related-work}

We now give a brief overview of the related work that has not already been discussed.

\paragraph{Subtyping recursive types}

The nature and complexity of the subtyping problem vary considerably depending on whether types are interpreted \emph{nominally} or \emph{structurally}; for type equivalence, the change from structural to nominal interpretation in \emph{non-regular} types leads to a decrease in complexity from doubly-exponential to linear~\cite{Mordido+:PLDI23}. If recursive types are nominal, the subtyping problem is simpler but also very limited. However, even under a nominal interpretation of how types are defined, issues arise when datasort refinements are considered~\cite{Davies:PhD05,Dunfield+Pfenning:POPL04,Freeman+Pfenning:PLDI91}. 

Although we have a broader setting, our interest in structural subtyping was influenced by session types~\cite{Honda+:ESOP98,Caires+Pfenning:CONCUR10}, where types are traditionally interpreted structurally, equirecursively, and coinductively. Subtyping in session types has been mostly treated coinductively~\cite{Gay+Hole:Acta05,Silva+:CONCUR23} and constitutes a particular case of our system. 

Developments toward structural subtyping began much earlier, regardless of whether types were treated coinductively~\cite{Amadio+Cardelli:TOPLAS93,Gay+Hole:Acta05,Hosoya+:98} or in a mixed inductive/coinductive setting~\cite{Brandt+Henglein:FI98,Danielsson+Altenkirch:MPC10,Lakhani+:ESOP22,Ligatti+:TOPLAS17}. In fact, the notion of structural subtyping dates back to 1988, introduced by \citet{Cardelli:POPL88}, but suggested even before \cite{Cardelli:ISSDT84,Cardelli:LITP85,Reynolds:TAPSOFT85}. In this paper, we have not explored the presence of empty or full types~\cite{Ligatti+:TOPLAS17,Lakhani+:ESOP22}; we leave this analysis for future work (more details are provided in Section~\ref{subsec:remarks-forward-search}).

\paragraph{Subtyping polymorphic types}

In this paper, we chose to have a foundational approach, including the features strictly necessary to handle parametric datatypes in programming languages. \citet{Bird+Meertens:MPC98} and \citet{Hinze:JFP00} noted that implementing generalized tries required the use of nested datatypes and non-regular recursion, which is our setting for this paper. Other type systems have been developed to explore \emph{non-regular} data structures. In the context of session types, \citet{Thiemann+Vasconcelos:ICFP16} proposed context-free session types with predicative polymorphism, extended later with impredicative polymorphism~\cite{Almeida+:IC22}; subtyping was also explored~\cite{Silva+:CONCUR23}. Nested session types were proposed by~\citet{Das+:TOPLAS22} and were proved to be more expressive than context-free session types~\cite{Das+:TOPLAS22,Gay+:FoSSaCS22}.

Inspired by the structural nature of types, these works have focused on structural subtyping and equivalence relations.
For session types, the subtyping problem has been shown to be undecidable~\cite{Padovani:TOPLAS19,Silva+:CONCUR23}, even though the corresponding type equality problems are decidable~\cite{Almeida+:TACAS20,Das+:TOPLAS22,Solomon:POPL78}.
In Section \ref{sec:structural:undecidability}, we present the result more generally for type systems with record types, explicitly identifying sets of minimal features that guarantee the undecidability of subtyping.
The undecidability of the subtyping relation leads to the design of incomplete algorithms. The \emph{parametric subtyping} relation we propose allows us to both tackle the incompleteness problem and to understand exactly to what extent the previous relations are incomplete, distinguishing cases where parametricity is not satisfied from cases where types actually exhibit distinct behaviors and are therefore not in (any) subtyping relation. Parametricity materializes the idea that types behave uniformly for all possible instantiations. This notion was first proposed by \citet{Reynolds:IFIP83} for System F~\cite{Girard:PhD72}, further explored by \citet{Wadler:FPCA89} and then extended to nested types by \citet{Johann+Ghiorzi:LMCS21}. None of these works focused on the subtyping relation. The combination of parametricity and subtyping is the main contribution of our work through type constructors that map (subtyping-)related arguments to (subtyping-)related results, in a relation that we called \emph{parametric subtyping}.

Several works focus on mixing subtyping with (explicit) parametric polymorphism via \emph{bounded quantification} \cite{Cardelli+Wegner:CSUR85}. The most standard formulation is the second-order lambda-calculus with bounded quantification, $F_\leq$ \cite{Cardelli+:IC94}, but subtyping was proved to be undecidable~\cite{Pierce:IC94}, even without recursion. Several $F_\leq$ fragments have been identified as having a decidable subtyping relation \cite{Cardelli+Wegner:CSUR85,Castagna+Pierce:POPL94,Katiyar+Sankar:WML92,Mackay+:APLAS20}, also extended with recursion \cite{Abadi+:POPL96,Zhou+:POPL23} or higher-order polymorphism and polarized application~\cite{Steffen:PhD99}. System $F_{\leq}$  was the ground for many developments in OOP~\cite{Rompf+Amin:OOPSLA16}.
As none of the applications we want for our type system seem to require bounded polymorphism at its most fundamental core, we limit our setting to parametric polymorphism and explicit quantifiers.

Other work studies the interaction of subtyping with type inference and \emph{implicit} polymorphism, such as that of \citet{Dolan+Mycroft:POPL17} and \citet{Lepigre+Raffalli:TOPLAS19}.
As previously mentioned, even without recursive types or type constructors, subtyping for implicit polymorphism is already undecidable~\cite{Wells:BU95,Tiuryn+Urzyczyn:IC02}, meaning that implicit polymorphism would not be a good starting point for our study of non-regular type constructors and parametricity.
In addition to this difference, \citet{Dolan+Mycroft:POPL17} focus on generative, regular type constructors, whereas our work deals with structural, non-regular type constructors.

\paragraph{Semantic \emph{vs} algorithmic subtyping definitions}

Semantic typing and subtyping are favored in non-regular structural type systems over their declarative versions. Initially, semantic relations were motivated by the set-theoretic properties of types~\cite{Castagna+Frisch:PPDP05,Frisch+:LICS02,Lakhani+:ESOP22}, but for non-regular types the need for a semantic relation to model the behavior of types was even more natural, by means of simulations and bisimulations \cite{Gay+Hole:Acta05,Das+:TOPLAS22,Silva+:CONCUR23}.

Algorithmic approaches for \emph{monomorphic} or \emph{regular} (sub)typing systems, such as standard fixed-point algorithms \cite{Gay+Hole:Acta05}, sequent calculus \cite{Das+:TOPLAS22}, cyclic proofs \cite{Lakhani+:ESOP22,Brotherston+Simpson:JLC10}, step indexing \cite{Ahmed:PhD04,Ahmed:ESOP06,Appel+McAllester:TOPLAS01,Dreyer+:LICS09,Lakhani+:ESOP22}, sized types \cite{Abel+Pientka:JFP16} or bouncing threads \cite{Baelde+:LICS22}, are effective. However, these mechanisms are not scalable when we navigate beyond regular types. In section \ref{sec:structural:examples}, we illustrate the narrow scope of the above approaches and the limitations of their application to nested types. To limit the recursion depth, a recursion bound is usually used, leading to incompleteness \cite{Das+:TOPLAS22} or even unsoundness\footnote{\label{ft:typescript}TypeScript has an \emph{unsound} structural subtyping: \url{https://www.typescriptlang.org/docs/handbook/type-compatibility.html}}. The undecidability of structural subtyping for the non-monomorphic fragment gives us no hope of finding a sound and complete algorithm. In this paper, we free our type system from this limitation by proposing the novel notion of \emph{parametric subtyping}, which takes advantage of parametricity without completely abandoning structural subtyping. In an attempt to overcome the limitations of alternatives such as bouncing threads or cyclic proofs, we end up finding a sweet spot that takes advantage of saturation-based methods to perform forward-inference, often used in constraint solving \cite{Jaffar+Lassez:POPL87} and unification~\cite{Martelli+Montanari:TOPLAS82,Huet:PhD76}.

\section{Conclusion}

In this paper, we presented a theory of parametricity for type constructors that forms the basis for parametric subtyping, a decidable, practical, and expressive fragment of structural subtyping for parametric polymorphism.
Moreover, the saturation-based decision procedure has led to an effective implementation that performs well on a variety of practical examples.

One opportunity for future work is to extend our results to a mixed inductive/coinductive type system, as in call-by-push-value~\cite{Levy:PhD01}, that accounts for subtypings that rely on types being uninhabited or full, such as $\dsubtype{{\plus*{}}{} < \sigma{}}$ and $\dsubtype{\tau{} < {\plus*{} \imp \sigma}{}}$ for all $\tau$ and $\sigma$.
We do already have a prototype implementation that does so, but the decision procedure's theory and accompanying correctness proofs appear to be more complicated.
Other opportunities include extension to bounded quantification, integration with intersection and union types, and development of a notion of parametric subtyping for higher-order kinds.
\begin{acks}
  We wish to thank the anonymous reviewers for their very valuable feedback on our paper.
  Support for this research was provided by the Funda\c{c}\~{a}o para a Ci\^{e}ncia e a Tecnologia through the project SafeSessions (PTDC/CCI-COM/6453/2020) and the LASIGE Research Unit (UIDB/00408/2020 and UIDP/00408/2020).
\end{acks}
\bibliographystyle{ACM-Reference-Format}
\bibliography{popl24}


\begin{thebibliography}{83}


\ifx \showCODEN    \undefined \def \showCODEN     #1{\unskip}     \fi
\ifx \showDOI      \undefined \def \showDOI       #1{#1}\fi
\ifx \showISBNx    \undefined \def \showISBNx     #1{\unskip}     \fi
\ifx \showISBNxiii \undefined \def \showISBNxiii  #1{\unskip}     \fi
\ifx \showISSN     \undefined \def \showISSN      #1{\unskip}     \fi
\ifx \showLCCN     \undefined \def \showLCCN      #1{\unskip}     \fi
\ifx \shownote     \undefined \def \shownote      #1{#1}          \fi
\ifx \showarticletitle \undefined \def \showarticletitle #1{#1}   \fi
\ifx \showURL      \undefined \def \showURL       {\relax}        \fi
\providecommand\bibfield[2]{#2}
\providecommand\bibinfo[2]{#2}
\providecommand\natexlab[1]{#1}
\providecommand\showeprint[2][]{arXiv:#2}

\bibitem[Abadi et~al\mbox{.}(1996)]%
        {Abadi+:POPL96}
\bibfield{author}{\bibinfo{person}{Martín Abadi}, \bibinfo{person}{Luca
  Cardelli}, {and} \bibinfo{person}{Ramesh Viswanathan}.}
  \bibinfo{year}{1996}\natexlab{}.
\newblock \showarticletitle{An Interpretation of Objects and Object Types}. In
  \bibinfo{booktitle}{\emph{Proceedings of the 23rd {ACM SIGPLAN-SIGACT}
  Symposium on Principles of Programming Languages}}.
  \bibinfo{pages}{396--409}.
\newblock
\urldef\tempurl%
\url{https://doi.org/10.1145/237721.237809}
\showDOI{\tempurl}


\bibitem[Abel and Pientka(2016)]%
        {Abel+Pientka:JFP16}
\bibfield{author}{\bibinfo{person}{Andreas Abel} {and}
  \bibinfo{person}{Brigitte Pientka}.} \bibinfo{year}{2016}\natexlab{}.
\newblock \showarticletitle{Well-Founded Recursion with Copatterns and Sized
  Types}.
\newblock \bibinfo{journal}{\emph{J. Funct. Program.}}  \bibinfo{volume}{26},
  Article \bibinfo{articleno}{e2} (\bibinfo{year}{2016}),
  \bibinfo{numpages}{61}~pages.
\newblock
\urldef\tempurl%
\url{https://doi.org/10.1017/S0956796816000022}
\showDOI{\tempurl}


\bibitem[Ahmed(2004)]%
        {Ahmed:PhD04}
\bibfield{author}{\bibinfo{person}{Amal~J. Ahmed}.}
  \bibinfo{year}{2004}\natexlab{}.
\newblock \emph{\bibinfo{title}{Semantics of Types for Mutable State}}.
\newblock \bibinfo{thesistype}{Ph.\,D. Dissertation}.
  \bibinfo{school}{Princeton University}.
\newblock
\urldef\tempurl%
\url{http://www.ccs.neu.edu/home/amal/ahmedsthesis.pdf}
\showURL{%
\tempurl}
\newblock
\shownote{AAI3136691}.


\bibitem[Ahmed(2006)]%
        {Ahmed:ESOP06}
\bibfield{author}{\bibinfo{person}{Amal~J. Ahmed}.}
  \bibinfo{year}{2006}\natexlab{}.
\newblock \showarticletitle{Step-Indexed Syntactic Logical Relations for
  Recursive and Quantified Types}. In \bibinfo{booktitle}{\emph{Programming
  Languages and Systems}}, Vol.~\bibinfo{volume}{3924}.
  \bibinfo{pages}{69--83}.
\newblock
\urldef\tempurl%
\url{https://doi.org/10.1007/11693024_6}
\showDOI{\tempurl}


\bibitem[Almeida et~al\mbox{.}(2022)]%
        {Almeida+:IC22}
\bibfield{author}{\bibinfo{person}{Bernardo Almeida}, \bibinfo{person}{Andreia
  Mordido}, \bibinfo{person}{Peter Thiemann}, {and} \bibinfo{person}{Vasco~T.
  Vasconcelos}.} \bibinfo{year}{2022}\natexlab{}.
\newblock \showarticletitle{Polymorphic Lambda Calculus with Context-Free
  Session Types}.
\newblock \bibinfo{journal}{\emph{Inform. Comput.}}  \bibinfo{volume}{289},
  Article \bibinfo{articleno}{104948} (\bibinfo{year}{2022}),
  \bibinfo{numpages}{36}~pages.
\newblock
\urldef\tempurl%
\url{https://doi.org/10.1016/j.ic.2022.104948}
\showDOI{\tempurl}


\bibitem[Almeida et~al\mbox{.}(2020)]%
        {Almeida+:TACAS20}
\bibfield{author}{\bibinfo{person}{Bernardo Almeida}, \bibinfo{person}{Andreia
  Mordido}, {and} \bibinfo{person}{Vasco~T. Vasconcelos}.}
  \bibinfo{year}{2020}\natexlab{}.
\newblock \showarticletitle{Deciding the Bisimilarity of Context-Free Session
  Types}. In \bibinfo{booktitle}{\emph{Tools and Algorithms for the
  Construction and Analysis of Systems}}, Vol.~\bibinfo{volume}{12079}.
  \bibinfo{pages}{39--56}.
\newblock
\urldef\tempurl%
\url{https://doi.org/10.1007/978-3-030-45237-7_3}
\showDOI{\tempurl}


\bibitem[Amadio and Cardelli(1993)]%
        {Amadio+Cardelli:TOPLAS93}
\bibfield{author}{\bibinfo{person}{Roberto~M. Amadio} {and}
  \bibinfo{person}{Luca Cardelli}.} \bibinfo{year}{1993}\natexlab{}.
\newblock \showarticletitle{Subtyping Recursive Types}.
\newblock \bibinfo{journal}{\emph{{ACM} Trans. Program. Lang. Syst.}}
  \bibinfo{volume}{15}, \bibinfo{number}{4} (\bibinfo{year}{1993}),
  \bibinfo{pages}{575--631}.
\newblock
\urldef\tempurl%
\url{https://doi.org/10.1145/155183.155231}
\showDOI{\tempurl}


\bibitem[Appel and McAllester(2001)]%
        {Appel+McAllester:TOPLAS01}
\bibfield{author}{\bibinfo{person}{Andrew~W. Appel} {and}
  \bibinfo{person}{David~A. McAllester}.} \bibinfo{year}{2001}\natexlab{}.
\newblock \showarticletitle{An Indexed Model of Recursive Types for
  Foundational Proof-Carrying Code}.
\newblock \bibinfo{journal}{\emph{{ACM} Trans. Program. Lang. Syst.}}
  \bibinfo{volume}{23}, \bibinfo{number}{5} (\bibinfo{year}{2001}),
  \bibinfo{pages}{657--683}.
\newblock
\urldef\tempurl%
\url{https://doi.org/10.1145/504709.504712}
\showDOI{\tempurl}


\bibitem[Baelde et~al\mbox{.}(2022)]%
        {Baelde+:LICS22}
\bibfield{author}{\bibinfo{person}{David Baelde}, \bibinfo{person}{Amina
  Doumane}, \bibinfo{person}{Denis Kuperberg}, {and} \bibinfo{person}{Alexis
  Saurin}.} \bibinfo{year}{2022}\natexlab{}.
\newblock \showarticletitle{Bouncing Threads for Circular and Non-Wellfounded
  Proofs}. In \bibinfo{booktitle}{\emph{Proceedings of the 37th Annual ACM/IEEE
  Symposium on Logic in Computer Science}}. Article \bibinfo{articleno}{63},
  \bibinfo{numpages}{13}~pages.
\newblock
\urldef\tempurl%
\url{https://doi.org/10.1145/3531130.3533375}
\showDOI{\tempurl}


\bibitem[Baeten et~al\mbox{.}(1993)]%
        {Baeten+:JACM93}
\bibfield{author}{\bibinfo{person}{J.~C.~M. Baeten}, \bibinfo{person}{J.~A.
  Bergstra}, {and} \bibinfo{person}{J.~W. Klop}.}
  \bibinfo{year}{1993}\natexlab{}.
\newblock \showarticletitle{Decidability of Bisimulation Equivalence for
  Processes Generating Context-Free Languages}.
\newblock \bibinfo{journal}{\emph{J. {ACM}}} \bibinfo{volume}{40},
  \bibinfo{number}{3} (\bibinfo{year}{1993}), \bibinfo{pages}{653--682}.
\newblock
\urldef\tempurl%
\url{https://doi.org/10.1145/174130.174141}
\showDOI{\tempurl}


\bibitem[Bergstra and Klop(1984)]%
        {Bergstra+Klop:IC84}
\bibfield{author}{\bibinfo{person}{J.~A. Bergstra} {and} \bibinfo{person}{J.~W.
  Klop}.} \bibinfo{year}{1984}\natexlab{}.
\newblock \showarticletitle{Process Algebra for Synchronous Communication}.
\newblock \bibinfo{journal}{\emph{Inform. Comput.}} \bibinfo{volume}{60},
  \bibinfo{number}{1} (\bibinfo{year}{1984}), \bibinfo{pages}{109--137}.
\newblock
\urldef\tempurl%
\url{https://doi.org/10.1016/S0019-9958(84)80025-X}
\showDOI{\tempurl}


\bibitem[Bird and Meertens(1998)]%
        {Bird+Meertens:MPC98}
\bibfield{author}{\bibinfo{person}{Richard Bird} {and} \bibinfo{person}{Lambert
  Meertens}.} \bibinfo{year}{1998}\natexlab{}.
\newblock \showarticletitle{Nested Datatypes}. In
  \bibinfo{booktitle}{\emph{Mathematics of Program Construction}},
  Vol.~\bibinfo{volume}{1422}. \bibinfo{pages}{52--67}.
\newblock
\urldef\tempurl%
\url{https://doi.org/10.1007/BFb0054285}
\showDOI{\tempurl}


\bibitem[Brandt and Henglein(1998)]%
        {Brandt+Henglein:FI98}
\bibfield{author}{\bibinfo{person}{Michael Brandt} {and} \bibinfo{person}{Fritz
  Henglein}.} \bibinfo{year}{1998}\natexlab{}.
\newblock \showarticletitle{Coinductive Axiomatization of Recursive Type
  Equality and Subtyping}.
\newblock \bibinfo{journal}{\emph{Fundamenta Informaticæ}}
  \bibinfo{volume}{33}, \bibinfo{number}{4} (\bibinfo{year}{1998}),
  \bibinfo{pages}{309--338}.
\newblock
\urldef\tempurl%
\url{https://doi.org/10.3233/FI-1998-33401}
\showDOI{\tempurl}


\bibitem[Brotherston and Simpson(2010)]%
        {Brotherston+Simpson:JLC10}
\bibfield{author}{\bibinfo{person}{James Brotherston} {and}
  \bibinfo{person}{Alex Simpson}.} \bibinfo{year}{2010}\natexlab{}.
\newblock \showarticletitle{Sequent Calculi for Induction and Infinite
  Descent}.
\newblock \bibinfo{journal}{\emph{J. Logic Comput.}} \bibinfo{volume}{21},
  \bibinfo{number}{6} (\bibinfo{year}{2010}), \bibinfo{pages}{1177--1216}.
\newblock
\urldef\tempurl%
\url{https://doi.org/10.1093/logcom/exq052}
\showDOI{\tempurl}


\bibitem[Caires and Pfenning(2010)]%
        {Caires+Pfenning:CONCUR10}
\bibfield{author}{\bibinfo{person}{Luís Caires} {and} \bibinfo{person}{Frank
  Pfenning}.} \bibinfo{year}{2010}\natexlab{}.
\newblock \showarticletitle{Session Types as Intuitionistic Linear
  Propositions}. In \bibinfo{booktitle}{\emph{Proceedings of the 21st
  International Conference on Concurrency Theory}},
  Vol.~\bibinfo{volume}{6269}. \bibinfo{pages}{222--236}.
\newblock
\urldef\tempurl%
\url{https://doi.org/10.1007/978-3-642-15375-4_16}
\showDOI{\tempurl}


\bibitem[Cardelli(1984)]%
        {Cardelli:ISSDT84}
\bibfield{author}{\bibinfo{person}{Luca Cardelli}.}
  \bibinfo{year}{1984}\natexlab{}.
\newblock \showarticletitle{A Semantics of Multiple Inheritance}. In
  \bibinfo{booktitle}{\emph{Semantics of Data Types}}. \bibinfo{pages}{51--67}.
\newblock
\urldef\tempurl%
\url{https://doi.org/10.1007/3-540-13346-1_2}
\showDOI{\tempurl}


\bibitem[Cardelli(1985)]%
        {Cardelli:LITP85}
\bibfield{author}{\bibinfo{person}{Luca Cardelli}.}
  \bibinfo{year}{1985}\natexlab{}.
\newblock \showarticletitle{Amber}. In \bibinfo{booktitle}{\emph{Combinators
  and Functional Programming Languages}}, Vol.~\bibinfo{volume}{242}.
  \bibinfo{pages}{21--47}.
\newblock
\urldef\tempurl%
\url{https://doi.org/10.1007/3-540-17184-3_38}
\showDOI{\tempurl}


\bibitem[Cardelli(1988)]%
        {Cardelli:POPL88}
\bibfield{author}{\bibinfo{person}{Luca Cardelli}.}
  \bibinfo{year}{1988}\natexlab{}.
\newblock \showarticletitle{Structural Subtyping and the Notion of Power Type}.
  In \bibinfo{booktitle}{\emph{Proceedings of the 15th ACM SIGPLAN-SIGACT
  Symposium on Principles of Programming Languages}}. \bibinfo{pages}{70--79}.
\newblock
\urldef\tempurl%
\url{https://doi.org/10.1145/73560.73566}
\showDOI{\tempurl}


\bibitem[Cardelli et~al\mbox{.}(1994)]%
        {Cardelli+:IC94}
\bibfield{author}{\bibinfo{person}{Luca Cardelli}, \bibinfo{person}{Simone
  Martini}, \bibinfo{person}{John~C. Mitchell}, {and} \bibinfo{person}{Andre
  Scedrov}.} \bibinfo{year}{1994}\natexlab{}.
\newblock \showarticletitle{An Extension of System {F} with Subtyping}.
\newblock \bibinfo{journal}{\emph{Inform. Comput.}} \bibinfo{volume}{109},
  \bibinfo{number}{1--2} (\bibinfo{year}{1994}), \bibinfo{pages}{4--56}.
\newblock
\urldef\tempurl%
\url{https://doi.org/10.1006/inco.1994.1013}
\showDOI{\tempurl}


\bibitem[Cardelli and Wegner(1985)]%
        {Cardelli+Wegner:CSUR85}
\bibfield{author}{\bibinfo{person}{Luca Cardelli} {and} \bibinfo{person}{Peter
  Wegner}.} \bibinfo{year}{1985}\natexlab{}.
\newblock \showarticletitle{On Understanding Types, Data Abstraction, and
  Polymorphism}.
\newblock \bibinfo{journal}{\emph{{ACM} Comput. Surv.}} \bibinfo{volume}{17},
  \bibinfo{number}{4} (\bibinfo{year}{1985}), \bibinfo{pages}{471--523}.
\newblock
\urldef\tempurl%
\url{https://doi.org/10.1145/6041.6042}
\showDOI{\tempurl}


\bibitem[Castagna and Frisch(2005)]%
        {Castagna+Frisch:PPDP05}
\bibfield{author}{\bibinfo{person}{Giuseppe Castagna} {and}
  \bibinfo{person}{Alain Frisch}.} \bibinfo{year}{2005}\natexlab{}.
\newblock \showarticletitle{A Gentle Introduction to Semantic Subtyping}. In
  \bibinfo{booktitle}{\emph{Proceedings of the 7th International {ACM}
  {SIGPLAN} Conference on Principles and Practice of Declarative Programming}}.
  \bibinfo{pages}{198--208}.
\newblock
\urldef\tempurl%
\url{https://doi.org/10.1145/1069774.1069793}
\showDOI{\tempurl}


\bibitem[Castagna and Pierce(1994)]%
        {Castagna+Pierce:POPL94}
\bibfield{author}{\bibinfo{person}{Giuseppe Castagna} {and}
  \bibinfo{person}{Benjamin~C. Pierce}.} \bibinfo{year}{1994}\natexlab{}.
\newblock \showarticletitle{Decidable Bounded Quantification}. In
  \bibinfo{booktitle}{\emph{Proceedings of the 21st ACM SIGPLAN-SIGACT
  Symposium on Principles of Programming Languages}}.
  \bibinfo{pages}{151--162}.
\newblock
\urldef\tempurl%
\url{https://doi.org/10.1145/174675.177844}
\showDOI{\tempurl}


\bibitem[Connelly and Morris(1995)]%
        {Connelly+Morris:MSCS95}
\bibfield{author}{\bibinfo{person}{Richard~H. Connelly} {and}
  \bibinfo{person}{F.~Lockwood Morris}.} \bibinfo{year}{1995}\natexlab{}.
\newblock \showarticletitle{A Generalization of the Trie Data Structure}.
\newblock \bibinfo{journal}{\emph{Math. Struct. Comp. Sci.}}
  \bibinfo{volume}{5}, \bibinfo{number}{3} (\bibinfo{year}{1995}),
  \bibinfo{pages}{381--418}.
\newblock
\urldef\tempurl%
\url{https://doi.org/10.1017/S0960129500000803}
\showDOI{\tempurl}


\bibitem[Danielsson and Altenkirch(2010)]%
        {Danielsson+Altenkirch:MPC10}
\bibfield{author}{\bibinfo{person}{Nils~Anders Danielsson} {and}
  \bibinfo{person}{Thorsten Altenkirch}.} \bibinfo{year}{2010}\natexlab{}.
\newblock \showarticletitle{Subtyping, Declaratively}. In
  \bibinfo{booktitle}{\emph{Mathematics of Program Construction}}.
  \bibinfo{pages}{100--118}.
\newblock
\urldef\tempurl%
\url{https://doi.org/10.1007/978-3-642-13321-3_8}
\showDOI{\tempurl}


\bibitem[Das et~al\mbox{.}(2021)]%
        {Das+:arXiv21}
\bibfield{author}{\bibinfo{person}{Ankush Das}, \bibinfo{person}{Henry
  DeYoung}, \bibinfo{person}{Andreia Mordido}, {and} \bibinfo{person}{Frank
  Pfenning}.} \bibinfo{year}{2021}\natexlab{}.
\newblock \bibinfo{title}{Subtyping on Nested Polymorphic Session Types}.
\newblock
\newblock
\showeprint[arXiv]{2103.15193}~[cs.PL]


\bibitem[Das et~al\mbox{.}(2022)]%
        {Das+:TOPLAS22}
\bibfield{author}{\bibinfo{person}{Ankush Das}, \bibinfo{person}{Henry
  DeYoung}, \bibinfo{person}{Andreia Mordido}, {and} \bibinfo{person}{Frank
  Pfenning}.} \bibinfo{year}{2022}\natexlab{}.
\newblock \showarticletitle{Nested Session Types}.
\newblock \bibinfo{journal}{\emph{{ACM} Trans. Program. Lang. Syst.}}
  \bibinfo{volume}{44}, \bibinfo{number}{3}, Article \bibinfo{articleno}{19}
  (\bibinfo{year}{2022}), \bibinfo{numpages}{45}~pages.
\newblock
\urldef\tempurl%
\url{https://doi.org/10.1145/3539656}
\showDOI{\tempurl}


\bibitem[Davies(2005)]%
        {Davies:PhD05}
\bibfield{author}{\bibinfo{person}{Rowan Davies}.}
  \bibinfo{year}{2005}\natexlab{}.
\newblock \emph{\bibinfo{title}{Practical Refinement-Type Checking}}.
\newblock \bibinfo{thesistype}{Ph.\,D. Dissertation}. \bibinfo{school}{Carnegie
  Mellon University}.
\newblock
\urldef\tempurl%
\url{http://reports-archive.adm.cs.cmu.edu/anon/2005/CMU-CS-05-110.pdf}
\showURL{%
\tempurl}


\bibitem[DeYoung et~al\mbox{.}(2023a)]%
        {DeYoung+:arXiv23}
\bibfield{author}{\bibinfo{person}{Henry DeYoung}, \bibinfo{person}{Andreia
  Mordido}, \bibinfo{person}{Frank Pfenning}, {and} \bibinfo{person}{Ankush
  Das}.} \bibinfo{year}{2023}\natexlab{a}.
\newblock \bibinfo{title}{Parametric Subtyping for Structural Parametric
  Polymorphism}.
\newblock
\newblock
\showeprint[arXiv]{2307.13661}~[cs.PL]


\bibitem[DeYoung et~al\mbox{.}(2023b)]%
        {DeYoung+:Zenodo23}
\bibfield{author}{\bibinfo{person}{Henry DeYoung}, \bibinfo{person}{Andreia
  Mordido}, \bibinfo{person}{Frank Pfenning}, {and} \bibinfo{person}{Ankush
  Das}.} \bibinfo{year}{2023}\natexlab{b}.
\newblock \bibinfo{title}{Parametric Subtyping for Structural Parametric
  Polymorphism ({Artifact})}.
\newblock
\newblock
\urldef\tempurl%
\url{https://zenodo.org/records/8423335}
\showURL{%
\tempurl}


\bibitem[DeYoung et~al\mbox{.}(2023c)]%
        {DeYoung+:Bitbucket23}
\bibfield{author}{\bibinfo{person}{Henry DeYoung}, \bibinfo{person}{Andreia
  Mordido}, \bibinfo{person}{Frank Pfenning}, {and} \bibinfo{person}{Ankush
  Das}.} \bibinfo{year}{2023}\natexlab{c}.
\newblock \bibinfo{title}{Standard {ML} Implementation of Parametric Subtyping
  Decision Procedure}.
\newblock
\newblock
\urldef\tempurl%
\url{https://bitbucket.org/structural-types/polyte}
\showURL{%
\tempurl}


\bibitem[Dolan and Mycroft(2017)]%
        {Dolan+Mycroft:POPL17}
\bibfield{author}{\bibinfo{person}{Stephen Dolan} {and} \bibinfo{person}{Alan
  Mycroft}.} \bibinfo{year}{2017}\natexlab{}.
\newblock \showarticletitle{Polymorphism, Subtyping, and Type Inference in
  {MLsub}}. In \bibinfo{booktitle}{\emph{Proceedings of the 44th ACM SIGPLAN
  Symposium on Principles of Programming Languages}}. \bibinfo{pages}{60--72}.
\newblock
\urldef\tempurl%
\url{https://doi.org/10.1145/3009837.3009882}
\showDOI{\tempurl}


\bibitem[Dreyer et~al\mbox{.}(2009)]%
        {Dreyer+:LICS09}
\bibfield{author}{\bibinfo{person}{Derek Dreyer}, \bibinfo{person}{Amal Ahmed},
  {and} \bibinfo{person}{Lars Birkedal}.} \bibinfo{year}{2009}\natexlab{}.
\newblock \showarticletitle{Logical Step-Indexed Logical Relations}. In
  \bibinfo{booktitle}{\emph{Proceedings of the 24th Annual {IEEE} Symposium on
  Logic in Computer Science}}. \bibinfo{pages}{71--80}.
\newblock
\urldef\tempurl%
\url{https://doi.org/10.1109/LICS.2009.34}
\showDOI{\tempurl}


\bibitem[Dunfield and Pfenning(2004)]%
        {Dunfield+Pfenning:POPL04}
\bibfield{author}{\bibinfo{person}{Jana Dunfield} {and} \bibinfo{person}{Frank
  Pfenning}.} \bibinfo{year}{2004}\natexlab{}.
\newblock \showarticletitle{Tridirectional Typechecking}. In
  \bibinfo{booktitle}{\emph{Proceedings of the 31st ACM SIGPLAN-SIGACT
  Symposium on Principles of Programming Languages}}.
  \bibinfo{pages}{281--292}.
\newblock
\urldef\tempurl%
\url{https://doi.org/10.1145/964001.964025}
\showDOI{\tempurl}


\bibitem[Freeman and Pfenning(1991)]%
        {Freeman+Pfenning:PLDI91}
\bibfield{author}{\bibinfo{person}{Tim Freeman} {and} \bibinfo{person}{Frank
  Pfenning}.} \bibinfo{year}{1991}\natexlab{}.
\newblock \showarticletitle{Refinement Types for {ML}}. In
  \bibinfo{booktitle}{\emph{Proceedings of the {ACM SIGPLAN} 1991 Conference on
  Language Design and Implementation}}. \bibinfo{pages}{268--277}.
\newblock
\urldef\tempurl%
\url{https://doi.org/10.1145/113445.113468}
\showDOI{\tempurl}


\bibitem[Friedman(1976)]%
        {Friedman:TCS76}
\bibfield{author}{\bibinfo{person}{Emily~P. Friedman}.}
  \bibinfo{year}{1976}\natexlab{}.
\newblock \showarticletitle{The Inclusion Problem for Simple Languages}.
\newblock \bibinfo{journal}{\emph{Theor. Comp. Sci.}} \bibinfo{volume}{1},
  \bibinfo{number}{4} (\bibinfo{year}{1976}), \bibinfo{pages}{297--316}.
\newblock
\urldef\tempurl%
\url{https://doi.org/10.1016/0304-3975(76)90074-8}
\showDOI{\tempurl}


\bibitem[Frisch et~al\mbox{.}(2002)]%
        {Frisch+:LICS02}
\bibfield{author}{\bibinfo{person}{Alain Frisch}, \bibinfo{person}{Giuseppe
  Castagna}, {and} \bibinfo{person}{Véronique Benzaken}.}
  \bibinfo{year}{2002}\natexlab{}.
\newblock \showarticletitle{Semantic Subtyping}. In
  \bibinfo{booktitle}{\emph{Proceedings of the 17th {IEEE} Symposium on Logic
  in Computer Science}}. \bibinfo{pages}{137--146}.
\newblock
\urldef\tempurl%
\url{https://doi.org/10.1109/LICS.2002.1029823}
\showDOI{\tempurl}


\bibitem[Gay and Hole(2005)]%
        {Gay+Hole:Acta05}
\bibfield{author}{\bibinfo{person}{Simon~J. Gay} {and} \bibinfo{person}{Malcolm
  Hole}.} \bibinfo{year}{2005}\natexlab{}.
\newblock \showarticletitle{Subtyping for session types in the pi calculus}.
\newblock \bibinfo{journal}{\emph{Acta Inform.}} \bibinfo{volume}{42},
  \bibinfo{number}{2--3} (\bibinfo{year}{2005}), \bibinfo{pages}{191--225}.
\newblock
\urldef\tempurl%
\url{https://doi.org/10.1007/s00236-005-0177-z}
\showDOI{\tempurl}


\bibitem[Gay et~al\mbox{.}(2022)]%
        {Gay+:FoSSaCS22}
\bibfield{author}{\bibinfo{person}{Simon~J. Gay}, \bibinfo{person}{Diogo
  Poças}, {and} \bibinfo{person}{Vasco~T. Vasconcelos}.}
  \bibinfo{year}{2022}\natexlab{}.
\newblock \showarticletitle{The Different Shades of Infinite Session Types}. In
  \bibinfo{booktitle}{\emph{Foundations of Software Science and Computation
  Structures}}, Vol.~\bibinfo{volume}{13242}. \bibinfo{pages}{347--367}.
\newblock
\urldef\tempurl%
\url{https://doi.org/10.1007/978-3-030-99253-8_18}
\showDOI{\tempurl}


\bibitem[Girard(1972)]%
        {Girard:PhD72}
\bibfield{author}{\bibinfo{person}{Jean-Yves Girard}.}
  \bibinfo{year}{1972}\natexlab{}.
\newblock \emph{\bibinfo{title}{Interprétation fonctionnelle et élimination
  des coupures de l'arithmétique d'ordre supérieur}}.
\newblock \bibinfo{thesistype}{Ph.\,D. Dissertation}. \bibinfo{school}{Éditeur
  inconnu}.
\newblock


\bibitem[Greenman et~al\mbox{.}(2014)]%
        {Greenman+:ACM14}
\bibfield{author}{\bibinfo{person}{Ben Greenman}, \bibinfo{person}{Fabian
  Muehlboeck}, {and} \bibinfo{person}{Ross Tate}.}
  \bibinfo{year}{2014}\natexlab{}.
\newblock \showarticletitle{Getting {F}-bounded Polymorphism into Shape}.
\newblock \bibinfo{journal}{\emph{ACM SIGPLAN Notices}} \bibinfo{volume}{49},
  \bibinfo{number}{6} (\bibinfo{year}{2014}), \bibinfo{pages}{89--99}.
\newblock
\urldef\tempurl%
\url{https://doi.org/10.1145/2666356.2594308}
\showDOI{\tempurl}


\bibitem[Greibach(1965)]%
        {Greibach:JACM65}
\bibfield{author}{\bibinfo{person}{Sheila~A. Greibach}.}
  \bibinfo{year}{1965}\natexlab{}.
\newblock \showarticletitle{A New Normal-Form Theorem for Context-Free Phrase
  Structure Grammars}.
\newblock \bibinfo{journal}{\emph{J. {ACM}}} \bibinfo{volume}{12},
  \bibinfo{number}{1} (\bibinfo{year}{1965}), \bibinfo{pages}{42--52}.
\newblock
\urldef\tempurl%
\url{https://doi.org/10.1145/321250.321254}
\showDOI{\tempurl}


\bibitem[Grigore(2017)]%
        {Grigore:ACM17}
\bibfield{author}{\bibinfo{person}{Radu Grigore}.}
  \bibinfo{year}{2017}\natexlab{}.
\newblock \showarticletitle{Java Generics Are Turing Complete}.
\newblock \bibinfo{journal}{\emph{ACM SIGPLAN Notices}} \bibinfo{volume}{52},
  \bibinfo{number}{1} (\bibinfo{year}{2017}), \bibinfo{pages}{73--85}.
\newblock
\urldef\tempurl%
\url{https://doi.org/10.1145/3009837.3009871}
\showDOI{\tempurl}


\bibitem[Groote and Hüttel(1994)]%
        {Groote+Huettel:IC94}
\bibfield{author}{\bibinfo{person}{Jan~Friso Groote} {and}
  \bibinfo{person}{Hans Hüttel}.} \bibinfo{year}{1994}\natexlab{}.
\newblock \showarticletitle{Undecidable Equivalences for Basic Process
  Algebra}.
\newblock \bibinfo{journal}{\emph{Inform. Comput.}} \bibinfo{volume}{115},
  \bibinfo{number}{2} (\bibinfo{year}{1994}), \bibinfo{pages}{354--371}.
\newblock
\urldef\tempurl%
\url{https://doi.org/10.1006/inco.1994.1101}
\showDOI{\tempurl}


\bibitem[Hinze(2000)]%
        {Hinze:JFP00}
\bibfield{author}{\bibinfo{person}{Ralf Hinze}.}
  \bibinfo{year}{2000}\natexlab{}.
\newblock \showarticletitle{Generalizing Generalized Tries}.
\newblock \bibinfo{journal}{\emph{J. Funct. Program.}} \bibinfo{volume}{10},
  \bibinfo{number}{4} (\bibinfo{year}{2000}), \bibinfo{pages}{327--351}.
\newblock
\urldef\tempurl%
\url{https://doi.org/10.1017/S0956796800003713}
\showDOI{\tempurl}


\bibitem[Honda et~al\mbox{.}(1998)]%
        {Honda+:ESOP98}
\bibfield{author}{\bibinfo{person}{Kohei Honda}, \bibinfo{person}{Vasco~T.
  Vasconcelos}, {and} \bibinfo{person}{Makoto Kubo}.}
  \bibinfo{year}{1998}\natexlab{}.
\newblock \showarticletitle{Language Primitives and Type Discipline for
  Structured Communication-Based Programming}. In
  \bibinfo{booktitle}{\emph{Programming Languages and Systems}},
  Vol.~\bibinfo{volume}{1381}. \bibinfo{pages}{122--138}.
\newblock
\urldef\tempurl%
\url{https://doi.org/10.1007/BFb0053567}
\showDOI{\tempurl}


\bibitem[Hosoya et~al\mbox{.}(1998)]%
        {Hosoya+:98}
\bibfield{author}{\bibinfo{person}{Haruo Hosoya}, \bibinfo{person}{Benjamin~C.
  Pierce}, {and} \bibinfo{person}{David~N. Turner}.}
  \bibinfo{year}{1998}\natexlab{}.
\newblock \bibinfo{title}{Datatypes and Subtyping}.  (\bibinfo{year}{1998}).
\newblock
\newblock
\shownote{Unpublished manuscript}.


\bibitem[Huet(1976)]%
        {Huet:PhD76}
\bibfield{author}{\bibinfo{person}{Gérard Huet}.}
  \bibinfo{year}{1976}\natexlab{}.
\newblock \emph{\bibinfo{title}{Resolution d'Equations dans des Langages
  d'Order 1, 2., $\omega$}}.
\newblock \bibinfo{thesistype}{Ph.\,D. Dissertation}.
  \bibinfo{school}{Universite de Paris VII}.
\newblock


\bibitem[Huet(1998)]%
        {Huet:MSCS98}
\bibfield{author}{\bibinfo{person}{Gérard Huet}.}
  \bibinfo{year}{1998}\natexlab{}.
\newblock \showarticletitle{Regular {Böhm} Trees}.
\newblock \bibinfo{journal}{\emph{Math. Struct. Comp. Sci.}}
  \bibinfo{volume}{8}, \bibinfo{number}{6} (\bibinfo{year}{1998}),
  \bibinfo{pages}{671--680}.
\newblock
\urldef\tempurl%
\url{https://doi.org/10.1017/S0960129598002643}
\showDOI{\tempurl}


\bibitem[Jaffar and Lassez(1987)]%
        {Jaffar+Lassez:POPL87}
\bibfield{author}{\bibinfo{person}{Joxan Jaffar} {and} \bibinfo{person}{J.-L.
  Lassez}.} \bibinfo{year}{1987}\natexlab{}.
\newblock \showarticletitle{Constraint Logic Programming}. In
  \bibinfo{booktitle}{\emph{Proceedings of the 14th ACM SIGACT-SIGPLAN
  Symposium on Principles of Programming Languages}}.
  \bibinfo{pages}{111--119}.
\newblock
\urldef\tempurl%
\url{https://doi.org/10.1145/41625.41635}
\showDOI{\tempurl}


\bibitem[Jančar(2021)]%
        {Jancar:JCSS21}
\bibfield{author}{\bibinfo{person}{Petr Jančar}.}
  \bibinfo{year}{2021}\natexlab{}.
\newblock \showarticletitle{Equivalence of Pushdown Automata via First-Order
  Grammars}.
\newblock \bibinfo{journal}{\emph{J. Comput. System Sci.}}
  \bibinfo{volume}{115} (\bibinfo{year}{2021}), \bibinfo{pages}{86--112}.
\newblock
\urldef\tempurl%
\url{https://doi.org/10.1016/j.jcss.2020.07.004}
\showDOI{\tempurl}


\bibitem[Johann and Ghiorzi(2021)]%
        {Johann+Ghiorzi:LMCS21}
\bibfield{author}{\bibinfo{person}{Patricia Johann} {and}
  \bibinfo{person}{Enrico Ghiorzi}.} \bibinfo{year}{2021}\natexlab{}.
\newblock \showarticletitle{Parametricity for Primitive Nested Types and
  GADTs}.
\newblock \bibinfo{journal}{\emph{Log. Meth. Comput. Sci.}}
  \bibinfo{volume}{17}, \bibinfo{number}{4}, Article \bibinfo{articleno}{23}
  (\bibinfo{year}{2021}), \bibinfo{numpages}{50}~pages.
\newblock
\urldef\tempurl%
\url{https://doi.org/10.46298/LMCS-17(4:23)2021}
\showDOI{\tempurl}


\bibitem[Katiyar and Sankar(1992)]%
        {Katiyar+Sankar:WML92}
\bibfield{author}{\bibinfo{person}{Dinesh Katiyar} {and}
  \bibinfo{person}{Sriram Sankar}.} \bibinfo{year}{1992}\natexlab{}.
\newblock \showarticletitle{Completely Bounded Quantification Is Decidable}. In
  \bibinfo{booktitle}{\emph{{ACM SIGPLAN} Workshop on {ML} and its
  Applications}}.
\newblock


\bibitem[Kennedy and Pierce(2007)]%
        {Kennedy+Pierce:06}
\bibfield{author}{\bibinfo{person}{Andrew~J. Kennedy} {and}
  \bibinfo{person}{Benjamin~C. Pierce}.} \bibinfo{year}{2007}\natexlab{}.
\newblock \showarticletitle{On Decidability of Nominal Subtyping with
  Variance}. In \bibinfo{booktitle}{\emph{FOOL-WOOD 2007}}.
\newblock
\urldef\tempurl%
\url{http://www.cis.upenn.edu/~bcpierce/papers/variance.pdf}
\showURL{%
\tempurl}


\bibitem[Korenjak and Hopcroft(1966)]%
        {Korenjak+Hopcroft:SWAT66}
\bibfield{author}{\bibinfo{person}{A.~J. Korenjak} {and} \bibinfo{person}{J.~E.
  Hopcroft}.} \bibinfo{year}{1966}\natexlab{}.
\newblock \showarticletitle{Simple Deterministic Languages}. In
  \bibinfo{booktitle}{\emph{7th Annual Symposium on Switching and Automata
  Theory}}. \bibinfo{pages}{36--46}.
\newblock
\urldef\tempurl%
\url{https://doi.org/10.1109/SWAT.1966.22}
\showDOI{\tempurl}


\bibitem[Lakhani et~al\mbox{.}(2022)]%
        {Lakhani+:ESOP22}
\bibfield{author}{\bibinfo{person}{Zeeshan Lakhani}, \bibinfo{person}{Ankush
  Das}, \bibinfo{person}{Henry DeYoung}, \bibinfo{person}{Andreia Mordido},
  {and} \bibinfo{person}{Frank Pfenning}.} \bibinfo{year}{2022}\natexlab{}.
\newblock \showarticletitle{Polarized Subtyping}. In
  \bibinfo{booktitle}{\emph{Programming Languages and Systems}},
  Vol.~\bibinfo{volume}{13240}. \bibinfo{pages}{431--461}.
\newblock
\urldef\tempurl%
\url{https://doi.org/10.1007/978-3-030-99336-8_16}
\showDOI{\tempurl}


\bibitem[Lepigre and Raffalli(2019)]%
        {Lepigre+Raffalli:TOPLAS19}
\bibfield{author}{\bibinfo{person}{Rodolphe Lepigre} {and}
  \bibinfo{person}{Christophe Raffalli}.} \bibinfo{year}{2019}\natexlab{}.
\newblock \showarticletitle{Practical Subtyping for Curry-Style Languages}.
\newblock \bibinfo{journal}{\emph{{ACM} Trans. Program. Lang. Syst.}}
  \bibinfo{volume}{41}, \bibinfo{number}{1}, Article \bibinfo{articleno}{5}
  (\bibinfo{year}{2019}), \bibinfo{numpages}{58}~pages.
\newblock
\urldef\tempurl%
\url{https://doi.org/10.1145/3285955}
\showDOI{\tempurl}


\bibitem[Levy(2001)]%
        {Levy:PhD01}
\bibfield{author}{\bibinfo{person}{Paul~Blain Levy}.}
  \bibinfo{year}{2001}\natexlab{}.
\newblock \emph{\bibinfo{title}{Call-By-Push-Value}}.
\newblock \bibinfo{thesistype}{Ph.\,D. Dissertation}.
  \bibinfo{school}{University of London}.
\newblock
\urldef\tempurl%
\url{https://www.cs.bham.ac.uk/~pbl/papers/thesisqmwphd.pdf}
\showURL{%
\tempurl}


\bibitem[Ligatti et~al\mbox{.}(2017)]%
        {Ligatti+:TOPLAS17}
\bibfield{author}{\bibinfo{person}{Jay Ligatti}, \bibinfo{person}{Jeremy
  Blackburn}, {and} \bibinfo{person}{Michael Nachtigal}.}
  \bibinfo{year}{2017}\natexlab{}.
\newblock \showarticletitle{On Subtyping-Relation Completeness, with an
  Application to Iso-Recursive Types}.
\newblock \bibinfo{journal}{\emph{{ACM} Trans. Program. Lang. Syst.}}
  \bibinfo{volume}{39}, \bibinfo{number}{4} (\bibinfo{year}{2017}),
  \bibinfo{pages}{4:1--4:36}.
\newblock
\urldef\tempurl%
\url{https://doi.org/10.1145/2994596}
\showDOI{\tempurl}


\bibitem[Mackay et~al\mbox{.}(2020)]%
        {Mackay+:APLAS20}
\bibfield{author}{\bibinfo{person}{Julian Mackay}, \bibinfo{person}{Alex
  Potanin}, \bibinfo{person}{Jonathan Aldrich}, {and} \bibinfo{person}{Lindsay
  Groves}.} \bibinfo{year}{2020}\natexlab{}.
\newblock \showarticletitle{Syntactically Restricting Bounded Polymorphism for
  Decidable Subtyping}.
\newblock In \bibinfo{booktitle}{\emph{18th Asian Symposium on Programming
  Languages and Systems}}. Vol.~\bibinfo{volume}{12470}.
  \bibinfo{pages}{125--144}.
\newblock
\urldef\tempurl%
\url{https://doi.org/10.1007/978-3-030-64437-6_7}
\showDOI{\tempurl}


\bibitem[Martelli and Montanari(1982)]%
        {Martelli+Montanari:TOPLAS82}
\bibfield{author}{\bibinfo{person}{Alberto Martelli} {and} \bibinfo{person}{Ugo
  Montanari}.} \bibinfo{year}{1982}\natexlab{}.
\newblock \showarticletitle{An Efficient Unification Algorithm}.
\newblock \bibinfo{journal}{\emph{{ACM} Trans. Program. Lang. Syst.}}
  \bibinfo{volume}{4}, \bibinfo{number}{2} (\bibinfo{year}{1982}),
  \bibinfo{pages}{258--282}.
\newblock
\urldef\tempurl%
\url{https://doi.org/10.1145/357162.357169}
\showDOI{\tempurl}


\bibitem[Mordido et~al\mbox{.}(2023)]%
        {Mordido+:PLDI23}
\bibfield{author}{\bibinfo{person}{Andreia Mordido}, \bibinfo{person}{Janek
  Spaderna}, \bibinfo{person}{Peter Thiemann}, {and} \bibinfo{person}{Vasco~T.
  Vasconcelos}.} \bibinfo{year}{2023}\natexlab{}.
\newblock \showarticletitle{Parameterized Algebraic Protocols}.
\newblock \bibinfo{journal}{\emph{Proc. {ACM} Program. Lang.}}
  \bibinfo{volume}{7}, \bibinfo{number}{PLDI}, Article \bibinfo{articleno}{163}
  (\bibinfo{year}{2023}), \bibinfo{numpages}{25}~pages.
\newblock
\urldef\tempurl%
\url{https://doi.org/10.1145/3591277}
\showDOI{\tempurl}


\bibitem[Mycroft(1984)]%
        {Mycroft:ISP84}
\bibfield{author}{\bibinfo{person}{Alan Mycroft}.}
  \bibinfo{year}{1984}\natexlab{}.
\newblock \showarticletitle{Polymorphic Type Schemes and Recursive
  Definitions}. In \bibinfo{booktitle}{\emph{International Symposium on
  Programming}}, Vol.~\bibinfo{volume}{167}. \bibinfo{pages}{217--228}.
\newblock
\urldef\tempurl%
\url{https://doi.org/10.1007/3-540-12925-1_41}
\showDOI{\tempurl}


\bibitem[Odersky and Läufer(1996)]%
        {Odersky+Laeufer:POPL96}
\bibfield{author}{\bibinfo{person}{Martin Odersky} {and}
  \bibinfo{person}{Konstantin Läufer}.} \bibinfo{year}{1996}\natexlab{}.
\newblock \showarticletitle{Putting Type Annotations to Work}. In
  \bibinfo{booktitle}{\emph{Proceedings of the 23rd ACM SIGPLAN-SIGACT
  Symposium on Principles of Programming Languages}}. \bibinfo{pages}{54--67}.
\newblock
\urldef\tempurl%
\url{https://doi.org/10.1145/237721.237729}
\showDOI{\tempurl}


\bibitem[Padovani(2019)]%
        {Padovani:TOPLAS19}
\bibfield{author}{\bibinfo{person}{Luca Padovani}.}
  \bibinfo{year}{2019}\natexlab{}.
\newblock \showarticletitle{Context-Free Session Type Inference}.
\newblock \bibinfo{journal}{\emph{{ACM} Trans. Program. Lang. Syst.}}
  \bibinfo{volume}{41}, \bibinfo{number}{2}, Article \bibinfo{articleno}{9}
  (\bibinfo{year}{2019}), \bibinfo{numpages}{37}~pages.
\newblock
\urldef\tempurl%
\url{https://doi.org/10.1145/3229062}
\showDOI{\tempurl}


\bibitem[Pierce(1994)]%
        {Pierce:IC94}
\bibfield{author}{\bibinfo{person}{Benjamin~C. Pierce}.}
  \bibinfo{year}{1994}\natexlab{}.
\newblock \showarticletitle{Bounded Quantification Is Undecidable}.
\newblock \bibinfo{journal}{\emph{Inform. Comput.}} \bibinfo{volume}{112},
  \bibinfo{number}{1} (\bibinfo{year}{1994}), \bibinfo{pages}{131--165}.
\newblock
\urldef\tempurl%
\url{https://doi.org/10.1006/inco.1994.1055}
\showDOI{\tempurl}


\bibitem[Pierce(2002)]%
        {Pierce:TAPL02}
\bibfield{author}{\bibinfo{person}{Benjamin~C. Pierce}.}
  \bibinfo{year}{2002}\natexlab{}.
\newblock \bibinfo{booktitle}{\emph{Types and Programming Languages}}.
\newblock \bibinfo{publisher}{The MIT Press}.
\newblock


\bibitem[Reynolds(1983)]%
        {Reynolds:IFIP83}
\bibfield{author}{\bibinfo{person}{John~C. Reynolds}.}
  \bibinfo{year}{1983}\natexlab{}.
\newblock \showarticletitle{Types, Abstraction and Parametric Polymorphism}. In
  \bibinfo{booktitle}{\emph{Information Processing 83, Proceedings of the
  {IFIP} 9th World Computer Congress}}. \bibinfo{pages}{513--523}.
\newblock


\bibitem[Reynolds(1985)]%
        {Reynolds:TAPSOFT85}
\bibfield{author}{\bibinfo{person}{John~C. Reynolds}.}
  \bibinfo{year}{1985}\natexlab{}.
\newblock \showarticletitle{Three Approaches to Type Structure}. In
  \bibinfo{booktitle}{\emph{Mathematical Foundations of Software Development}},
  Vol.~\bibinfo{volume}{185}. \bibinfo{pages}{97--138}.
\newblock
\urldef\tempurl%
\url{https://doi.org/10.1007/3-540-15198-2_7}
\showDOI{\tempurl}


\bibitem[Robinson(1965)]%
        {Robinson:JACM65}
\bibfield{author}{\bibinfo{person}{J.~A. Robinson}.}
  \bibinfo{year}{1965}\natexlab{}.
\newblock \showarticletitle{A Machine-Oriented Logic Based on the Resolution
  Principle}.
\newblock \bibinfo{journal}{\emph{J. {ACM}}} \bibinfo{volume}{12},
  \bibinfo{number}{1} (\bibinfo{year}{1965}), \bibinfo{pages}{23--41}.
\newblock
\urldef\tempurl%
\url{https://doi.org/10.1145/321250.321253}
\showDOI{\tempurl}


\bibitem[Rompf and Amin(2016)]%
        {Rompf+Amin:OOPSLA16}
\bibfield{author}{\bibinfo{person}{Tiark Rompf} {and} \bibinfo{person}{Nada
  Amin}.} \bibinfo{year}{2016}\natexlab{}.
\newblock \showarticletitle{Type Soundness for Dependent Object Types ({DOT})}.
\newblock \bibinfo{journal}{\emph{ACM SIGPLAN Notices}} \bibinfo{volume}{51},
  \bibinfo{number}{10} (\bibinfo{year}{2016}), \bibinfo{pages}{624--641}.
\newblock
\urldef\tempurl%
\url{https://doi.org/10.1145/3022671.2984008}
\showDOI{\tempurl}


\bibitem[Silva et~al\mbox{.}(2023)]%
        {Silva+:CONCUR23}
\bibfield{author}{\bibinfo{person}{Gil Silva}, \bibinfo{person}{Andreia
  Mordido}, {and} \bibinfo{person}{Vasco~T. Vasconcelos}.}
  \bibinfo{year}{2023}\natexlab{}.
\newblock \showarticletitle{Subtyping Context-Free Session Types}. In
  \bibinfo{booktitle}{\emph{34th International Conference on Concurrency
  Theory}}, Vol.~\bibinfo{volume}{279}. \bibinfo{pages}{11:1--11:19}.
\newblock
\urldef\tempurl%
\url{https://doi.org/10.4230/LIPIcs.CONCUR.2023.11}
\showDOI{\tempurl}


\bibitem[Skalka(1997)]%
        {Skalka:MSc97}
\bibfield{author}{\bibinfo{person}{Christian Skalka}.}
  \bibinfo{year}{1997}\natexlab{}.
\newblock \emph{\bibinfo{title}{Some Decision Problems for {ML} Refinement
  Types}}.
\newblock \bibinfo{thesistype}{Master's\ thesis}. \bibinfo{school}{Carnegie
  Mellon University}.
\newblock
\urldef\tempurl%
\url{http://ceskalka.w3.uvm.edu/skalka-pubs/skalka-ms-thesis.ps}
\showURL{%
\tempurl}


\bibitem[Solomon(1978)]%
        {Solomon:POPL78}
\bibfield{author}{\bibinfo{person}{Marvin~H. Solomon}.}
  \bibinfo{year}{1978}\natexlab{}.
\newblock \showarticletitle{Type Definitions with Parameters}. In
  \bibinfo{booktitle}{\emph{Proceedings of the 5th ACM SIGACT-SIGPLAN Symposium
  on Principles of Programming Languages}}. \bibinfo{pages}{31--38}.
\newblock
\urldef\tempurl%
\url{https://doi.org/10.1145/512760.512765}
\showDOI{\tempurl}


\bibitem[Steffen(1999)]%
        {Steffen:PhD99}
\bibfield{author}{\bibinfo{person}{Martin Steffen}.}
  \bibinfo{year}{1999}\natexlab{}.
\newblock \emph{\bibinfo{title}{Polarized Higher-Order Subtyping}}.
\newblock \bibinfo{thesistype}{Ph.\,D. Dissertation}.
  \bibinfo{school}{University of Erlangen-Nuremberg}.
\newblock
\urldef\tempurl%
\url{https://martinsteffen.github.io/assets/download/theses/diss/diss.pdf}
\showURL{%
\tempurl}


\bibitem[Stirling(2001a)]%
        {Stirling:TCS01}
\bibfield{author}{\bibinfo{person}{Colin Stirling}.}
  \bibinfo{year}{2001}\natexlab{a}.
\newblock \showarticletitle{Decidability of {DPDA} Equivalence}.
\newblock \bibinfo{journal}{\emph{Theor. Comp. Sci.}} \bibinfo{volume}{255},
  \bibinfo{number}{1} (\bibinfo{year}{2001}), \bibinfo{pages}{1--31}.
\newblock
\urldef\tempurl%
\url{https://doi.org/10.1016/S0304-3975(00)00389-3}
\showDOI{\tempurl}


\bibitem[Stirling(2001b)]%
        {Stirling:FSTTCS01}
\bibfield{author}{\bibinfo{person}{Colin Stirling}.}
  \bibinfo{year}{2001}\natexlab{b}.
\newblock \showarticletitle{An Introduction to Decidability of {DPDA}
  Equivalence}. In \bibinfo{booktitle}{\emph{FST TCS 2001: Foundations of
  Software Technology and Theoretical Computer Science}},
  Vol.~\bibinfo{volume}{2245}. \bibinfo{pages}{42--56}.
\newblock
\urldef\tempurl%
\url{https://doi.org/10.1007/3-540-45294-X_4}
\showDOI{\tempurl}


\bibitem[Sénizergues(2001)]%
        {Senizergues:TCS01}
\bibfield{author}{\bibinfo{person}{Géraud Sénizergues}.}
  \bibinfo{year}{2001}\natexlab{}.
\newblock \showarticletitle{{L(A)=L(B)?} {Decidability} Results from Complete
  Formal Systems}.
\newblock \bibinfo{journal}{\emph{Theor. Comp. Sci.}} \bibinfo{volume}{251},
  \bibinfo{number}{1} (\bibinfo{year}{2001}), \bibinfo{pages}{1--166}.
\newblock
\urldef\tempurl%
\url{https://doi.org/10.1016/S0304-3975(00)00285-1}
\showDOI{\tempurl}


\bibitem[Thiemann and Vasconcelos(2016)]%
        {Thiemann+Vasconcelos:ICFP16}
\bibfield{author}{\bibinfo{person}{Peter Thiemann} {and}
  \bibinfo{person}{Vasco~T. Vasconcelos}.} \bibinfo{year}{2016}\natexlab{}.
\newblock \showarticletitle{Context-Free Session Types}. In
  \bibinfo{booktitle}{\emph{Proceedings of the 21st ACM SIGPLAN International
  Conference on Functional Programming}}. \bibinfo{pages}{462--475}.
\newblock
\urldef\tempurl%
\url{https://doi.org/10.1145/2951913.2951926}
\showDOI{\tempurl}


\bibitem[Tiuryn and Urzyczyn(2002)]%
        {Tiuryn+Urzyczyn:IC02}
\bibfield{author}{\bibinfo{person}{Jerzy Tiuryn} {and} \bibinfo{person}{Paweł
  Urzyczyn}.} \bibinfo{year}{2002}\natexlab{}.
\newblock \showarticletitle{The Subtyping Problem for Second-Order Types Is
  Undecidable}.
\newblock \bibinfo{journal}{\emph{Inform. Comput.}} \bibinfo{volume}{179},
  \bibinfo{number}{1} (\bibinfo{year}{2002}), \bibinfo{pages}{1--18}.
\newblock
\urldef\tempurl%
\url{https://doi.org/10.1006/inco.2001.2950}
\showDOI{\tempurl}


\bibitem[Wadler(1989)]%
        {Wadler:FPCA89}
\bibfield{author}{\bibinfo{person}{Philip Wadler}.}
  \bibinfo{year}{1989}\natexlab{}.
\newblock \showarticletitle{Theorems for Free!}. In
  \bibinfo{booktitle}{\emph{Proceedings of the Fourth International Conference
  on Functional Programming Languages and Computer Architecture}}.
  \bibinfo{pages}{347--359}.
\newblock
\urldef\tempurl%
\url{https://doi.org/10.1145/99370.99404}
\showDOI{\tempurl}


\bibitem[Wadsworth(1979)]%
        {Wadsworth:EATCS79}
\bibfield{author}{\bibinfo{person}{C.~P. Wadsworth}.}
  \bibinfo{year}{1979}\natexlab{}.
\newblock \showarticletitle{Recursive Type Operators Which Are More Than Type
  Schemes}.
\newblock \bibinfo{journal}{\emph{Bulletin of the EATCS}}  \bibinfo{volume}{8}
  (\bibinfo{year}{1979}), \bibinfo{pages}{87--88}.
\newblock


\bibitem[Wells(1995)]%
        {Wells:BU95}
\bibfield{author}{\bibinfo{person}{Joe~B. Wells}.}
  \bibinfo{year}{1995}\natexlab{}.
\newblock \bibinfo{booktitle}{\emph{The Undecidability of {Mitchell}'s
  Subtyping Relationship}}.
\newblock \bibinfo{type}{{T}echnical {R}eport} 95-019.
  \bibinfo{institution}{Boston University}.
\newblock
\urldef\tempurl%
\url{http://www.cs.bu.edu/ftp/pub/jbw/types/subtyping-undecidable.ps.gz}
\showURL{%
\tempurl}


\bibitem[Zhou et~al\mbox{.}(2023)]%
        {Zhou+:POPL23}
\bibfield{author}{\bibinfo{person}{Litao Zhou}, \bibinfo{person}{Yaoda Zhou},
  {and} \bibinfo{person}{Bruno C.~d.~S. Oliveira}.}
  \bibinfo{year}{2023}\natexlab{}.
\newblock \showarticletitle{Recursive Subtyping for All}.
\newblock \bibinfo{journal}{\emph{Proc. {ACM} Program. Lang.}}
  \bibinfo{volume}{7}, \bibinfo{number}{POPL}, Article \bibinfo{articleno}{48}
  (\bibinfo{year}{2023}), \bibinfo{numpages}{30}~pages.
\newblock
\urldef\tempurl%
\url{https://doi.org/10.1145/3571241}
\showDOI{\tempurl}


\end{thebibliography}

\clearpage
\appendix

\section{Proofs}\label{sec:appendix:proofs}

\undecidablewith*
\begin{proof}
  We will prove each direction separately, beginning with the left-to-right direction.
  Its proof is by coinduction on the (potentially infinite) derivation of $\dsubtype{{(\bpa{q}{t})}{} < {(\bpa{p}{t})}{}}$.
  We distinguish cases on the structure of $p$.
  \begin{description}
  \item[Case:]
    Consider the case in which $p = \bpaseq{X}{p_2}$, where $X \defd \bpasum*[\ell \in L]{(\bpaseq{\ell}{p'_\ell})}$ for some nonempty label set $L$ and some processes $(p'_\ell)_{\ell \in L}$ and $p_2$.
    Observe that $p \trans[\ell] \bpaseq{p'_\ell}{p_2}$ for all $\ell \in L$.
    Because $L$ is nonempty and $p$ is simulated by $q$, it follows that $q$ can make at least one transition.
    Therefore, $q = \bpaseq{Y}{q_2}$, where $Y \defd \bpasum*[k \in K]{(\bpaseq{k}{q'_k})}$ for some nonempty $K$ and some processes $(q'_k)_{k \in K}$ and $q_2$.
    The derivation of $\dsubtype{{(\bpa{q}{t})}{} < {(\bpa{p}{t})}{}}$ may therefore begin with
    \begin{equation*}
      \infer{\dsubtype{{t_Y[\bpa{q_2}{t}]}{} < {t_X[\bpa{p_2}{t}]}{}}}{
        \infer[\with]{\dsubtype{{\with*[k \in K]{k\colon \bpa{q'_k}*{\bpa{q_2}{t}}}}{} < {\with*[\ell \in L]{\ell\colon \bpa{p'_\ell}*{\bpa{p_2}{t}}}}{}}}{
          (L \subseteq K) &
          \forall \ell \in L\colon 
            \dsubtype{{(\bpa{q'_\ell}*{\bpa{q_2}{t}})}{} < {(\bpa{p'_\ell}*{\bpa{p_2}{t}})}{}}}}
    \end{equation*}
    Now we must establish the premises of the above $\with$ rule.

    Because $p$ is simulated by $q$, it follows that $L \subseteq K$ and $\bpaseq{p'_\ell}{p_2} \simu \bpaseq{q'_\ell}{q_2}$ for all $\ell \in L$.
    Appealing to the coinductive hypothesis, we have $\dsubtype{{(\bpa*{\bpaseq{q'_\ell}{q_2}}{t})}{} < {(\bpa*{\bpaseq{p'_\ell}{p_2}}{t})}{}}$ for all $\ell \in L$.
    These appeals are justified because they are guarded by the $?$ and $\with$ rules.
    Because $(\bpa*{\bpaseq{q'_\ell}{q_2}}{t}) = (\bpa{q'_\ell}*{\bpa{q_2}{t}})$ and $(\bpa*{\bpaseq{q'_\ell}{q_2}}{t}) = (\bpa{q'_\ell}*{\bpa{q_2}{t}})$ for all $\ell \in L$, this completes the required derivation of $\dsubtype{{(\bpa{q}{t})}{} < {(\bpa{p}{t})}{}}$.

  \item[Case:] 
    Consider the case in which $p = \bpaemp$.
    In this case, $(\bpa{p}{t}) = t$ and we must show that $\dsubtype{{(\bpa{q}{t})}{} < t{}}$.
    If $q = \bpaseq{Y}{q_2}$ for some $Y \defd \bpasum*[k \in K]{(\bpaseq{k}{q'_k})}$, nonempty $K$, and processes $(q'_k)_{k \in K}$ and $q_2$, then the following is a derivation of $\dsubtype{{(\bpa{q}{t})}{} < t{}}$.
    \begin{equation*}
      \infer{\dsubtype{{t_Y[\bpa{q_2}{t}]}{} < t{}}}{
        \infer[\with]{\dsubtype{{\with*[k \in K]{k\colon \bpa{q'_k}*{\bpa{q_2}{t}}}}{} < {\with*{}}{}}}{}}
    \end{equation*}
    Otherwise, if $q = \bpaemp$, then the following is a derivation of $\dsubtype{{(\bpa{q}{t})}{} < t{}}$.
    \begin{equation*}
      \infer{\dsubtype{t{} < t{}}}{
        \infer[\with]{\dsubtype{{\with*{}}{} < {\with*{}}{}}}{}}
    \end{equation*}
  \end{description}

  Next, we will prove the right-to-left direction; the proof is by coinduction on the similarity $p \simu q$.

  It suffices to show that $p \trans[a] p'$ implies $q \trans[a]\umis p'$ when $\dsubtype{{(\bpa{q}{t})}{} < {(\bpa{p}{t})}{}}$.
  Assume that $p \trans[a] p'$.
  Because $p$ has a transition, there must exist an equation $X \defd \bpasum*[\ell \in L]{(\bpaseq{\ell}{p'_\ell})}$ and a process $p_2$ such that $p = \bpaseq{X}{p_2}$ and $p' = \bpaseq{p'_a}{p_2}$, with $a \in L$.
  We will distinguish cases on the structure of $q$.
  \begin{description}
  \item[Case:] 
    Consider the case in which $q = \bpaseq{Y}{q_2}$, where $Y \defd \bpasum*[k \in K]{(\bpaseq{k}{q'_k})}$ for a nonempty label set $K$.
    Notice that $q \trans[a] \bpaseq{q'_a}{q_2}$.
    By inversion on the (potentially infinite) derivation of $\dsubtype{{(\bpa{q}{t})}{} < {(\bpa{p}{t})}{}}$, we have $L \subseteq K$ and $\dsubtype{{(\bpa{q'_\ell}*{\bpa{q_2}{t}})}{} < {(\bpa{p'_\ell}*{\bpa{p_2}{t}})}{}}$ for all $\ell \in L$.
    In particular, $\dsubtype{{(\bpa*{\bpaseq{q'_a}{q_2}}{t}) = (\bpa{q'_a}*{\bpa{q_2}{t}})}{} < {(\bpa{p'_a}*{\bpa{p_2}{t}}) = (\bpa*{\bpaseq{p'_a}{p_2}}{t})}{}}$.
    Appealing to the coinductive hypothesis, $\bpaseq{q'_a}{q_2} \simu \bpaseq{p'_a}{p_2}$.
    This appeal is valid because it is guarded by ??.

  \item[Case:]
    Consider the case in which $q = \bpaemp$.
    We have $\dsubtype{t{} < {(\bpa{p}{t})}{}}$.
    By inversion, we have $p = \bpaemp$.
    Therefore, the result is vacuously true in this case.
  \qedhere
  \end{description}
\end{proof}

\parasound*
\begin{proof}
  Using the mixed induction and coinduction proof technique described by \citet{Danielsson+Altenkirch:MPC10}.
  Specifically, the proof is by lexicographic mixed induction and coinduction, first by coinduction on the (potentially infinite) structural subtyping derivation, and then by induction on the finite substitution stack $\subs$.
  \begin{description}
  \item[Case:] 
    Consider the case in which the parametric subtyping derivation of $\dsubtype{\tau{\subs} < {\sigma}{\subs*}}$ begins with:
    \begin{equation*}
      \infer[\prule{\jrule{INST-}}]{\dsubtype{{t[\theta]}{\subs} < {u[\phi]}{\subs*}}}{
        t[\alphas] \defd A & u[\betas] \defd B &
        \dsubtype{A{\theta ; \subs} < B{\phi ; \subs*}}}
    \end{equation*}
    Because $\applysubs{\subs}{t[\theta]} = t[\subs \circ \theta]$ and $\applysubs{\subs*}{u[\phi]} = u[\subs* \circ \phi]$, the structural subtyping derivation may begin with
    \begin{equation*}
      \infer[\srule{\jrule{UNF-}}]{\dsubtype{{t[\subs \circ \theta]}{} < {u[\subs* \circ \phi]}{}}}{
        \dsubtype{{(\subs \circ \theta)(A)}{} < {(\subs* \circ \phi)(B)}{}}}
    \end{equation*}
    To construct a derivation of $\dsubtype{{(\subs \circ \theta)(A)}{} < {(\subs* \circ \phi)(B)}{}}$, we may appeal to the coinductive hypothesis.
    Despite the substitution stack becoming larger, this appeal is valid because it is guarded by the above $\srule{\jrule{UNF-}}$ rule.

  \item[Case:]
    Consider the case in which the derivation of $\dsubtype{\tau{\subs} < {\sigma}{\subs*}}$ begins with:
    \begin{equation*}
      \infer[\prule{\jrule{PARAM-}}]{\dsubtype{\alpha{\theta ; \subs'} < \beta{\phi ; \subs*'}}}{
        \dsubtype{{\theta(\alpha)}{\subs'} < {\phi(\beta)}{\subs*'}}}
    \end{equation*}
    Because $\applysubs{(\theta ; \subs')}{\alpha} = \applysubs{\subs'}{\theta(\alpha)}$ and $\applysubs{(\phi ; \subs*)}{\beta} = \applysubs{\subs*'}{\phi(\beta)}$, we must construct a derivation of the structural subtyping $\dsubtype{{\applysubs{\subs'}{\theta(\alpha)}}{} < {\applysubs{\subs*'}{\phi(\beta)}}{}}$.
    We can do so by appealing to the inductive hypothesis at the smaller substitution stack $\subs'$.
  \qedhere
  \end{description}
\end{proof}

\polysoundlem*
\begin{proof}
  By mutual coinduction on the (potentially infinite) derivations of $\dsubtype{\tau{\subs} < {_{\xi'} \sigma}{\subs*}}$ and $\dsubtype{A{\subs} < {_{\xi'} \sigma}{\subs*}}$.
  We consider only the case of $\xi' = \cov$.
  The case of $\xi' = \ctv$ is symmetric.
  \begin{enumerate}
  \item
    Assume that
    $\asubtype[\xi]{t[\alphas] < u[\betas*]} \entails \asubtype[\xi']{\tau < \sigma}$;
    $\asubtype[\xi]{t[\alphas] < u[\betas*]} \nentails \bot$; and
    $\dsubtype{\alpha{\subs} < {_\zeta \beta}{\subs*}}$ for each $\asubtype[\xi]{t[\alphas] < u[\betas*]} \entails \asubtype[\zeta]{\alpha < \beta}$.
    We will now distinguish cases on the types $\tau$ and $\sigma$.
    \begin{description}
    \item[Case:] Consider the case in which $\tau = t'[\theta']$ and $\sigma = u'[\phi']$, where $t'[\alphas'] \defd A'$ and $u'[\betas*'] \defd B'$.
      The derivation of $\dsubtype{\tau{\subs} < \sigma{\subs*}}$ can therefore begin with
      \begin{equation*}
        \infer[\prule{\jrule{INST-}}]{\dsubtype{{t'[\theta']}{\subs} < {u'[\phi']}{\subs*}}}{
          \dsubtype{{A'}{\theta' ; \subs} < {B'}{\phi' ; \subs*}}}
      \end{equation*}
      To construct a derivation of $\dsubtype{{A'}{\theta' ; \subs} < {B'}{\phi' ; \subs*}}$, we will appeal to the coinductive hypothesis.
      This appeal is valid because it will be guarded by the above rule.
      To make the appeal, we need to first establish the three preconditions.
      \begin{itemize}
      \item
        Because $t'[\alphas'] \defd A'$ and $u'[\betas*'] \defd B'$, it follows from the $\arule{\jrule{INST-}}$ rule that the saturated database contains $\asubtype[\xi']{t'[\alphas'] < u'[\betas*']} \entails \asubtype[\xi']{A' < B'}$.

      \item
        Assume, for the sake of contradiction, that $\asubtype[\xi']{t'[\alphas'] < u'[\betas*']} \entails \bot$.
        Because $\asubtype[\xi]{t[\alphas] < u[\betas*]} \entails \asubtype[\xi']{t'[\theta'] < u'[\phi']}$, it follows from the $\arule{\jrule{COMPOSE-}}_\bot$ rule that the saturated database would also contain  $\asubtype[\xi]{t[\alphas] < u[\betas*]} \entails \bot$.
        This contradicts our earlier assumption, so we conclude that $\asubtype[\xi']{t'[\alphas'] < u'[\betas*']} \nentails \bot$.

      \item 
        Choose an arbitrary $\asubtype[\xi']{t'[\alphas'] < u'[\betas*']} \entails \asubtype[\zeta']{\alpha' < \beta'}$;
        we must show that $\dsubtype{{\alpha'}{\theta' ; \subs} < {_{\zeta'} \beta'}{\phi' ; \subs*}}$.
        This derivation can begin with 
        \begin{equation*}
          \infer[\prule{\jrule{PARAM-}}]{\dsubtype{{\alpha'}{\theta' ; \subs} < {_{\zeta'} \beta'}{\phi' ; \subs*}}}{
            \dsubtype{{\theta'(\alpha')}{\subs} < {_{\zeta'} \phi'(\beta')}{\subs*}}}
        \end{equation*}
        To construct a derivation of $\dsubtype{{\theta'(\alpha')}{\subs} < {_{\zeta'} \phi'(\beta')}{\subs*}}$, we will again appeal to the coinductive hypothesis.
        This appeal is valid because it will be guarded, by this rule as well as the above rule.
        The second and third preconditions follow immediately from our earlier assumptions. 
        The first precondition, that $\asubtype[\xi]{t[\alphas] < u[\betas*]} \entails \asubtype[\zeta']{\theta'(\alpha') < \phi'(\beta')}$, follows from the $\arule{\jrule{COMPOSE-}}$ rule, given that both $\asubtype[\xi]{t[\alphas] < u[\betas*]} \entails \asubtype[\xi']{t'[\theta'] < u'[\phi']}$ and $\asubtype[\xi']{t'[\alphas'] < u'[\betas*']} \entails \asubtype[\zeta']{\alpha' < \beta'}$.
      \end{itemize}
    \end{description}

  \item
    Assume that
    $\asubtype[\xi]{t[\alphas] < u[\betas*]} \entails \asubtype[\xi']{A < B}$;
    $\asubtype[\xi]{t[\alphas] < u[\betas*]} \nentails \bot$; and
    $\dsubtype{\alpha{\subs} < {_\zeta \beta}{\subs*}}$ for each $\asubtype[\xi]{t[\alphas] < u[\betas*]} \entails \asubtype[\zeta]{\alpha < \beta}$.
    We will now distinguish cases on the types $A$ and $B$.
    \begin{description}
    \item[Case:] Consider the case in which $A = \plus*[\ell \in L]{\ell\colon \tau_\ell}$ and $B = \plus*[k \in K]{k\colon \sigma_k}$ with $L \subseteq K$.
      The derivation of $\dsubtype{A{\subs} < B{\subs*}}$ can therefore begin with
      \begin{equation*}
        \infer[\prule{\plus}]{\dsubtype{{\plus*[\ell \in L]{\ell\colon \tau_\ell}}{\subs} < {\plus*[k \in K]{k\colon \sigma_k}}{\subs*}}}{
          \forall \ell \in L\colon
            \dsubtype{{\tau_\ell}{\subs} < {\sigma_\ell}{\subs*}}}
      \end{equation*}
      To construct a derivation of $\dsubtype{{\tau_\ell}{\subs} < {\sigma_\ell}{\subs*}}$ for each $\ell \in L$, we will appeal to the coinductive hypothesis.
      This appeal is valid because it will be guarded by the above rule.
      To make the appeal, we need to first establish the three preconditions.
      The second and third preconditions were already assumed.
      The first precondition, that $\asubtype[\xi]{t[\alphas] < u[\betas*]} \entails \asubtype[\xi']{\tau_\ell < \sigma_\ell}$, follows from the $\arule{\plus}$ rule, given that $\asubtype[\xi]{t[\alphas] < u[\betas*]} \entails \asubtype[\xi']{\plus*[\ell \in L]{\ell\colon \tau_\ell} < \plus*[k \in K]{k\colon \sigma_k}}$.
    \qedhere
    \end{description}    
  \end{enumerate}
\end{proof}

\polysoundthm*
\begin{proof}
  By structural induction on the finite derivation of $\asubtype[\xi]{\tau < \sigma}$.
  \begin{description}
  \item[Case:]
    Consider the case in which the derivation of $\asubtype[\xi]{\tau < \sigma}$ begins with:
    \begin{equation*}
      \infer[\brule{\jrule{COMPOSE-}}]{\asubtype[\xi]{t[\theta] < u[\phi]}}{
        \asubtype[\xi]{t[\alphas] < u[\betas]} \nentails \bot &
        \forall (\asubtype[\xi]{t[\alphas] < u[\betas]} \entails \asubtype[\zeta]{\alpha < \beta})\colon
          \asubtype[\zeta]{\theta(\alpha) < \phi(\beta)}}
    \end{equation*}
    where $t[\alphas] \defd A$ and $u[\betas*] \defd B$ for some structural types $A$ and $B$.
    Because $t[\alphas] \defd A$ and $u[\betas*] \defd B$, it follows that the saturated database contains $\asubtype[\xi]{t[\alphas] < u[\betas*]} \entails \asubtype[\xi]{A < B}$.
    By the inductive hypothesis on the subderivations and the $\prule{\jrule{PARAM-}}$ rule, we have $\dsubtype{\alpha{\theta ; \subs} < {_\zeta \beta}{\phi ; \subs*}}$ for each $\asubtype[\xi]{t[\alphas] < u[\betas*]} \entails \asubtype[\zeta]{\alpha < \beta}$.
    Therefore, we may appeal to \cref{lem:poly-sound} to deduce that $\dsubtype{A{\theta ; \subs} < {_\xi B}{\phi ; \subs*}}$.
    Applying the $\prule{\jrule{INST-}}$ rule, $\dsubtype{{t[\theta]}{\subs} < {_\xi u[\phi]}{\subs*}}$, as required.

  \item[Case:]
    Consider the case in which the derivation of $\asubtype[\xi]{\tau < \sigma}$ is:
    \begin{equation*}
      \infer[\brule{\jrule{VAR-}}]{\asubtype[\xi]{x < x}}{}
    \end{equation*}
    It follows immediately from the $\prule{\jrule{VAR-}}$ rule that $\dsubtype{x{\subs} < {_\xi x}{\subs*}}$.
  \qedhere
  \end{description}
\end{proof}

\polycompletelem*
\begin{proof}
  Each part is proved in sequence.
  \begin{enumerate}
  \item
  Assume that $\dsubtype{{t[\theta]}{\subs} < {_\xi u[\phi]}{\subs*}}$ and that $\asubtype[\xi]{t[\alphas] < u[\betas*]} \entails \asubtype[\xi']{A < B}$.
  Because $\asubtype[\xi']{A < B}$ is a consequence of $\asubtype[\xi]{t[\alphas] < u[\betas*]}$, inversion allows us to deduce that $t[\alphas] \defd A$, $u[\betas*] \defd B$, and $\xi' = \xi$.
  By inversion on the derivation of $\dsubtype{{t[\theta]}{\subs} < {_\xi u[\phi]}{\subs*}}$, we therefore have $\dsubtype{A{\theta ; \subs} < {_\xi B}{\phi ; \subs*}}$, just as obliged.

  \item
  By structural induction on the finite derivation of $\asubtype[\xi]{t[\alphas] < u[\betas*]} \entails \asubtype[\xi']{\tau < \sigma}$.
  \begin{description}
  \item[Case:] 
    \begin{equation*}
      \infer{\asubtype[\xi]{t[\alphas] < u[\betas]} \entails \asubtype[\zeta]{\theta'(\alpha') < \phi'(\beta')}}{
        \asubtype[\xi]{t[\alphas] < u[\betas]} \entails \asubtype[\xi']{t'[\theta'] < u'[\phi']} &
        \asubtype[\xi']{t'[\alphas'] < u'[\betas']} \entails \asubtype[\zeta]{\alpha' < \beta'}}
    \end{equation*}
    We must show that $\dsubtype{{\theta'(\alpha')}{\theta ; \subs} < {_\zeta \phi'(\beta')}{\phi ; \subs*}}$.

    Appealing to the inductive hypothesis on the left-hand subderivation, we may deduce $\dsubtype{{t'[\theta']}{\theta ; \subs} < {_{\xi'} u'[\phi']}{\phi ; \subs*}}$.
    With this in hand, we then appeal to the inductive hypothesis on the right-hand subderivation to deduce $\dsubtype{{\alpha'}{\theta' ; (\theta ; \subs)} < {_\zeta \beta'}{\phi' ; (\phi ; \subs*)}}$.
    By inversion, we have $\dsubtype{{\theta'(\alpha')}{\theta ; \subs} < {_\zeta \phi'(\beta')}{\phi ; \subs*}}$, just as required.

  \item[Case:]
    \begin{equation*}
      \infer{\asubtype[\xi]{t[\alphas] < u[\betas]} \entails \asubtype[\xi']{\tau_1 < \sigma_1}}{
        \asubtype[\xi]{t[\alphas] < u[\betas]} \entails \asubtype[\xi']{\tau_1 \tensor \tau_2 < \sigma_1 \tensor \sigma_2}}
    \end{equation*}
    We must show that $\dsubtype{{\tau_1}{\theta ; \subs} < {_{\xi'} \sigma_1}{\phi ; \subs*}}$.

    Applying part (1) of \cref{lem:poly-complete} to the subderivation, we may deduce $\dsubtype{{\tau_1 \tensor \tau_2}{\theta ; \subs} < {_{\xi'} \sigma_1 \tensor \sigma_2}{\phi ; \subs*}}$.
    By inversion, $\dsubtype{{\tau_1}{\theta ; \subs} < {_{\xi'} \sigma_1}{\phi ; \subs*}}$, as obliged.

  \item[Cases:] The other cases are similar.
  \end{description}

  \item
  We will prove that $\dsubtype{{t[\theta]}{\subs} < {_\xi u[\phi]}{\subs*}}$ and $\asubtype[\xi]{t[\alphas] < u[\betas*]} \entails \bot$ together imply a meta-contradiction, by induction on the finite derivation of $\asubtype[\xi]{t[\alphas] < u[\betas*]} \entails \bot$.
  \begin{description}
  \item[Case:] 
    \begin{equation*}
      \infer{\asubtype[\xi]{t[\alphas] < u[\betas]} \entails \bot}{
        \asubtype[\xi]{t[\alphas] < u[\betas]} \entails \asubtype[\xi']{t'[\theta'] < u'[\phi']} &
        \asubtype[\xi']{t'[\alphas'] < u'[\betas']} \entails \bot}
    \end{equation*}
    Applying part (2) of \cref{lem:poly-complete} to the left-hand subderivation, we may deduce $\dsubtype{{t'[\theta']}{\theta ; \subs} < {_{\xi'} u'[\phi']}{\phi ; \subs*}}$.
    Then, appealing to the inductive hypothesis on the right-hand subderivation, we have a meta-contradiction, as required.

  \item[Case:] 
    \begin{equation*}
      \infer{\asubtype[\xi]{t[\alphas] < u[\betas]} \entails \bot}{
        \asubtype[\xi]{t[\alphas] < u[\betas]} \entails \asubtype[\xi']{\alpha < \sigma} &
        (\sigma \neq \beta)}
    \end{equation*}
    Applying part (2) of \cref{lem:poly-complete} to the left-hand subderivation, we may deduce $\dsubtype{\alpha{\theta ; \subs} < {_{\xi'} \sigma}{\phi ; \subs*}}$.
    Inversion of this derivation yields a meta-contradiction, as required, because there is no declarative subtyping rule that will conclude this judgment when the type $\sigma$ is not a type parameter.

  \item[Case:] 
    \begin{equation*}
      \infer{\asubtype[\xi]{t[\alphas] < u[\betas]} \entails \bot}{
        \asubtype[\xi]{t[\alphas] < u[\betas]} \entails \asubtype[\xi']{\plus*[\ell \in L]{\ell\colon \tau_\ell} < \plus*[k \in K]{k\colon \sigma_k}} &
        (L \nsubseteq_{\xi'} K)}
    \end{equation*}
    Applying part (1) \cref{lem:poly-complete} to the subderivation, we may deduce $\dsubtype{{\plus*[\ell \in L]{\ell\colon \tau_\ell}}{\theta ; \subs} < {_{\xi'} \plus*[k \in K]{k\colon \sigma_k}}{\phi ; \subs*}}$.
    By inversion, we see that $L \subseteq_{\xi'} K$.
    This contradicts the above $L \nsubseteq_{\xi'} K$.

  \item[Cases:] 
    Each of the other cases is similar to one of the above.
  \qedhere
  \end{description}
  \end{enumerate}
\end{proof}

\polycomplete*
\begin{proof}
  By structural induction on the named type $\tau$.
  \begin{description}
  \item[Case:] 
    Consider the case in which $\tau = t[\theta]$.
    By inversion, the derivation of $\dsubtype{\tau{\subs} < \sigma{\subs*}}$ is one of $\dsubtype{{t[\theta]}{\subs} < {u[\phi]}{\subs*}} \vof \xi$.
    Therefore, the finite derivation of $\asubtype[\xi]{\tau < \sigma}$ can begin with the following rule, so long as we can derive its premises.
    \begin{equation*}
      \infer{\asubtype[\xi]{t[\theta] < u[\phi]}}{
          \asubtype[\xi]{t[\alphas] < u[\betas]} \nentails \bot
        &
          \forall (\asubtype[\xi]{t[\alphas] < u[\betas]} \entails \asubtype[\zeta]{\alpha < \beta})\colon
            \asubtype[\zeta]{\theta(\alpha) < \phi(\beta)}}
    \end{equation*}
    These premises are indeed derivable.
    \begin{itemize}
    \item
      The first premise is satisfied by an appeal to part (3) of \cref{lem:poly-complete}.

    \item
      Let $\asubtype[\xi]{t[\alphas] < u[\betas*]} \entails \asubtype[\zeta]{\alpha < \beta}$ be an arbitrary parameter consequence of $\asubtype[\xi]{t[\alphas] < u[\betas*]}$.
      Applying part (2) of \cref{lem:poly-complete}, we have $\dsubtype{\alpha{\theta ; \subs} < \beta{\phi ; \subs*}} \vof \zeta$.
      By inversion on this derivation, we may deduce that $\dsubtype{{\theta(\alpha)}{\subs} < {\phi(\beta)}{\subs*}} \vof \zeta$.
      Appealing to the inductive hypothesis at a smaller type $\theta(\alpha)$, we may deduce that $\asubtype[\zeta]{\theta(\alpha) < \phi(\beta)}$.
    \end{itemize}
    This fills in the premise of the above rule, completing the required derivation of $\asubtype[\xi]{\tau < \sigma}$.

  \item[Case:]
    Consider the case in which $\tau = x$.
    By inversion, the derivation of $\dsubtype{\tau{\subs} < \sigma{\subs*}} \vof \xi$ is:
    \begin{equation*}
      \infer[\prule{\jrule{VAR-}}]{\dsubtype{x{\subs} < {_\xi x}{\subs*}}}{}
    \end{equation*}
    The required $\asubtype[\xi]{x < x}$ follows immediately from $\brule{\jrule{VAR-}}$ rule.
  \qedhere
  \end{description}
\end{proof}

\end{document}